\documentclass[10pt,a4paper,dvipsnames]{article}

\usepackage{fullpage,amssymb,amsmath,amsfonts,amsthm,graphicx,color,float, mathtools,tikz} 
\usepackage{paralist} 
\usepackage[utf8]{inputenc}
\usepackage[english]{babel}
\usepackage[normalem]{ulem}
\usepackage{xcolor}
\usepackage[textsize=scriptsize, backgroundcolor=blue!20!white,bordercolor=red]{todonotes} 
\usepackage{dsfont} 
\usepackage{mathrsfs}
\usepackage{graphicx}
\usepackage{amsmath}
\usepackage{amsthm}
\usepackage{amsfonts}
\usepackage{amssymb}
\usepackage{multicol}
\usepackage{fancyhdr}
\usepackage{mathtools}

\usepackage{cite}

\usepackage{hyperref} 

\newtheorem{theorem}{Theorem}[section]
\newtheorem{Lemma}[theorem]{Lemma}

\newtheorem{corollary}[theorem]{Corollary}

\theoremstyle{definition}
\newtheorem{definition}[theorem]{Definition}
\newtheorem*{remark*}{Remark}

\newcommand{\R}{\mathbb R}

\newcommand{\N}{\mathbb N}

\newcommand{\norm}[1]{\lVert#1\rVert}

\numberwithin{equation}{section}
\allowdisplaybreaks 

\title{On the mean-field limit for the Vlasov-Poisson system }
\author{Manuela Feistl-Held\thanks{Fachbereich ANG, Technische Hochschule Rosenheim, Hochschulstraße 1, 83024 Rosenheim, Germany. Email: Manuela.Feistl-Held@th-rosenheim.de},  
	Peter Pickl\thanks{Fachbereich Mathematik, Eberhard Karls Universität Tübingen, Auf der Morgenstelle 10, 72076 Tübingen, Germany. Email: p.pickl@uni-tuebingen.de}}

\begin{document}
\maketitle

\begin{abstract}
We present a probabilistic proof of the mean-field limit and propagation of chaos of a classical N-particle system in three dimensions with Coulomb interaction force of the form $f^N(q)=\pm\frac{q}{|q|^3}$ and $N$-dependent cut-off at $|q|>N^{-\frac{5}{12}+\sigma}$ where $\sigma>0$ can be chosen arbitrarily small. This cut-off size is much smaller than the  typical distance to the nearest neighbour. 
In particular, for typical initial data, we show convergence of the Newtonian trajectories to the characteristics of the Vlasov-Poisson system. The proof is based on a Gronwall estimate for the maximal distance between the exact microscopic dynamics and the approximate mean-field dynamics. Thus our result leads to a derivation of the Vlasov-Poisson equation from the microscopic $N$-particle dynamics with force term arbitrary close to the physically relevant Coulomb force.
\end{abstract}

\section{Introduction} \label{sec:intro}

We are interested in a microscopic derivation of the 
Vlasov-Poisson system, which describes a plasma of identically
charged particles with electrostatic or gravitational interactions.
Therefore we consider a system consisting of $N$ interacting particles subject to Newtonian time evolution. Our system is distributed by a trajectory in phase space $\R^{6N}$ with $X=(Q,P)=(q_1,\hdots,q_N,p_1,\hdots,p_N) \in \R^{6N}$, where $(Q)_j= q_j\in \R^{3}$ denotes the one-particle position and and $(P)_j=p_j$ stands for its momentum.
The evolution of the system is given by the coupled differential equations
\begin{align} 
i\in \{1,...,N\},\ \begin{cases}
& {\dot{q}}_i=\frac{p_i}{m}\\
& \dot{p}_i=\frac{1}{N}\sum_{j\neq i}f^N(q_i-q_j) 
\end{cases} \label{Newtoneq}
\end{align}
with particle mass $m>0$, which will always be set equal to $1$ in our considerations.
We consider a Coulomb force with a cut-off at $N^{-\beta}$ for $\beta\leq \frac{5}{12}-\sigma$ and arbitrary $\sigma>0$. Remarkably this cut-off can be chosen distinctly smaller than the  typical distance to the nearest neighbour  which is given by $N^{-\frac{1}{3}}$.

\begin{definition}\label{force f}
For $N\in \N\cup\lbrace \infty \rbrace$ the interaction force is given by 
\begin{align*}
f^{N}:\R^{3}\rightarrow\R^{3}, q\mapsto\begin{cases}
aN^{3\beta}q& \text{if}\  |q|\leq N^{-\beta} \\
a \frac{q}{|q|^{3}} & \text{if}\  |q|> N^{-\beta} 
\end{cases}
\end{align*}
for $0<\beta\leq\frac{5}{12}-\sigma,$ some positive $\sigma$ and  $a\in\{\pm 1\}$, distinguishing between attractive and repulsive interactions.
\end{definition}

\begin{remark*}
 Our results can easily be generalized to any $a\in\mathbb{R}$ by simply rescaling the system accordingly.
\end{remark*}	

 We  will use the notation $F^N:\R^{6N}\rightarrow\R^{3N}$ for the total force of the system. Thus the i'th compontent of $F^N$ gives the force exhibited on a single coordinate $j$: $$(F^N(X))_j\coloneqq  \sum_{i\neq j}\frac{1}{N}f^N(q_i-q_j).$$

We consider the system in the mean-field scaling, so that the total mass of the system remains of order $1$. The prefactor $\frac{1}{N}$ constitutes such a scaling factor and seems to be the most common choice in this setting \cite{Jabin lecture notes}. 
Accompanying this, we rescale time, position, and momentum. To preserve the initial data $ X = (q_i(0), p_i(0))$ for $1 \leq i \leq N$ at an order of $1$, we define $p_i = N^{1/2} \cdot \bar{p}_i$ and $t_i = N^{-1/2} \cdot \bar{t}_i$.
The Virial theorem  further justifies this particular scaling. For homogeneous $f^N$, the long-time averages of the total kinetic energy and the potential energy are of the same order.

Our goal is to derive the Vlasov-Poisson equation, which describes the time evolution of a plasma consisting of charged particles with gravitational or electrostatic force, from the microscopic Newtonian $N$-particle dynamics. 
It reads as follows
\begin{align}
    \begin{cases}
        \partial_t k + p \cdot \nabla_q k +  \nabla_p k\cdot f\ast\tilde{k}_t = 0,\\
        \tilde{k}_t(q) = \int_{\mathbb{R}^3} k_t( q, p) \, dp,
    \end{cases}
\end{align}
with interaction force $f:\R^3\rightarrow\R^3$ and initial density $k_0(q, p)$.
The term  $p \cdot \nabla_q k$ denotes the free transport term and $\nabla_p k\cdot f\ast\tilde{k}_t $ the acceleration term with the mean-force $f\ast\tilde{k}_t$.

\subsection{Previous results}
While the existence theory of the Vlasov-Poisson dynamics is well understood its microscopic derivation from systems without cut-off is still an open problem. To our knowledge, the first paper
to discuss a mathematically rigorous derivation of Vlasov equations is Neunzert and Wick in 1974 \cite{NeunzertWick}.
Classical results of this kind are valid for Lipschitz-bounded forces \cite{BraunHepp,Dobrushin}. 
One major difference to our work is that the results rely on deterministic initial conditions even if some of them are formulated probabilistically.
Handling clustering of particles for singular interactions (see \cite{SpohnBook}) like Coulomb or Newtons gravitational force brings further difficulties.
Hauray and Jabin could include singular interaction forces scaling like $1/\vert q \vert^{\lambda}$ in three dimensions with ${\lambda} < 1$ \cite{Hauray2007} and later as well the physically more interesting case with ${\lambda}$ smaller but close to $2$, and a lower bound on the cut-off at $q = N^{-1/6}$ \cite{Hauray2013}. They had to choose quite specific initial conditions, according to the respective $N$-particle law.
The last deterministic result we like to mention here is\cite{Kiessling}, which is valid for repulsive pair-interactions and assumes no cut-off but instead a bound on the maximal forces of the microscopic system.
Assuming monokineticity, Serfaty and Duerinckx proved the validity of the mean-field description - in that case the pressureless Euler-Poisson system - for Coulomb-interactions without cutoff \cite{Serfaty}.

In contrast to the previous approaches Boers and Pickl \cite{BoersPickl} derive the Vlasov equations for stochastic initial conditions with interaction forces scaling like $|x|^{-3\lambda+1}$ with $(5/6<\lambda<1)$. 
They obtained a cut-off as small as the typical inter particle distance at $ N^{-\frac{1}{3}}$.
By exploiting the second order nature of the dynamics and introducing anisotropic scaling of the relevant metric to include the Coulomb singularity Lazarovici and Pickl \cite{Dustin} extended the method in \cite{BoersPickl} and obtained a microscopic derivation of the Vlasov-Poisson equation with a cut-off of $N^{-\delta}$ with $0<\delta<\frac{1}{3}$.
More recently, by examining the collisions which could occur and using the second order nature of the dynamics, the cut-off parameter was reduced to as small as $N^{-\frac{7}{18}+\sigma}$, with $\sigma > 0$ in \cite{grass}.

In this paper we provide  a derivation of Vlasov-Poisson equation with the so far weakest condition on the cutoff. This equation is a classical example of an effective equation approximating the  time evolution of a $N$-particle system with Coulomb or Newtonian pair interaction in the large $N$ limit.
Specifically, this interaction is given by 
$
f^N(q)=\pm\frac{q}{|q|^3} \ \text{for} \ |q|>N^{-\frac{5}{12}+\sigma}
$
with cut-off at $|q|= N^{-\frac{5}{12}+\sigma}$ for arbitrarily small $\sigma>0$. The cut-off diameter is of smaller order than the average distance of a particle to its nearest neighbour and has been significantly improved compared to the results of Grass and Pickl \cite{grass}.
The underlying poof technique hints that a further improvement is possible by utilizing a finer subdivision of particle subsets.

\subsection{Dynamics of the Newtonian and of the effective system}
 Having introduced the N-particle force in Definition \ref{force f} we wish to define the respective  Newtonian flow.
As the vector field is Lipschitz for fixed $N$ we have global existence and uniqueness of solutions for \eqref{Newton}, therefore a well defined flow. 
\begin{definition}\label{Def:Newtonflow}
The Newtonian flow $\Psi_{t,s}^{N}(X)=(\Psi_{t,s}^{1,N}(X)),\Psi_{t,s}^{2,N}(X))$ on $\R^{6N}$ is defined by the solution of
\begin{align}
\frac{d}{dt} \Psi_{t,s}^{N}(X)=(\Psi_{t,s}^{2,N}(X),F(\Psi_{t,s}^{1,N}(X))) \in \R^{3N}\times\mathbb{R}^{3N}
\end{align}
with $\Psi_{s,s}^{N}(X)=X$. \end{definition}

The first $3N$ components of $\Psi_{t,s}^{N}(X)$  describe the positions of the particles at time $t$, given the configuration of all particles at time $s$ was $X$. The other $3N$ components decribe the velocities of the particles respectively.

Due to the symmetry of the respective force   $\Psi^N_{t,s\in\R}:\R^{6N}\to \R^{6N}$,    is     symmetric under permutation of coordinates. 
Looking for a macroscopic law of motion for the particle density leads us to a continuity equation of Vlasov type.
 For $N \in \N\cup\lbrace \infty \rbrace$, and $k:\R^6\rightarrow\R_0^{+}$ we consider the corresponding mean-field equation, namely the Vlasov-Poisson equation
\begin{equation}\label{Vlasov}
	 \begin{cases}
        \partial_t k + p \cdot \nabla_q k +  \nabla_p k\cdot f\ast\tilde{k}_t = 0,\\
        \tilde{k}_t(q) = \int_{\mathbb{R}^3} k_t( q, p) \, dp,
    \end{cases}
\end{equation}
This equation describes a plasma of identically charged particles with electrostatic interactions or a gravitational system (galaxy).
For a fixed initial distribution $k_0 \in L^\infty(\R^3\times \R^3)$ with $k_0 \geq 0$ we denote by $k^N_t$ the unique solution of  \eqref{Vlasov} with initial datum $k^N_t(0, \cdot,\cdot) = k_0$.\\
The global existence and uniqueness of solutions of this equation for suitable initial conditions is well understood, even for singular interactions (see \cite{Pfaffelmoser}, \cite{schaeffer1991} \cite{Horst} and \cite{Lions}).
For our purposes, a result established by Horst \cite{Horst} is sufficient, as it provides global existence of classical solutions (uniquely) under conditions that closely align with the assumptions required for the proof of our Theorem \ref{maintheorem} in Section \ref{Chapter: VP}. 
 For repulsive Coulomb interactions he  specifically shows that there is a continuously differentiable function $k:~[0, T] \times \mathbb{R}^6 \to [0, \infty)$ for any $T > 0$ that satisfies the Vlasov-Poisson equation for any initial condition $k(0, \cdot) = k_0\in L^1(\mathbb{R}^6)$, which is non-negative, continuously differentiable, and satisfies the following conditions for a suitable constant $C > 0$, some $\delta > 0$, and all $(q,p) \in \mathbb{R}^6$:
\begin{align*}
(i) & \ \ k_0(q,p) \leq \frac{C}{(1+|p|)^{3+\delta}} \notag \\
(ii) & \ \ |\nabla k_0(q,p)| \leq \frac{C}{(1+|p|)^{3+\delta}} \notag \\
(iii) & \ \ \int_{\mathbb{R}^6}|p|^2 k_0(q,p) d^6(q,p) < \infty.
\end{align*}
 Under these conditions one gets global existence and uniqueness  of solutions of the Vlasov-Poisson equation with initial data $k_0$, such that
for each time interval $[0, T)$, there exists a constant $C> 0$, depending on $k_0$ and $T$ with
\begin{align*}
\sup_{0 \leq s < T} |\widetilde{k}_s|_{\infty} < C(T, k_0). 
\end{align*}

The characteristics of Vlasov-Poisson equation, similar to \eqref{Newtoneq}, are given by the following system of Newtonian differential equations
\begin{align}
 \begin{cases}
&\dot{\bar{q}}_i=\bar{p}_i \\
&\dot{\bar{p}}_i=f^N*\widetilde{k}_t(\bar{q}_i). \label{eff.flow}
\end{cases}
\end{align}
where $\widetilde{k}_t$ denotes the previously introduced `spatial density'.
Here the mean-field force $\bar{f}_{t}^{N}$ is defined by $\bar{f}_{t}^{N}=f^{N}*\tilde{k}_{t}^{N}$ and $\tilde{k}_t^{N}:=\R\times \R^{3}\rightarrow \R_{0}^{+}$ is given by 
\begin{align*}
\tilde{k}_{t}^{N}(q)\coloneqq \int k_t^{N}(p,q) d^{3}p.
\end{align*}
The system \eqref{eff.flow} is uniquely solvable on any interval  $[0,T]$. This provides a  flow $(\varphi^{\infty}_{s,t})_{s,t\in \R}$. $(\varphi_{s,t})_{s,t\in \R}=(^1\varphi_{\cdot,s}(x),{^2\varphi_{\cdot,s}(x)})$ solves the equations \eqref{eff.flow} where $\varphi_{s,s}(x)=x$ for any $x\in \R^{6}$ and $s\in \R$.
By construction we receive a trajectory which is influenced by the mean-field force and not by the pair interaction force like in the Newtonian system defined in Definition \ref{Def:Newtonflow}.
Later we will show that the two trajectories defined in Definition \ref{Def:Newtonflow} and Definition \ref{Def:Meanfieldflow} are close to each other.
To this end, we consider the lift of $\varphi^N_{t,s}(\cdot)$ to the $N$-particle phase-space, which we denote by $\Phi^N_{t,s}$. 
Denoting $\bar{F}:\R^{3N}\to\R^{3N}$ as the lift of the mean field force to the $N$-particle phase-space, i.e. $$(\overline{F}_t(X))_i \coloneqq  f^{N}*\tilde{k}_{t}^{N}\lbrack\tilde{k}_t\rbrack(x_i)$$ for $ X=(x_1,...,x_N)$ we finally define the mean-field flow analogously to Definition \ref{Def:Newtonflow}.
\begin{definition}\label{Def:Meanfieldflow}
 The  effective flow $\Phi_{t,s}^{N}=(\Phi_{t,s}^{1,N},\Phi_{t,s}^{2,N})=(\varphi_{t,s}^{N})^{\otimes N}$ is defined via  
\begin{align*}
\frac{d}{dt}\Phi_{t,s}^{N}(X)=(\Phi_{t,s}^{2,N}(X),\bar{F}(\Phi_{t,s}^{1,N}(X))
\end{align*}
with $\Phi_{s,s}^{N}(X)=X.$
\end{definition}
In contrast to the Newtonian Flow $\Psi_{t,s}^{N}$, the effective flow $\Phi_{t,s}^{N}$ conserves independence, which is crucial for the later proof.
For the purpose of justification of the common physical description, we compare the microscopic $N$-particle time evolution $\Psi^N_{t,s}$  with an effective one-particle description given by the Vlasov-Poisson flow $(\varphi^N_{t,s})_{t,s\in\R}:\R^{6}\to \R^{6}$ and prove convergence of $\Psi^N_{t,s}$ to the product of $\varphi^N_{t,s}$ in the limit $N\to\infty$ in a suitable sense. From this,  weak convergence of the $s$-particle marginals of the $N$-particle system to the corresponding $s$-fold products of  solutions of the Vlasov-Poisson equation follows.
It is usually referred to as propagation of molecular chaos. 
This is due to the fact that $\Phi^N$ consists of $N$ copies of $\varphi_t.$ Hence the particles are distributed i.i.d. with respect to the particle density $k^N$ defined in \eqref{Vlasov}.
The mean-field particles move independently, because we use the same force for every particle and thus we do not have pair interactions, which lead to correlations. 

In summary, for fixed $k_0$ and $N \in \N$, we consider for any initial configuration $X \in \R^{6N}$  two different time-evolutions: $\Psi^N_{t,0}(X)$, given by the microscopic equations and $\Phi^N_{t,0}(X)$, given by the time-dependent mean-field force generated by $f^N$. We are going to show that for typical $X$, the two time-evolutions are close in an appropriate sense. 
In other words, we have non-linear time-evolution in which $\varphi^N_{t,s}(\cdot\,; k_0)$ is the one-particle flow induced by the mean-field dynamics with initial distribution $k_0$, while, in turn, $k_0$ is transported with the flow $\varphi^N_{t,s}$. 
\section{On the mean-field limit for the Vlasov-Poisson system}\label{Chapter: VP}
In the following section we show that the $N$-particle trajectory $\Psi_t$ starting from $\Psi_0$ (i.i.d. with the common density $k_0$) remains close to the mean-field trajectory $\Phi_t$ with the same initial configuration $\Psi_0=\Phi_0$  on any finite time-interval $[0,T]$ and so the microscopic and the macroscopic descriptions are close.  Throughout this paper C denotes a positive finite constant which may vary from place to place but most importantly it will be independent of N.
\begin{theorem} \label{maintheorem}
Let $T>0$ and $k_0\in L^1(\R^{6})$ be a continuously differentiable probability density fulfilling
$
\sup_{N\in\N}\sup_{0\leq s\leq T}||\tilde{k}^N_s||_{\infty}\leq \infty.
$ Moreover, let $(\Phi^{\infty}_{t,s})_{t, s\in \R}$ be the related lifted effective flow defined in Definition \ref{Def:Meanfieldflow} as well as $({\Psi}^{N}_{t,s})_{t,s\in \R}$ the $N$-particle flow defined in Definition \ref{Def:Newtonflow}. 
 If $\sigma>0$ and $\beta=\frac{5}{12}-\sigma$, then for any $\gamma>0$ there exists a $C_\gamma>0$ such that for all $N\in \N$ it holds that
\begin{align}
\mathbb{P}\left(X\in \R^{6N}:\sup_{0\le s \le T}\left|\Psi_{s,0}^{N}(X)-{\Phi_{s,0}^{\infty}}(X)\right|_{\infty}>  N^{-\frac{1}{6}} \right)\le C_\gamma N^{-\gamma}.
\end{align}
\end{theorem}
This Theorem implies Propagation of Chaos. 
The main difference to \cite{BoersPickl}and \cite{Dustin} is that in the current case we analyse the advantages of the second order nature of the equation to transfer more information from the mean-field system to the true particles as introduced in \cite{grass}. 
As long as the true and their related mean-field particles are close in phase space, the types of their collisions are expected to be similar. 
Therefore we will divide the particles into sets, a `good', a `bad' and a `superbad' set, depending on their mean-field particle partners.
If for certain particles, pair collision are expected according to their auxiliary trajectories, then depending on the distance and their relative velocity, they will be labelled `bad' or `superbad'.
As for such particles larger deviations are expected after the collisions, we will allow larger distances to their related mean-field particles.
 With increasing  distance to their related mean-field description, particles are called ``good'', ``bad'' or ``superbad''.
We will use  that the number of `bad' or `superbad' particles is typically much smaller than the total particle number $N$.
Additionally, by using the integral version of Gronwalls Lemma we will make full use of the second order nature of the dynamics.
If two particle come exceptionally close to each other, one can expect a correspondingly large deviation of the true and mean-field trajectory.
However, for the vast majority  of particles, these deviations are typically only of a very limited duration.
In order not to overestimate the deviations between them, it makes sense to compare the dynamics on longer time periods.
The idea of dividing the particles into sets and using the integral version of Gronwalls Lemma  was  previously implemented in \cite{grass} for two particle sets, a so called `good' and a `bad' one.

\subsection{Heuristics for the particle groups}
The technical implementation is based on the technique introduced  in \cite{BoersPickl}. A heuristic introduction of this technique can be found in \cite{FeistlPickl}. 

In this paper we extended the technique  of  dividing the particles into subsets. The closer particles get to each other and the lower their relative speed, the worse they are in the sense that their interactions lead to comparably large deviations from their mean-field evolution. However, it will be shown that the number of bad particles compared to $N$ is extremely small. The small number will be useful in the estimates. It helps to control the future effect of the other particles despite their comparably large deviation form the mean field particle.
In a first step we will classify the particles according to their distance from one another and their relative velocities. Roughly one should think of
\begin{align*}
M_0:&=\lbrace i\in \lbrace 1,\hdots, N\rbrace |\exists t\geq 0 :|\bar{q}_j-\bar{q}_k|\le N^{-r_0}  \ \text{and} \  |\bar{p}_j-\bar{p}_k|\leq N^{-v_0} \rbrace\\
M_1:&=\lbrace i\in \lbrace 1,\hdots, N\rbrace |\exists t\geq 0 :|\bar{q}_j-\bar{q}_k|<N^{-r_1} \ \text{and} \ |\bar{p}_j-\bar{p}_k|\leq N^{-v_1} \rbrace\setminus M_0\\
\vdots\\
M_l:&=\lbrace i\in \lbrace 1,\hdots, N\rbrace |\exists t\geq 0 :|\bar{q}_j-\bar{q}_k|\leq N^{-r_{l}} \ \text{and} \ |\bar{p}_j-\bar{p}_k|\leq N^{-v_{l}} \rbrace\setminus \bigcup_{n=0}^{l-1}M_n.
\end{align*}
for $0\leq r_l\leq r_1 \leq r_0$ and $0\leq v_l\leq v_1 \leq v_0$.
It holds that
$\lbrace 1\hdots N\rbrace= \dot{\bigcup}M_n$. The particles contained in $M_0$ are the most problematic particles, the so-called `superbad' particles. 
An adjusted definition to the precise technical needs will be defined in Section \ref{Sec2} by so called collision classes.
As we are only interested to show the advantage of introducing more particle subsets we limit ourselves to three subsets.
Note that the definition of the sets $M_l$   refers to the mean-field dynamics $\Phi$ which conserves independence, not to $\Psi$! 
This makes it easy to calculate a bound for the probability of $X_i$ belonging to these sets. 
Standard law of large numbers arguments give that for all $\gamma\in \mathbb{N}$ there exists a $C_\gamma$ such that
$\mathbb{P}(|M_l|\geq N^{\delta_l})\leq C_\gamma N^{-\gamma}$ for some $\delta_l>0$.\\
The probability for a hit should be given by the Bolzmanzylinder $\mathbb{P}(\text{hit})=Cr^2v_{rel}$ for the relative velocity $v_{rel}$.
In our case $v_{rel}$ is also probabilistic with $\mathbb{P}(v_{rel}\leq v_{cut})\approx v_{cut}^3$.
So we should get a probabilistic bound of the form
\begin{align*}
\mathbb{P}(v_{rel}\leq v_{cut} \ \text{and} \ \text{hit})\leq Cr^2v_{cut}^4.
\end{align*}
The probability of finding $k$ particles inside the set $M_{l}$ around a bad particle is thus bounded from above by the binomial probability mass function with parameter $p:=\mathbb{P}(j\in M_{l}) $ at position $k$, i.e. for any natural number $0\leq A\leq N$ and any $t_n\leq t\leq t_{n+1}$
	$$\mathbb{P}\left(\mbox{card }(M_{l}\geq A\right))\leq\sum_{j=A}^ N \begin{pmatrix}
	N  \\
	j 
	\end{pmatrix}
	p^j (1-p)^{N-j}.$$
The mean of a binomially distributed random variables is given by $Np$ and thus the standard deviation by $\sqrt{Np(1-p)}<\sqrt{Np}$. The probability to find more than $Np+ a \sqrt{Np}$ particles in the set $M_l$ is exponentially small in $a$, i.e. there is a sufficiently large $N$ for any $\gamma>0$ and any $t$ with $t\in [t_n,t_{n+1}]$ such that
	$$\mathbb{P}\left(\mbox{card }(M_{l})\geq Np+a\sqrt{Np}\right)\leq a^{-\gamma}\;.$$
 The probability of finding more than $2Np=Np+\sqrt{Np}\sqrt{Np}$ (i.e. $a=\sqrt{Np})$ particles in the set $M_{l}$  is smaller than any polynomial in $N$,  i.e. there is a $C_\gamma$ for any $\gamma>0$ and any $t$ with $t_n\le t\le t_{n+1}$ such that
	$$\mathbb{P}\left(\mbox{card }(M_{l})\geq 2Np\right)\leq C_{\gamma}N^{-\gamma}.$$
This preliminary consideration leads us to assume that the number of particles in a bad subset can be estimated by $N^{2-2r_l-4v_l}$, which will also be proven later.
\subsection{Preliminary studies}
To implement this proposed strategy we collect and derive necessary results and properties.
Constants appearing in this paper will generically be denoted by $C$.
More precisely we will not distinguish constants appearing in a sequence of estimates, i.e. in an inequality chain $a\le Cb\le Cd$, the constants $C$ may differ.
The following Lemma constitutes the probability of a hit  i.e. the probability of the different types of collisions.
\begin{Lemma}\label{Prob of group}
Let $(\varphi_{t,s}^{N})_{t,s\in\R}$ be the related effective flow for $\beta \geq 0$ then there is an $C>0$ such that for $N^{-a_k}, N^{-b_k}>0, N\in \N$ and $\lbrack t_1,t_2\rbrack \subset \lbrack 0,T\rbrack$ it holds that
\begin{align*}
&\mathbb{P}\Big(X \in \R^6:\big( \exists t\in \lbrack t_1,t_2\rbrack: |\varphi_{t,0}^1(X)-\varphi_{t,0}^1(Y)|\leq  N^{-a_{k}}\\
&\qquad\qquad\qquad\qquad\qquad\quad\wedge |\varphi_{t,0}^2(X)-\varphi_{t,0}^2(Y)|
\leq  N^{-b_{k}}\big)\Big)\\
&\leq C( (N^{-a_{k}})^2 (N^{-b_{k}})^4 (t_2-t_1)+(N^{-a_{k}})^3 \max( N^{-a_{k}}, N^{-b_{k}})^3)
\end{align*}
\end{Lemma}
The proof of Lemma \ref{Prob of group} can be found in \cite[Lemma 2.1.4]{grass}.
This Lemma constitutes a probability bound for amount of particles belonging to a certain particle group, i.e.
\begin{align*}
\mathbb{P}(Y\in\R^6:Y\in M_l(X_k))\leq C  ( N^{-a_{l}})^2( N^{-b_{l}})^4.
\end{align*}
So far all $N$ particles were taken into account as possible interaction partners for the considered particle $X_i$.
This constitutes a worst case estimate.
The possible types of collisions and, accordingly, the impact on the force term can differ.
This will be taken into account later by defining collision classes.

We further introduce the underlying Gronwall Lemma, which takes into account the second order nature of the equation. The unlikely collisions are usually only of a limited duration. An integral Gronwall version pays respect to that.
\begin{Lemma}\label{Gronwall Lemma}
Let $u:[0,\infty)\to [0,\infty)$ be a continuous and monotonously increasing map as well as $l,f_1:\R\to [0,\infty)$ and $f_2:\R\times \R\to [0,\infty)$ continuous maps such that for some $n\in \N$ and for all $t_1>0,\ x_1,x_2\geq 0$
\begin{itemize}
\item[(i)]
$\begin{aligned}&
x_1< x_2 \Rightarrow f_2(t_1,x_1)\le f_2(t_1,x_2)
\end{aligned}$
\item[(ii)]
$ \exists K_1,\delta>0: \sup\limits_{\substack{x,y  \in [f_1(0),f_1(0)+\delta]\\s \in [0,\delta]}}|f_2(s,x)-f_2(s,y)|\le K_1|x-y|$.
\item[(iii)]
$ \begin{aligned}
& f_1(t_1)+ \int_0^{t_1}...\int_0^{t_n} f_2(s,u(s))dsdt_n...dt_2<  u(t_1) \ \land \\ 
& f_1(t_1)+\int_0^{t_1}...\int_0^{t_n}  f_2(s,l(s))dsdt_n...dt_2 \geq l(t_1),
\end{aligned}$
\end{itemize} 
then it holds for all $t\geq 0$ that $l(t)\le u(t)$.
\end{Lemma}
The proof of Lemma \ref{Gronwall Lemma} can be found in \cite[Lemma 2.1.1]{grass}.
 It uses a mean-value Theorem and relies on an estimate of the 
first derivative of $f$. The respective bound is given by the following function $g^N$ which is defined such that  $g^N(q)\geq \left|\nabla f^N(q)\right|$ wherever the latter exists, i.e. for all $|g|\neq N^{-\beta}$ 
\begin{definition}
\label{force g}
For $N\in \N\cup\lbrace \infty \rbrace$ we define
\begin{align*}
g^{N}:\R^{3}\rightarrow\R^{3}, q\mapsto\begin{cases}
 2 N^{3\beta}& \text{if}\ |q|\leq 3 N^{-\beta} \\
 54\frac{1}{|q|^{3}} & \text{if}\  |q|> 3N^{-\beta} 
\end{cases} 
\end{align*}
for $0<\beta$.
\end{definition}
Analogously to the total force of the system $F^N$, we  will use the notation $G^N:\R^{6N}\rightarrow\R^{3N}$ the total fluctuation of the system. Thus the i'th component of $G^N$ gives the fluctuation exhibited on a single coordinate $j$: $$(G^N(X))_j\coloneqq  \sum_{i\neq j}\frac{1}{N}g^N(q_i-q_j).$$
 Since $f$ is differentiable for any $|g|\neq N^{-\beta}$  we can use a mean-value argument to control differences of the values of $f$ at different points. 
\begin{Lemma}\label{LipschitzLemma for f}
\begin{itemize}
\item[a)] For $a,b,c\in \R^3$ with $|a|\leq \min(|b|,|c|)$ the following relations hold
\end{itemize}
\begin{align}\label{minabschätzung}
|f^N(b)-f^N(c)|&\leq g^{N}(a)|b-c|.
\end{align} 
\item[b)]
If $\|X_t-\overline X_t\|_\infty\leq 2N^{-\beta}$, then it holds that
	\begin{align}\label{LipschitzLemma_f}
	\left\| F^N(X_t)-F^N(\overline X_t)\right\|_\infty\leq C\|G^N(\overline X_t)\|_\infty\|X_t-\overline X_t\|_\infty,
	\end{align}
	for some $C>0$ independent of $N$.

\end{Lemma}
\begin{proof}
\begin{itemize}
	\item[a)] For the case $|a|\leq 3N^{-\beta}$ we have $\norm{\nabla f^N}_\infty\leq  2N^{3\beta}$ and thus $2N^{3\beta}$ constitutes a Lipschitz-constant for $f^N.$\\
	For  $|a|\geq 3N^{-\beta},$ we get by the mean value theorem and the fact, that  $\nabla f^N(x)$ is decreasing 
	\begin{align*}
|f^N(b)-f^N(c)|\leq |\nabla f^N(a)||b-c|\leq C\left(\frac{1}{|a|}\right)^3|b-c|\leq Cg^{N}(a)|b-a|.
	\end{align*}
\item[b)]
For any $x,\xi\in\R^3$ with $|\xi|<2N^{-\beta}$, we have for $|x|<3N^{-\beta}$
\begin{align}\label{gbound1}
|f^N(x+\xi)-f^N(x)|\leq 2N^{3\beta}|\xi|\leq g^{N}(x)|\xi|
	\end{align} 
by applying  estimate \ref{minabschätzung} and for choosing without loss of generality $a=b=x+\xi$ and $c=x$.
For $|x|\geq 3N^{-\beta}$ we use the fact that in this case small changes in the argument of the function lead to small changes in the function values, i.e. for  $\xi\leq 2N^{-\beta}$ we have $g^{N}(x+\xi)\leq C g^{N}(x)$. Thus we have by estimate \ref{minabschätzung}
\begin{align*}
|f^N(x+\xi)-f^N(x)|\leq C g^{N}(x+\xi)|\xi|\leq Cg^{N}(x)|\xi|.
\end{align*}
Applying claim \eqref{gbound1} one has
	\begin{align}
	|(F^N(X_t))_i-(F^N(\overline X_t))_i|&\leq \frac{1}{N}\sum\limits_{j\neq i}^{N}\left| f^N(x_i^t-x_j^t)-f^N(\overline x_i^t-\overline x_j^t)\right|\notag\\
	&\leq \frac{C}{N}\sum\limits_{j\neq i}^{N}g^{N}(\overline x_i^t-\overline x_j^t)\left|x_i^t-x_j^t-\overline x_i^t+\overline x_j^t\right|\notag \\
	&\leq C( g^{N}(\overline X_t))_i\left|X_t-\overline X_t\right|_\infty,
	\end{align}
	which leads to estimate \eqref{LipschitzLemma_f}.
	\end{itemize}
	\end{proof}

\begin{Lemma} \label{Lemma distance same order}
Let $T>0$ and $k_0$ be a probability density fulfilling the assumptions of Theorem \ref{maintheorem} where $(\varphi^{N,c}_{t,s})_{t, s\in \R}$ shall be the related effective flow defined in Definition \ref{Def:Meanfieldflow}. Then there exist a $C_1,C_2>0$ such that for all configurations $X,Y\in \R^6,\ N\in \N\cup \{\infty\}$ and $t,t_0\in [0,T]$ it holds that 
\begin{align*}
|\varphi_{t,t_0}^{N}(X)-\varphi_{t,t_0}^{N}(Y)| \le |X-Y|e^{C_1|t-t_0|}
\end{align*}
and 
\begin{align*}
|f_c^N*\widetilde{k}^{N}_t(^1X)-f^N_c*\widetilde{k}^{N}_t(^1Y)| \le C_2|^1X-{^1Y}|.
\end{align*}
\end{Lemma}
The proof of this Lemma can be found in \cite{grass} (Lemma 2.1.2).
Last but not least we come to the most important corollary of this chapter. It provides suitable upper bounds for almost all integrals appearing in the proof of the main theorem.
\begin{corollary}\label{corollary phi and  psi}
Let $k_0$ be a probability density fulfilling the assumptions of Theorem \ref{maintheorem} and $(\varphi^{N,c}_{t,s})_{t, s\in \R}$ be the related effective flow defined in Definition \ref{Def:Meanfieldflow} as well as $(\Psi^{N,c}_{t,s})_{t,s\in \R}$ the $N$-particle flow defined in Definition \ref{Def:Newtonflow}. Let additionally for $N,n\in \N$, $1<\lambda\le 3$, $C_0>0$ and $c_N>0$ $h_N:\R^{3}\to \R^n$ be a continuous map fulfilling $$ |h_N(q)|\le\begin{cases} C_0 c_N^{-\lambda},&\ |q|\le c_N\\\frac{C_0}{|q|^{\lambda}},&\   |q|> c_N \end{cases}.$$
\begin{itemize}
\item[(i)] Let for $Y,Z\in \R^6$ $t_{min}\in [0,T]$ be a point in time where 
\begin{align*}
 \min_{0\le s \le T}|\varphi^{1,N}_{s,0}(Z)-{\varphi^{1,N}_{s,0}}(Y)|=&|\varphi^{1,N}_{t_{min},0}(Z)-{\varphi^{1,N}_{t_{min},0}}(Y)|=:\Delta r>0 \ \land \\
 &|\varphi^{2,N}_{t_{min},0}(Z)-{\varphi^{2,N}_{t_{min},0}}(Y)|=:\Delta v>0,
\end{align*}
then there exists a $C_1>0$ (independent of $Y,Z\in \R^{6}$ and $N\in \N$) such that
\begin{align*}
& \int^{T}_{0}|h_N(^1\varphi^{N}_{s,0}(Z)-{^1\varphi^{N}_{s,0}}(Y))|ds
\le  C_1 \min\big(\frac{1}{\Delta r^{\lambda}},\frac{1}{c_N^{\lambda-1}\Delta v},\frac{1}{\Delta r^{\lambda-1}\Delta v}\big).
\end{align*}
\item[(ii)] Let $T>0$, $i,j\in \{1,...,N\}$, $i\neq j$, $X\in \R^{6N}$ and $Y,Z\in \R^6$ be given such that for some $\delta>0$
$$N^{\delta}|\varphi^{1,N}_{t_{min},0}(Y)-{\varphi^{1,N}_{t_{min},0}}(Z)|\le |\varphi^{2,N}_{t_{min},0}(Y)-{\varphi^{2,N}_{t_{min},0}}(Z)|=:\Delta v$$
and
\[\sup_{0\le s \le T}|\varphi^{N}_{s,0}(Y)-[{\Psi^{N}_{s,0}}(X)]_i|\le N^{-\delta}\Delta v\land \sup_{0\le s \le T}|\varphi^{N}_{s,0}(Z)-[{\Psi^{N}_{s,0}}(X)]_j|\le N^{-\delta}\Delta v\]
where $t_{min}$ shall fulfil the same conditions as in item (i). Then there exists a $N_0\in \N$ and $C_2>0$ (independent of $X\in \R^{6N}$, $Y,Z\in \R^6$) such that for all $N\geq N_0$ 
\begin{align*}
&\int^{T}_{0}|h_N([\Psi^{1,N}_{s,0}(X)]_i-[\Psi^{1,N}_{s,0}(X)]_j)|ds\\
\le & C_2\min\big(\frac{1}{c_N^{\lambda-1}\Delta v}, \frac{1}{\min\limits_{0\le s\le T}|[\Psi^{1,N}_{s,0}(X)]_i-[\Psi^{1,N}_{s,0}(X)]_j|^{\lambda-1}\Delta v}\big).
\end{align*}
\end{itemize}
\end{corollary}
The proof of this Corollary can be found in \cite{grass}( Corollary 2.1.1).
\section{Proof of Theorem \ref{maintheorem}}\label{Sec2}
This proof and the notation is based on \cite{grass}.
Some of their estimates can be directly implied in our situation.
For simplification we consider three different subsets of particles depending on their distance and relative velocity to other particles. 
The first set $\mathcal{M}_{s}$ of the `superbad' ones includes all particles $j\in \lbrace 1\hdots N\rbrace$ for which there is a time $t\geq 0$ such that $|\bar{q}_j-\bar{q}_k|\leq N^{-s_{r}}$ and $|\bar{p}_j-\bar{p}_k|\leq N^{-s_{v}}$. 
They are expected to come very close to other particles with small relative velocity. 
The second set $\mathcal{M}_b$, containing the so called `bad' particles, which come  intermediately  close with intermediate relative velocity, is defined by analogue conditions $|\bar{q}_j-\bar{q}_k|\leq N^{-b_r}$ and $|\bar{p}_j-\bar{p}_k|\leq N^{-b_v}$, excluding the particles already in $\mathcal{M}_{s}$.
Finally the reaming unproblematic `good' ones, which never come close to each other while having small relative velocity are contained in $\mathcal{M}_g=(M_b\cup M_{s})^c$.
An important point in the proof is that the better the particle is, the less distance  we allow  to the mean-field particle.
Furthermore it depends only on their corresponding mean-field particle whether a particle is considered good, bad or superbad .
In the course of a simple notation we introduce collision classes, which turn out to be very important throughout the proof, as each collision class has a different impact on the force term.
They are intended to cover all possible ways in which particles can interact and thus the particle subsets can be defined using this notation. 
\begin{definition}\label{Def:setM}
For $r,R,v,V\in \R_{0}^{+}\cup \lbrace\infty\rbrace, t_1,t_2\in \lbrack 0,T\rbrack$ and $Y\in \R^6$ the set $M_{(r,R),(v,V)}^{N,(t_1,t_2)}(Y)\subset \R^6$ is defined as follows:
\begin{align*}
&Z\in M_{(r,R),(v,V)}^{N,(t_1,t_2)}(Y)\Leftrightarrow Z\neq Y \wedge\exists t\in \lbrack t_1,t_2\rbrack:\\
& r\leq \min_{t_1\leq s\leq t_2} |\varphi_{s,0}^1(Z)-\varphi_{s,0}^1(Y)|=|\varphi_{t,0}^1(Z)-\varphi_{t,0}^1(Y)|\leq R\\
&\wedge v\leq|\varphi_{t,0}^2(X)-\varphi_{t,0}^2(Y)|\leq V.
\end{align*}
\end{definition}
Here $(\varphi_{s,r}^N)_{s,r\in\R}$ is the one particle mean-field flow, defined in Definition \ref{Oneparticleflow}, related to the considered initial density $k_0$.
In addition, we will use the following short notation for the sets defined in Definition \ref{Def:setM}:
\begin{align*}
M_{R;V}^{N,(t_1,t_2)}(Y):=M_{(0,R),(0,V)}^{N,(t_1,t_2)}(Y)\\
M_{(r,R),(v,V)}^{N}(Y)\coloneqq  M_{(r,R),(v,V)}^{N,(0,T)}(Y)\\
M_{R,V}^{N}(Y)\coloneqq  M_{(0,R),(0,V)}^{N,(0,T)}(Y).
\end{align*}
The set $G^N(Y)\subset\R^6$ of non-problematic particle interactions is defined by
\begin{align}\label{Def:good}
G^N(Y)\coloneqq  (M^{N}_{6r_b,v_b}\cup M^{N}_{6r_s,v_s})^c=(M^{N}_{6r_b,v_b})^c,
\end{align}
for $r_b= N^{-\frac{7}{24}-\sigma},v_b=N^{-\frac{1}{6}},r_s=N^{-\frac{1}{3}-\sigma}$ and $v_s=N^{-\frac{5}{18}}$.
Next we split the particles in three subsets using the notation of the collision classes as mentioned before: A `superbad' subset where super hard collisions are expected to happen, a `bad' subset where hard collisions are expected  and a subset of the remaining `good' particles.
\begin{align*}
\mathcal{M}_{g}^{N}(X):&=\lbrace i\in \lbrace 1,\hdots , N\rbrace :\forall j\in\lbrace 1,\hdots,N\rbrace\setminus \lbrace i\rbrace :X_j\in G^N(X_i))\rbrace\\
\mathcal{M}_{s}^{N}(X):&=\lbrace i \in \lbrace 1,\hdots,N\rbrace :\exists j\in \lbrace 1,\hdots ,N\rbrace\setminus \lbrace i\rbrace:X_j\in M_{(0,r_s),(0,v_s)}^{N}(X_j)\rbrace\\
\mathcal{M}_{b}^{N}(X):&=\lbrace 1,\hdots,N\rbrace \setminus (\mathcal{M}_{g}^N(X)\cup \mathcal{M}_{s}^N(X)).
\end{align*}
The  distinction between `good', `bad' or `superbad' particles depends only on their  mean-field dynamics,  as the sets above are defined by application of the collision classes which themselves are defined by the mean-field flow.

Each of the three particle subsets has its own stopping time which is defined by
\begin{align*}
\tau_g^N\coloneqq  \sup\lbrace t\in \lbrack 0,T\rbrack: \max_{i\in \mathcal{M}_{g}^{N}}\sup_{0\leq s\leq t}|\lbrack \Psi_{s,0}^N(X)\rbrack_i-\varphi_{s,0}^N(X_i)|\leq \delta_g^N=N^{-\frac{5}{12}+\sigma}\rbrace\\
\tau_b^N\coloneqq  \sup\lbrace t\in \lbrack 0,T\rbrack: \max_{i\in \mathcal{M}_{b}^{N}}\sup_{0\leq s\leq t}|\lbrack \Psi_{s,0}^N(X)\rbrack_i-\varphi_{s,0}^N(X_i)|\leq \delta_b^N=N^{-\frac{7}{24}-\sigma}\rbrace\\
\tau_{s}^N\coloneqq  \sup\lbrace t\in \lbrack 0,T\rbrack: \max_{i\in \mathcal{M}_{s}^{N}}\sup_{0\leq s\leq t}|\lbrack \Psi_{s,0}^N(X)\rbrack_i-\varphi_{s,0}^N(X_i)|\leq \delta_{sb}^N=N^{-\frac{1}{6}-\sigma}\rbrace.
\end{align*}
The stopping time for the whole system is given by 
\begin{align}\label{stoppingtime}
\tau^N(X)\coloneqq  \min(\tau_g^N(X),\tau_b^N(X),\tau_{s}^N(X)),
\end{align}
where $\delta_g^N=N^{-\beta}$, $\delta_b^N=N^{-d_b}$ and $\delta_{s}^N=N^{-d_{s}}.$\\
We will see that configurations fulfilling $\tau^N(X)<T$ become sufficiently small in probability for large values of $N$ and hence Theorem \ref{maintheorem} follows.\\ 
The main part of the proof is based on the application of Gronwall's Lemma to show that $\sup_{0\leq s\leq t}|\lbrack \Psi_{s,0}^N(X)\rbrack_i-\varphi_{s,0}^N(X_i)|_{\infty}$ stays typically small for large $N$.\\
Therefore we estimate the right derivative of $\sup_{0\leq s\leq t}|\lbrack \Psi_{s,0}^N(X)\rbrack_i-\varphi_{s,0}^N(X_i)|$, which is given by
\begin{align*}
& \frac{d}{dt_+}\sup_{0\le s\le t}\left|\left[\Psi^{1,N}_{s,0}(X)\right]_i-{\varphi^{1,N}_{s,0}}(X_i)\right|\\ \leq &  \left|\left[\Psi^{2,N}_{t,0}(X)\right]_i-{\varphi^{2,N}_{t,0}}(X_i)\right|\\
\leq &\left|\int_{0}^t\frac{1}{N}\sum_{j\neq i}f^{N}\left(\left[\Psi^{1,N}_{s,0}(X)\right]_i-\left[\Psi^{1,N}_{s,0}(X)\right]_j\right)-f^{N}*\widetilde{k}^N_s({\varphi^{1,N}_{s,0}}(X_i))ds\right|.
\end{align*}
For technical reasons we will distinguish between observing a `good', `bad' or `superbad' particle for further estimation of this expression.

\subsection{Controlling the deviations of good particles}\label{case1}
In the first Section we focus on the case, that the considered particle $X_i$ is `good' and use a similar proof technique as presented in \cite{BoersPickl,Dustin,grass}.
First we break down the equation in terms of interaction partners. 
They themselves can be `superbad', `bad' or `good' relative to $X_i$. 
Of course the set of particles having a bad or superbad interaction is empty in this case as having an unpleasant collision is symmetrical and consequently the underlying term will vanish later, but still, it will be technically useful to split the equation in that way.

Let $i\in \mathcal{M}_g^N(X)$ and $0\leq t_1\leq t\leq T$ 
{\allowdisplaybreaks \begin{align}
 & \left|\int_{t_1}^t\frac{1}{N}\sum_{j\neq i}f^{N}([\Psi^{1,N}_{s,0}(X)]_i-[\Psi^{1,N}_{s,0}(X)]_j)-f^{N}*\widetilde{k}^N_s({\varphi^{1,N}_{s,0}}(X_i))ds\right| \label{terminitial}\\
\le & \left|\int_{t_1}^t\frac{1}{N}\sum_{j\neq i}f^{N}([\Psi^{1,N}_{s,0}(X)]_i-[\Psi^{1,N}_{s,0}(X)]_j)\mathds 1_{(G^N(X_i))^C}(X_j)ds\right| \notag\\
& + \left|\int_{t_1}^t\Big(\frac{1}{N}\sum_{j\neq i}f^{N}([\Psi^{1,N}_{s,0}(X)]_i-[\Psi^{1,N}_{s,0}(X)]_j)\mathds 1_{G^N(X_i)}(X_j)\right. \notag \\
& \ \ \ \ \left. -f^{N}*\widetilde{k}^N_s({\varphi^{1,N}_{s,0}}(X_i))\Big)ds\right| .\label{terminitialsplit}
\end{align}

Using triangle inequality in the last two lines of Equation \ref{terminitialsplit} one gets that the previous Term \ref{terminitial} is bounded by
\begin{align}
 & \left|\int_{t_1}^t\frac{1}{N}\sum_{j\neq i}f^{N}([\Psi^{1,N}_{s,0}(X)]_i-[\Psi^{1,N}_{s,0}(X)]_j)\mathds 1_{(G^N(X_i))^C}(X_j)ds\right| \label{eq:good1} \\
& + \left|\int_{t_1}^t\frac{1}{N}\sum_{j\neq i}\Big(  f^{N}([\Psi^{1,N}_{s,0}(X)]_i-[\Psi^{1,N}_{s,0}(X)]_j)\mathds 1_{G^N(X_i)}(X_j)  \notag\right. \\
& \ \ \ \ \left.-f^{N}(\varphi^{1,N}_{s,0}(X_i)-{\varphi^{1,N}_{s,0}}(X_j))\mathds 1_{G^N(X_i)}(X_j)\Big) ds\right| \label{eq:good2} \\
& +\left| \int_{t_1}^t \frac{1}{N}\sum_{j\neq i}f^{N}(\varphi^{1,N}_{s,0}(X_i)-{\varphi^{1,N}_{s,0}}(X_j))\mathds 1_{G^N(X_i)}(X_j)ds\right. \notag \\
& \ \ \ \ \left.-\int_{t_1}^t\int_{\R^6}f^N({\varphi^{1,N}_{s,0}}(X_i)-{\varphi^{1,N}_{s,0}}(Y))\mathds 1_{G^N(X_i)}(Y)k_0(Y)d^6Yds\right| \label{eq:good3}  \\
& +\left|\int_{t_1}^t\int_{\R^6}f^N({\varphi^{1,N}_{s,0}}(X_i)-{\varphi^{1,N}_{s,0}}(Y))\mathds 1_{G^N(X_i)}(Y)k_0(Y)d^6Yds\notag\right.\\
& \ \ \ \ \left. -\int_{t_1}^tf^{N}*\widetilde{k}^N_s(\varphi^{1,N}_{s,0}(X_i))ds\right| \label{eq:good4} 
\end{align}

\subsubsection{Estimate of Term \ref{eq:good1} and Term \ref{eq:good4}}
Recall that $i\in \mathcal{M}_g^N(X)$ and that the set $(G^N(X_i))^c=M^N_{6r_b,v_b}$ includes all particles which come close to $X_i$ while having small relative velocity. 
Thus the characteristic function $\mathds 1_{(G^N(X_i))^c}(X_j)=0$ for $i\in \mathcal{M}_g^N(X)$ and therefore Term \ref{eq:good1} vanishes and we are left to estimate Term \ref{eq:good4}. 
For the Lebesgue measure preserving diffeomorphism the following holds
\begin{align*}
 f^{N}*\widetilde{k}^N_s({\varphi^{1,N}_{s,0}}(X_i))
&= \int_{\R^6}f^{N}({\varphi^{1,N}_{s,0}}(X_i)-{^1Y})k^N_s(Y)d^6Y\\
&= \int_{ \R^6}f^{N}({\varphi^{1,N}_{s,0}}(X_i)-{\varphi^{1,N}_{s,0}}(Y))k_0(Y)d^6Y.
\end{align*}
So we get for Term \ref{eq:good4}
\begin{align*}
&\left|\int_{t_1}^t\int_{ \R^6}f^{N}(\varphi^{1,N}_{s,0}(X_i)-{\varphi^{1,N}_{s,0}}(Y))k_0(Y)\mathds 1_{G^N(X_i)}(Y)d^6Yds\right. \notag \\
&\left.\qquad\qquad-\int_{t_1}^tf^{N}*\widetilde{k}^N_s({\varphi^{1,N}_{s,0}}(X_i))ds\right| \notag \\
= & \big|\int_{t_1}^t\int_{ \R^6}f^{N}(\varphi^{1,N}_{s,0}(X_i)-{\varphi^{1,N}_{s,0}}(Y))k_0(Y)(\mathds 1_{G^N(X_i)}(Y)-1)d^6Yds\big| \notag \\
\le & T\|f^{N}\|_{\infty}\int_{ \R^6}\mathds 1_{(G^N(X_i))^C}(Y)k_0(Y)d^6Y \notag \\
\le & TN^{2\beta}\mathbb{P}\big(Y \in \R^6:Y\notin G^N(X_i) \big)\\
\le&TN^{2\beta}\mathbb{P}\big(Y\in \R^6:Y\in M^N_{r_b,v_b}(X_i)  \big)\\
\le& CTN^{2\beta} N^{-2b_r-4b_v}
\end{align*}
This is small under a suitable choice of parameters.
\subsubsection{Law of large numbers for Term \ref{eq:good2} and Term \ref{eq:good3}}
For the remaining Terms \eqref{eq:good2} and \eqref{eq:good3} we provide a version of law of large numbers which takes into account the different types of collision classes which could occur. 
Each collision type has a different impact on the force and a certain probability. 
For that reason it is useful for the estimates to distinguish between them.
\begin{theorem}\label{Prop:LLN} Let $\delta,C_0>0$, $N{\in \N} $ and let $(X_{k})_{k\in \N}$ be a sequence of i.i.d. random variables $X_k:\Omega \to \R^6$ distributed with respect to a probability density $k\in \mathcal{L}^1(\R^6)$. 
Moreover, let $(M^N_i)_{i\in I}$ be a family of (possibly $N$-dependent) sets $M^N_i\subseteq \R^6$ fulfilling $\bigcup_{i \in I}M^N_i=\R^6$ where $|I|<C_0$ and $h_N\coloneqq  \R^6\to \R$ are measurable functions which fulfil on the one hand $\|h_N\|_{\infty}\le C_0N^{1-\delta}$ and on the other hand \[\max_{i \in I}\int_{M^N_i}h_N(X)^2k(X)d^6X\le C_0N^{1-\delta}.\]
Then for any $\gamma>0$ there exists a constant $C_{\gamma}>0$ such that for all $N\in\N$
	\begin{align} \mathbb{P}_t\Bigl[ \Bigl\lvert \frac{1}{N}\, \sum\limits_{j=1}^N h_N(X_j) -\int_{\R^6} h_N(X)k_t (Z)d^6X \Bigr\rvert \geq 1\Bigr] \leq \frac{C_{\gamma}}{N^\gamma}.
	\end{align}
\end{theorem}	
\begin{proof}
 By Markov's inequality, we have for every $M \in \N$: 
	\begin{align} 
	&\mathbb{P}_t\Bigl[ \Bigl\lvert \frac{1}{N}\, \sum\limits_{j=1}^N h_N(X_j) -\int_{\R^6} h_N(X)k_t (X)d^6X\rvert\geq 1 \Bigr]\\
	&\leq \mathbb{E}\Bigl[ N^{-2M}\, \Bigl\lvert \frac{1}{N}\, \sum\limits_{j=1}^N h_N(X_j) - \int_{\R^6} h_N(X)k_t (X)d^6X \Bigr\rvert^{2M} \Bigr],
	\end{align}
	where $\mathbb{E}\lbrack\cdot\rbrack$ denotes the expectation with respect to the N-fold product of $k$.
	
	\noindent Let $\mathcal{M} \coloneqq   \lbrace \mathbf{\gamma} \in \N_0^N \mid \lvert \mathbf{\gamma} \rvert = 2M \rbrace$ be the set of multiindices $\mathbf{\gamma} = (\gamma_1, ..., \gamma_N)$ with $\sum\limits_{i=1}^{N} \gamma_i = 2M$. Let
	\begin{align*} G_\mathbf{\gamma}(X) \coloneqq   \prod \limits_{j=1}^N \bigl( h_N(X_j) -\int_{\R^6} h_N(X)k_t (X)d^6X)^{\gamma_i}. \end{align*}
	\noindent Then 
	\begin{align*} &N^{-2M}\mathbb{E} \Bigr[ \Bigl(\sum\limits_{j=1}^N h_N(X_j) -\int_{\R^6} h_N(X)k_t (X)d^6X\bigr) \Bigr)^{2M}\Bigr] \\
	&\leq N^{-2M}\sum_{\gamma_1,\hdots,\gamma_N\in\mathcal{M}} \mathbb{E} \Bigr[ \Bigl(G_\mathbf{\gamma}(X))^{\gamma_j}\Bigr]. \end{align*}
	
	\noindent Note that $\mathbb{E}(G_\mathbf{\gamma }) = 0$ whenever there is a $1 \leq i \leq N$ such that $\gamma_i =1$. 
	This can be seen by integrating the $i$'th variable first.\\

	\noindent For the remaining  terms, we have for any $1 \leq m \leq M$:
	\begin{align*} \lvert \bigl( h_N(X_j) -\int_{\R^6} h_N(X)k_t (X)d^6X\bigl)^{\gamma_i}\rvert \leq2^{\gamma_j} |h_N(X_j)|^{\gamma_j} +|\int_{\R^6} h_N(X)k_t (X)d^6X| ^{\gamma_j} .
	\end{align*}
	As $||h_N||\leq C_0N^{1-\delta}$, it follows for $m\geq 2$ 
	
	\begin{align*}
	&\int_{\R^6} \lvert h\rvert^m (X)k_t(X) \, \mathrm{d}^6X \leq C_0 \max_{i\in I}\int_{M_i^N} \lvert h\rvert^m (X)k_t(X) \, \mathrm{d}^6X\\
	\leq &C_0 ||h_N||_{\infty}^{m-2}\max_{i\in I}\int_{M_i^N} h_N(X)^2k(X)\mathrm{d}^6X\leq C_0 (C_0^{m-2}N^{(m-2)(1-\delta)})(C_0N^{1-\delta})
	\end{align*} 
	Let $R\coloneqq  \sqrt{\int_{\R^6}h_N^2(X)\mathrm{d}^6X}$, then it holds that
	\begin{align*}
	&\int_{\mathbb{R}^6} \lvert h\rvert (X)k_t(X) \, \mathrm{d}^6X \leq \frac{1}{R}\underbrace{ \int_{\mathbb{R}^6}  h^2 (X)k_t(X) }_{=R^2}+\underbrace{ \int_{\mathbb{R}^6}  |h (X)|\mathds 1_{\lbrack 0,R\rbrack}k_t(X) }_{\leq R}
\\
&\leq 2\big(C\max_{i\in I}\int_{M_i^N} h_N^2(X)k(X)\mathrm{d}^6X\big)^{\frac{1}{2}}	\leq C M^{\frac{1}{2}(1-\delta)}.
	\end{align*}
Since the constraints on the maps $h_N$ become more stringent with an increase in the chosen value of $\delta$, we can restrict our consideration to specific values, such as the interval $(0,1]$.
	If we additionally identify $|\gamma|\coloneqq  |\{i\in\{1,...,N\}: \gamma_i\neq 0\}|$ and recall that only tuples matter where $\gamma_i \neq 1\ $ for all $i\in \{1,...,N\}$ as well as $\sum_{i=1}^N\gamma_i=2M$, then application of these estimates and relations above yield that for all other multiindices, we get
	\begin{align*}\label{LLNestimate} \mathbb{E}_t (G^\mathbf{\gamma}) \leq  &\prod \limits_{j=1:\gamma_i\geq 2}^N  \big(C^{\gamma_i}N^{(\gamma_i-2)(1-\delta)}N^{1-\delta} \big)\leq C^{2M}N^{2M(1\delta)}N^{|\gamma|(\delta-1)},
	\end{align*}
	by using that the particles are statistically independent.
	\noindent Finally, we observe that for any $l \geq 1$, the number of multiindices $\mathbf{\gamma} \in \mathcal{M} $ with $\#\mathbf{\gamma} = l$ is bounded by
	\begin{equation*} \sum\limits_{\#\mathbf{\gamma} = l} 1 \leq \binom{N}{l} (2M)^l  \leq (2M)^{2M} N^l. \end{equation*}
	
	\noindent Thus
	\begin{align*}  
	&\frac{1}{N^{2M}} \sum\limits_{\mathbf{\gamma} \in \mathcal{M}} \mathbb{E}(G^\mathbf{\gamma})\\
	&\leq  \frac{N^{2M(1-\delta)}}{N^{2M}}\, \sum_{\gamma\in\mathcal{M}} C^{M}N^{|\gamma|(\delta-1)}\\
	&\leq C^{M} N^{-2M(\delta}\,\sum_{k=1}^{M} N^{k}(2M)^MN^{k(\delta-1)} \\
	&\leq (CM)^{M} N^{-\delta M},
	\end{align*}
	where $C$ is some constant depending on $M$. 
	Choosing $M$ arbitrary large proofs the Theorem.
\end{proof}
\subsubsection{Estimate of Term \ref{eq:good3}}
It is left to show that the third Term \ref{eq:good3}, respectively 
\begin{align*}
 \left|\int_{t_1}^t  \frac{1}{N}\sum_{j\neq i}f^{N}\right.&\left(\varphi^{1,N}_{s,0}\left(X_i\right)-{\varphi^{1,N}_{s,0}}\left(X_j\right)\right) \mathds 1_{G^N\left(X_i\right)}\left(X_j\right)\\
&\left. -\int_{\mathbb{R}^6}^{\;}f^N ({\varphi^{1,N}_{s,0}}(X_i)-{\varphi^{1,N}_{s,0}}(Y))\mathds 1_{G^N(X_i)}(Y)k_0 (Y)d^6Y  ds \right|
\end{align*}
  stays small for typical initial data. 
	Analogously to the function $h_N$ from Theorem \ref{Prop:LLN}, we define for arbitrary $Y\in\mathbb{R}^6$ the function
\begin{align}\label{h1,N}
h_{1,N}^t(y,\cdot):\mathbb{R}^6\rightarrow\mathbb{R}^3,Z\mapsto N^{\alpha}\int_{0}^{t}f^{N}(\varphi^{1,N}_{s,0}(Y)-{\varphi^{1,N}_{s,0}}(Z))ds \mathds 1_{G^N(Y)}(Z),
\end{align}	
with $0<\alpha\leq \frac{5}{12}$ or more precisely $0<\alpha=\beta+\sigma$.
As $h_{1,N}^t(Y,\cdot)$ does not map to $\mathbb{R}$ as assumed in Theorem \ref{Prop:LLN} it can still be applied on each component separately. 
If it holds for each component then it holds for the related vector valued map.
The fact that the Theorem only makes statements for certain points in time will be generalized later.

We are left to check if the assumptions of Theorem \ref{Prop:LLN} on the force term are fulfilled. 
Therefore we abbreviate $\tilde{r}\coloneqq  \max(r,N^{-\beta})$ for $r\geq 0$ and we obtain by Corollary \ref{corollary phi and  psi} and Lemma \ref{Prob of group} for $0\leq v\leq V$, $0\leq r\leq R$ and $\lambda=2 $ that
{\allowdisplaybreaks
\begin{align}\label{Var(f)}
&\int_{M^N_{(r,R),(v,V)}(Y)}\left(\int_{0}^t |f^{N}(\varphi^{1,N}_{s,0}(Z)-{\varphi^{1,N}_{s,0}}(Y))|ds\right)^2k_0(Z) d^6Z \notag \\
\le & C\big(\min(\frac{1}{\widetilde{r}^{\lambda}},\frac{1}{\widetilde{r}^{\lambda-1}v})\big)^{2}\int_{M^N_{(r,R),(v,V)}(Y)}k_0(Z) d^6Z\notag \\
\le &C\min \big(\frac{1}{\widetilde{r}^{2\lambda}},\frac{1}{\widetilde{r}^{2(\lambda-1)}v^2}\big)\min \big(1,R^2,R^2V^4+R^3\max(V^3,R^3)\big)\notag \\
\le &C\min\big(\frac{1}{\widetilde{r}^{2(\lambda-1)}v^2},\frac{R^2}{\widetilde{r}^{2(\lambda-1)}v^2},\frac{R^2V^4}{\widetilde{r}^{2(\lambda-1)}\max(\widetilde{r},v)^2}+\frac{R^6}{\widetilde{r}^{2\lambda}}\big) \notag\\
\le &C\min\big(\frac{1}{\widetilde{r}^2v^2},\frac{R^2}{\widetilde{r}^{2}v^2},\frac{R^2V^4}{\widetilde{r}^{2}\max(\widetilde{r},v)^2}+\frac{R^6}{\widetilde{r}^4}\big) .
\end{align}}

Let us define a suitable cover of $\mathbb{R}^6$, i.e. the collision classes, in order to apply Theorem \ref{Prop:LLN}. 
The classes are chosen finer as the collision strength becomes larger. 
If the particles keep distance of order 1 no splitting will be necessary.
Let therefore be $k,l\in\mathbb{Z},N\in\N\setminus \lbrace 1\rbrace,\delta>0$ and $0\leq r,v\leq 1$ and the family of sets given by
\begin{align} \label{Family of sets}
(i)\ &M_{(0,r)(0,v)}^{N}(Y)  &(ii) \  &M_{(0,r)(N^{l\delta}v,N^{N(l+1)\delta}v)}^{N}(Y)& \\ \nonumber
(iii) \ &M_{(0,r)(1,\infty)}^{N}(Y)  & (iv) \ &M_{(N^{k\delta}r,N^{N(k+1)\delta}r)(0,v)}^{N}(Y)& \\ \nonumber
(v) \ &M_{(N^{k\delta}r,N^{N(k+1)\delta}r)(N^{l\delta}v,N^{N(l+1)\delta}v)}^{N}(Y)& (vi) \ &M_{(N^{k\delta}r,N^{(k+1)\delta}r)(1,\infty)}^{N}(Y) \\\nonumber
(vii) \ &M_{(N^{-\delta}r,\infty)(0,\infty)}^{N}(Y), &
\end{align}
for $ 0\leq k\leq \lfloor\frac{\ln(\frac{1}{r})}{\delta\ln(N)}\rfloor,0\leq l\leq \lfloor\frac{\ln(\frac{1}{v})}{\delta\ln(N)}\rfloor$.
In this case we choose $r=v=N^{-\beta}$ and the number of sets belonging to this list is some integer $I_{\delta}$ independent of $N$.

We will apply \ref{Var(f)} for each collision class family and get the bounds
\begin{align*}
(i)\ &\frac{(N^{-\beta})^6}{(N^{-\beta})^4}=N^{-2\beta}  &(ii) \  &\frac{N^{-2\beta}N^{4\lbrack(k+1)\delta-\beta\rbrack}}{N^{-2\beta}N^{2(k\delta-\beta)}}=N^{-2\beta+2k\delta+4\delta} & \\
(iii) \ &\frac{(N^{-\beta})^2}{(N^{-\beta})^2}=1  & (iv) \ &\frac{N^{6(k\delta-\beta)}}{N^{4(k\delta-\beta)}}=N^{-2\beta+2k\delta+6\delta}& \\
(v) \ &\frac{N^{2(k\delta+\delta-\beta)}
N^{4(l\delta+\delta-\beta)}}{N^{2(k\delta-\beta)}N^{2(l\delta -\beta)}}+\frac{N^{6(k\delta-\beta)}}{N^{4(k\delta -\beta)}}=\mathrlap{ N^{-2\beta+2l\delta+6\delta}+N^{-2\beta+2k\delta+6\delta}} & \\
(vi) \ &\frac{(N^{k\delta}N^{-\beta})^2}{(N^{k\delta -\beta})^2}=N^{2\delta} & (vii) \ &\frac{1}{(N^{-\delta})^4}=N^{4\delta} 
\end{align*}
for $0\leq k,l\leq \lfloor\frac{\beta}{\delta}\rfloor$.
All these terms are bounded by $N^{6\delta}$.

For a law of large numbers argument we need
\begin{align*}
\|h_{1,N}\|_{\infty}\le C_0N^{1-\delta} \ \text{and } \
\max_{i \in I}\int_{M^N_i}h_{1,N}(X)^2k(X)d^6X\le C_0N^{1-\delta}.
\end{align*} 
Due to the estimates for each collision class it follows for all $i\in I$
\begin{align*}
& \int_{M^N_{(r_i,R_i),(v_i,V_i)}(Y)}h^t_{1,N}(Y,Z)^2k_0(Z)d^{6}Z\le CN^{2\alpha}N^{6\delta}
\le CN^{2 (3\delta+\alpha)} .
\end{align*} 
For $\delta>0$ small enough and due to the fact that $\alpha=\beta+\sigma$ it follows that $6\delta+2\alpha<1$ and the first assumption of Theorem \ref{Prop:LLN} is fulfilled as $\beta<\frac{1}{2}-3\delta.$\\
It holds due to Corollary \ref{corollary phi and  psi} that for a point in time $t_{min}$, where the mean-field particles are close
\begin{align}
& \int_{0}^t|f^{N}(\varphi^{1,N}_{s,0}(Y)-{\varphi^{1,N}_{s,0}}(Z))|\mathds 1_{G^N(Z)}(Y)ds \notag \\
\le & \min\Big(\frac{Ct}{|\varphi^{1,N}_{t_{min},0}(Y)-{\varphi^{1,N}_{t_{min},0}}(Z)|^2},\frac{CN^{\beta}}{| \varphi^{2,N}_{t_{min},0}(Y)-{\varphi^{2,N}_{t_{min},0}}(Z)|},\notag\\
&\frac{C}{|\varphi^{1,N}_{t_{min},0}(Y)-{\varphi^{1,N}_{t_{min},0}}(Z)|\cdot|\varphi^{2,N}_{t_{min},0}(Y)-{\varphi^{2,N}_{t_{min},0}}(Z)|}\Big)\mathds 1_{G^N(Z)}(Y) .
\end{align} 
This is where we break down the time integral into several parts. 
If $v$ is large, the assumptions of Theorem \ref{Prop:LLN} are fulfilled directly. 
If $v$ is small we made use of the fact that the collision time is not very large.
Remember the definition of the `good' set
\begin{align*}
G^N(Z)\coloneqq  \Big((M^N_{r_b,v_b}(Z)\setminus M^N_{ r_s,v_s}(Z)) \cup  M^N_{ r_s,v_s}(Z)\Big)^C.
\end{align*}
For $x_{min}\coloneqq  |\varphi^{1,N}_{t_{min},0}(Y)-{\varphi^{1,N}_{t_{min},0}}(Z)|$ and $v_{min}\coloneqq  |\varphi^{2,N}_{t_{min},0}(Y)-{\varphi^{2,N}_{t_{min},0}}(Z)|$ the following implication holds due to the definition of $G^N(Z)$ 
\begin{align}
 &x_{min}\leq N^{-r_s} \Rightarrow v_{min}\geq N^{-v_b}\label{case_1}\\
&N^{-s_r}\leq x_{min}\leq N^{-b_r}\Rightarrow v_{min}\geq N^{-b_v}\label{case_2}\\
&N^{-b_r}\leq x_{min}\Rightarrow v_{min}\in \mathbb{R}^+\label{case_3}
\end{align}
and thus the term is bounded in the first case \eqref{case_1} by
\begin{align*}
CN^{\beta+b_v}.
\end{align*}
for the second case \eqref{case_2}, the term is bounded by 
\begin{align*}
\min(CN^{\beta+b_v},CN^{s_r+b_v}).
\end{align*}
And for the last case \eqref{case_3} we get a bound of
\begin{align*}
CtN^{2b_r}.
\end{align*}         
As $\alpha=\beta+\sigma$ from Theorem \ref{Prop:LLN} the term is bounded by 
\begin{align*}
CtN^{2b_r}+CN^{\beta+b_v}.
\end{align*}
The second upper bound controls the cases where $x_{min}\leq 6N^{-b_r}$.
This yields for small enough $\sigma>0$ and $\beta+\alpha+b_v<1$ that 
\begin{align*}
||h_{1,N}^t(Y,\cdot)||_{\infty}\leq N^{\alpha}C (N^{2b_r}+N^{\beta+b_v})\leq CN^{1-\sigma}.
\end{align*}
We now apply our estimate on $h_{1,N}^t(y)$ defined in \eqref{h1,N} to control Term \ref{eq:good3}.
 Therefore we introduce the set $\mathcal{B}_{1,i}^{N,\sigma}\subset\mathbb{R}^{6N},i\in\lbrace 1,\hdots, N\rbrace$:
\begin{align}
\begin{split}
& X\in \mathcal{B}_{1,i}^{N,\sigma}\subseteq \mathbb{R}^{6N}\\ \label{set.B_1 good}
\Leftrightarrow & \exists t_1,t_2\in [0,T]:\\
& \Big|\frac{1}{N}\sum_{j\neq i}\int_{t_1}^{t_2}f^{N}(\varphi^{1,N}_{s,0}(X_i)-{\varphi^{1,N}_{s,0}}(X_j))\mathds 1_{G^N(X_i)}(X_j)ds\\
&  -\int_{\mathbb{R}^6}\int_{t_1}^{t_2} f^{N}(\varphi^N_{s,0}(X_i)-\varphi^N_{s,0}(Y))\mathds 1_{G^N(X_i)}(Y)dsk_0(Y)d^6Y\Big| >N^{-\alpha}= N^{-\beta-\sigma}.
\end{split}
\end{align}
The law of large numbers makes only statements for certain points in time. However, on very short time intervals fluctuations cannot change significantly since the force is bounded due to the cut off by $N^{2\beta}$.  This allows us to estimate fluctuations uniformly in time. 
By the definition of the set $\mathcal{B}_{1,i}^{N,\sigma}$ and by the fact that any continuous map $a:\mathbb{R}\rightarrow\mathbb{R}^m$ fulfills
\begin{align*}
& \big|\int_{t_1}^{t_2 }a(s)ds\big| = \big|\int_{0}^{t_2 }a(s)ds-\int_{0}^{t_1 }a(s)ds\big| \notag \\
\le & \big|\int_{0}^{\lfloor \frac{t_2}{\delta_N}\rfloor \delta_N }a(s)ds\big|+\int_{\lfloor \frac{t_2}{\delta_N}\rfloor \delta_N }^{t_2}|a(s)|ds +  \big|\int_{0 }^{\lfloor \frac{t_1}{\delta_N}\rfloor \delta_N}a(s)ds\big|+\int_{\lfloor \frac{t_1}{\delta_N}\rfloor \delta_N }^{t_1}|a(s)|ds \notag \\
\le & 2\max_{k\in \{0,...,\lfloor \frac{T}{\delta_N}\rfloor \}}\Big(\big|\int_{0 }^{k \delta_N}a(s)ds\big| +\int_{k \delta_N }^{(k+1) \delta_N }|a(s)|ds\Big),
\end{align*}
for $m\in\N,t_1,t_2\in\lbrack 0,T\rbrack$ it follows for $\delta_N>0$ that
\begin{align*}
& X\in \mathcal{B}_{1,i}^{N,\sigma} \notag  \\ 
\Rightarrow & \exists k\in \{0,...,\lfloor \frac{T}{\delta_N } \rfloor \}:  \notag \\
&  \Big(\big|\int_{0}^{k \delta_N }\Big(\frac{1}{N}\sum_{j\neq i}f^{N}(\varphi^{1,N}_{s,0}(X_i)-{\varphi^{1,N}_{s,0}}(X_j))\mathds 1_{G^N(X_i)}(X_j)  \notag \\
& -\int_{\mathbb{R}^6} f^{N}(\varphi^{1,N}_{s,0}(X_i)-{\varphi^{1,N}_{s,0}}(Y))\mathds 1_{G^N(X_i)}(Y)k_0(Y)
d^6Y\Big) ds\big| \geq \frac{N^{-\frac{5}{12}}}{4}\Big) \ \vee  \notag \\
& \Big(\int_{k \delta_N }^{(k+1) \delta_N }\Big(\big|\frac{1}{N}\sum_{j\neq i}f^{N}(\varphi^{1,N}_{s,0}(X_i)-{\varphi^{1,N}_{s,0}}(X_j))\mathds 1_{G^N(X_i)}(Y)
\big| \notag \\
& +\big|\int_{\mathbb{R}^6} f^{N}(\varphi^{1,N}_{s,0}(X_i)-{\varphi^{1,N}_{s,0}}(Y))\mathds 1_{G^N(X_i)}(Y)k_0(Y)
d^6Y\big|\Big) ds \geq \frac{N^{-\frac{5}{12}}}{4}\Big)
\end{align*} 
If we choose $\delta_N\coloneqq  \frac{N^{-\alpha}}{8||f_N||_{\infty}}\leq CN^{-\alpha-\beta\lambda}=N^{-\alpha-2\beta}=N^{-3\beta-\sigma}$ the second constraint of the assumption is true. 
For the current estimate we assumed that all particles form a single cluster because it is sufficient for our estimates. We could choose $\delta_N$ of much larger order.

According to the previous reasoning for at least one $k\in\lbrace 0,\hdots,\lfloor\frac{T}{\delta_N}\rfloor\rbrace$ the event related to the first constraint must occur if $X\in\mathcal{B}_{1,i}^{N,\sigma}$, but the law of large numbers tells us that for any of these events and any $\gamma>0$ there exists a $C_{\gamma}>0$ such that its probability is smaller than $C_{\gamma}N^{-\gamma}$ since $h_{1,N}^t(Y,\cdot)$ fulfils the assumptions of Theorem \ref{Prop:LLN}.

As $\beta=\frac{5}{12}-\sigma$ and $\alpha=\beta+\sigma$ the number of such events is bounded by $$\lfloor\frac{T}{\delta_N}\rfloor +1\le CN^{\alpha+2\beta} \leq CN^{\sigma+3\beta}=CN^{\frac{5}{4}+\sigma}$$ and thus it holds for all $N\in\N$ that
\begin{align*}
\mathbb{P}(\exists i\in\lbrace 1,\hdots,N\rbrace:X\in\mathcal{B}_{1,i}^{N\sigma})
&\leq N \mathbb{P}(X\in\mathcal{B}_{1,i}^{N\sigma}))\\
&\leq N \left(CN^{\frac{5}{4}}(C_{\gamma+\frac{9}{4}}N^{-(\gamma+\frac{9}{4})})\right)\\
&\leq C_{\gamma}N^{-\gamma}.
\end{align*}
For typical initial data and large enough $N\in\N$ Term \ref{eq:good3} stays smaller than $N^{-\frac{5}{12}+\sigma}$.

\subsubsection{Estimate of Term \ref{eq:good2}}\label{Term 2 Sec}
Let us estimate Term \ref{eq:good2}, i.e. the difference of the real force acting on the real particles and the real force acting on the mean-field particles
\begin{align*}
 &\left|\int_{t_1}^t\frac{1}{N}\sum_{j\neq i}\left( f^{N}\left([\Psi^{1,N}_{s,0}(X)]_i-[\Psi^{1,N}_{s,0}(X)]_j\right)\right.\right.\\
&\left.\left.\qquad\qquad\quad-f^{N}(\varphi^{1,N}_{s,0}(X_i)-{\varphi^{1,N}_{s,0}}(X_j))\right)\mathds 1_{G^N(X_i)}(X_j) ds\right|. 
\end{align*} 
We abbreviate the following notation for the allowed difference between mean-field particle and the real one, depending on the subset membership. 
We allow less control if the particle is bad but have strict requirements if the particle is good.
$\Delta^N_g(t,X)$ describes the largest spatial deviation of the `good' particles, $\Delta^N_b(t,X)$ the corresponding value for the `bad' ones and $\Delta^N_{sb}(t,X)$ the corresponding value for the `superbad' ones. 
The worse the subset (in the sense of `bad' or `superbad'), the more deviation is allowed.
\begin{align*}
 \Delta^N_g(t,X)&\coloneqq  \max_{j\in \mathcal{M}^N_g(X)}\sup_{0\le s\le t}\left|\left[\Psi^{1,N}_{s,0}(X)\right]_j-{\varphi^{1,N}_{s,0}}(X_j)\right|=N^{-\frac{5}{12}+\sigma} \notag \\
 \Delta^N_b(t,X)&\coloneqq  \max_{j\in \mathcal{M}^N_b(X)}\sup_{0\le s\le t}\left|\left[\Psi^{1,N}_{s,0}(X)\right]_j-{\varphi^{1,N}_{s,0}}(X_j)\right|=N^{-\frac{7}{24}-\sigma}  \notag \\
 \Delta^N_{sb}(t,X)&\coloneqq  \max_{j\in \mathcal{M}^N_{sb}(X)}\sup_{0\le s\le t}\left|\left[\Psi^{1,N}_{s,0}(X)\right]_j-{\varphi^{1,N}_{s,0}}(X_j)\right| =N^{-\frac{1}{6}-\sigma}. \notag 
\end{align*}
We further introduce a subset of the good particles $$\widetilde{G}^N(\cdot)\coloneqq  G^N(\cdot)\cap \big(M^N_{3N^{-\frac{1}{2}+\sigma},\infty}(\cdot)\big)^C$$ which helps us to shorten the upcoming estimates.
By definition of $\widetilde{G}^N(\cdot)$ (applied for the first inequality) and the stopping time $\tau^N(X)$
\begin{align*}
\tau_g^N\coloneqq  \sup\left\lbrace t\in \lbrack 0,T\rbrack: \max_{i\in \mathcal{M}_{g}^{N}}\sup_{0\leq s\leq t}\left|\left[ \Psi_{s,0}^N(X)\right]_i-\varphi_{s,0}^N(X_i)\right|\leq \delta_g^N\right\rbrace\\
\tau_b^N\coloneqq  \sup\left\lbrace t\in \lbrack 0,T\rbrack: \max_{i\in \mathcal{M}_{b}^{N}}\sup_{0\leq s\leq t}\left|\left[ \Psi_{s,0}^N(X)\right]_i-\varphi_{s,0}^N(X_i)\right|\leq \delta_b^N\right\rbrace\\
\tau_{sb}^N\coloneqq  \sup\left\lbrace t\in \lbrack 0,T\rbrack: \max_{i\in \mathcal{M}_{sb}^{N}}\sup_{0\leq s\leq t}\left|\left[ \Psi_{s,0}^N(X)\right]_i-\varphi_{s,0}^N(X_i)\right|\leq \delta_{sb}^N\right\rbrace
\end{align*}
as well as $\tau^N(X)\coloneqq \min(\tau_g^N(X),\tau_b^N(X),\tau_{sb}^N(X))$ with $\delta_g^N=N^{-\beta}=N^{-\frac{5}{12}+\sigma}$, $\delta_b^N=N^{-d_b}=N^{-\frac{7}{24}-\sigma}$ and $\delta_{sb}^N=N^{-d_{sb}}=N^{-\frac{1}{6}-\sigma}$ it holds for $X_j\in \widetilde{G}^N(X_i)$ and times $s\in [0,\tau^N(X)]$ that
\begin{align*}
& \max\big( 2N^{-\beta},\frac{2}3|{\varphi^{1,N}_{s,0}}(X_j)-{\varphi^{1,N}_{s,0}}(X_i)|\big)\geq  \max \big( 2N^{-\beta},2N^{-\frac{1}{2}+\sigma}\big) \geq 2\Delta^N_g(t,X).
\end{align*}
In the next step we subdivide the sum according to whether the particle interacting with $i$ is itself `superbad', `bad' or `good'.
Furthermore, the map $g^N$ was defined such that $|f^N(q+\delta)-f^N(q)|\le g^N(q)|\delta|$ for $q,\delta \in \mathbb{R}^3$ where $ \max\big(2 N^{-\beta},\frac{2}{3}|q|\big)\geq |\delta|$, see Definition \ref{force g}.
Thus the subsequent estimates are fulfilled for all times $0\le t_1\le t\le \tau^N(X)$.
\begin{align}
& \big|\int_{t_1}^t\Big(\frac{1}{N}\sum_{j\neq i}\Big(f^{N}([\Psi^{1,N}_{s,0}(X)]_j-[\Psi^{1,N}_{s,0}(X)]_i)  \notag  \\
&-f^{N}(\varphi^{1,N}_{s,0}(X_j)-{\varphi^{1,N}_{s,0}}(X_i))\Big)\mathds 1_{G^N(X_i)}(X_j)\Big) ds\big| \label{term0} \\
\le & \int_{0}^t\Big(\frac{1}{N}\sum_{\substack{j\neq i\\ j\in \mathcal{M}^N_{sb}(X)}}\Big(\big|f^{N}([\Psi^{1,N}_{s,0}(X)]_j-[\Psi^{1,N}_{s,0}(X)]_i) \notag  \\
&- f^{N}(\varphi^{1,N}_{s,0}(X_j)-{\varphi^{1,N}_{s,0}}(X_i))\big|\Big)\mathds 1_{G^N(X_i)}(X_j)\Big) ds\label{term a}\\
+ & \int_{0}^t\Big(\frac{1}{N}\sum_{\substack{j\neq i\\ j\in \mathcal{M}^N_{b}(X)}}\Big(\big|f^{N}([\Psi^{1,N}_{s,0}(X)]_j-[\Psi^{1,N}_{s,0}(X)]_i) \notag  \\
&- f^{N}(\varphi^{1,N}_{s,0}(X_j)-{\varphi^{1,N}_{s,0}}(X_i))\big|\Big)\mathds 1_{G^N(X_i)}(X_j)\Big) ds\label{term b}\\
&+\int_{0}^t\Big(\frac{1}{N}\sum_{\substack{j\neq i\\ j\in \mathcal{M}^N_g(X)}}\Big(\big|f^{N}([\Psi^{1,N}_{s,0}(X)]_j-[\Psi^{1,N}_{s,0}(X)]_i)\big| \notag  \\
&+ \big|f^{N}(\varphi^{1,N}_{s,0}(X_j)-{\varphi^{1,N}_{s,0}}(X_i))\big|\Big)\mathds 1_{G^N(X_i)\cap M^N_{3N^{-\frac{1}{2}+\sigma},\infty}(X_i)}(X_j)\Big) ds \label{term c}\\
 & + \int_{0}^t\frac{2}{N}\sum_{\substack{j\neq i\\ j\in \mathcal{M}^N_g(X)}}g^{N}(\varphi^{1,N}_{s,0}(X_j)-{\varphi^{1,N}_{s,0}}(X_i))\Delta^N_g(s,X)\mathds 1_{\widetilde{G}^N(X_i)}(X_j) ds. \label{term c2}
\end{align}
For the last term we applied the previous considerations and to estimate this one we define a set 
\begin{align}
\begin{split}
& X\in \mathcal{B}_{2,i}^{N,\sigma}\subseteq \mathbb{R}^{6N} \label{def.B_2 good}\\
\Leftrightarrow & \exists t_1,t_2\in [0,T]:\\
& \Big|\frac{1}{N}\sum_{j\neq i}\int_{t_1}^{t_2}g^{N}({\varphi^{1,N}_{s,0}}(X_j)-{\varphi^{1,N}_{s,0}}(X_i))\mathds 1_{\widetilde{G}^N(X_i)}(X_j)ds\\
&  -\int_{\mathbb{R}^6}\int_{t_1}^{t_2} g^{N}(\varphi^{1,N}_{s,0}(Y)-{\varphi^{1,N}_{s,0}}(X_i))\mathds 1_{\widetilde{G}^N(X_i)}(Y)dsk_0(Y)
d^6Y\Big| > 1
\end{split} 
\end{align}
For $Y,Z\in \mathbb{R}^6$ it holds by definition of $\widetilde{G}^N(\cdot)$ and the definition of $g^N$ (see \ref{force g}) that
\begin{align}
& \int_{0}^t g^{N}(\varphi^{1,N}_{s,0}(Y)-{\varphi^{1,N}_{s,0}}(Z)) \mathds 1_{\widetilde{G}^N(Z)}(Y)ds \notag \\
\leq  & C
N^{\beta}\int_{0}^t|f^{N}(\varphi^{1,N}_{s,0}(Y)-{\varphi^{1,N}_{s,0}}(Z))|\mathds 1_{G^N(Z)}(Y)ds\\
\leq  & C
N^{\frac{5}{12}-\sigma}\int_{0}^t|f^{N}(\varphi^{1,N}_{s,0}(Y)-{\varphi^{1,N}_{s,0}}(Z))|\mathds 1_{G^N(Z)}(Y)ds. \label{check.l.o.l.n for g}
\end{align}
Analogously to the previous section, Term \ref{check.l.o.l.n for g} fulfils the assumptions of Theorem \ref{Prop:LLN}. 
Following the same reasoning for the map $h^t_N(Y,\cdot)$ one can show that for an arbitrary $\gamma>0$ there exists a $C_{\gamma}>0$ such that for all $N\in \N$
\begin{align}
\mathbb{P}\big(\exists i\in \{1,...,N\}:X\in \mathcal{B}_{2,i}^{N,\sigma}\big)\le C_{\gamma} N^{-\gamma}. \label{prob.b.2}
\end{align}
It remains to determine an upper bound for the terms \eqref{term c}, \eqref{term a} and \eqref{term b}.\\
We start with the last two terms, which describe the interaction of a good particle with a superbad particle respectively bad one. 
We show that the `superbad' and `bad' particles do typically not infect the `good' ones which corresponds to deriving a suitable bound for Term \eqref{term a} and \eqref{term b}. 
Since the allowed maximal value for for the largest deviation of a `bad' or `superbad'  particle $\Delta^N_b(t,X)$ and $\Delta^N_s(t,X)$ is distinctly larger than the corresponding value for the good particle $\Delta^N_g(t,X)$, problems could arise if the number of `bad' or `superbad'  particles coming close to a `good' one exceeds a certain value. 
But we can show that the probability of such events is sufficiently small for large $N$.  

Analogously to the previous section we introduce $h^t_{2,N}(Y,\cdot)$ according to Theorem \ref{Prop:LLN} with
\begin{align}
h_{2,N}^t(y,\cdot):\mathbb{R}^6\rightarrow\mathbb{R}^3,Z\mapsto N^{\alpha}\int_{0}^{t}f^{N}(\varphi^{1,N}_{s,0}(Y)-{\varphi^{1,N}_{s,0}}(Z))ds \mathds 1_{G^N(Y)}(Z).
\end{align} 
Let us also implement a family of `collision classes' $\big(M^N_{(r_i,R_i),(v_i,V_i)}(Y)\big)_{i \in I_\delta}$ which covers $\mathbb{R}^6$ and check if $h^t_{2,N}(Y,\cdot)$ in combination with this cover fulfils the assumptions of Theorem \ref{Prop:LLN} to derive an upper bound for the terms \eqref{term a} and \eqref{term b}.  
Similar to the list stated in \eqref{Family of sets} we define $\big(M^N_{(r_i,R_i),(v_i,V_i)}(Y)\big)_{i \in I_{\delta}}$ for the parameters $r\coloneqq  r_b=6N^{-\frac{7}{24}-\sigma}$ and $v\coloneqq  6v_b=6N^{-\frac{1}{6}}$ for Term \eqref{term b} and for the parameters $r\coloneqq  r_s=6N^{-\frac{1}{3}-\sigma}$ and $\ v\coloneqq  6v_s=6N^{-\frac{5}{18}}$ for Term \eqref{term a} (instead of $r=v\coloneqq  N^{-c}$ and $\delta\coloneqq  \sigma $ like in \eqref{Family of sets}). 
Thus we define for $i\in  \{1,...,N\}$ the sets $ \mathcal{B}_{3_b,i}^{N,\sigma}, \mathcal{B}_{3_s,i}^{N,\sigma}\subseteq \mathbb{R}^{6N}$ as follows
\begin{align}
\begin{split}
& X\in \mathcal{B}_{3_b,i}^{N,\sigma} \subseteq \mathbb{R}^{6N} \\
\Leftrightarrow & \exists l\in I_{\sigma}: \Big(R_l\neq \infty\ \land\\ &\sum_{j \in \mathcal{M}^N_b(X)}\mathbf{1}_{M^N_{(r_l,R_l),(v_l,V_l)}(X_i)}(X_j) \geq  N^{\sigma\frac{3}{4}}\big\lceil N^{\frac{3}{4}} R_l^2\min\big(\max(V_l,R_l),1\big)^4\big\rceil\Big) \ \vee \label{def.B_3b}\\
& \sum_{j \in \mathcal{M}^N_b(X)}1=| \mathcal{M}^N_b(X)|\geq N^2v_{b}^4r_{b}^2  \geq N^{\frac{3}{4}(1+\sigma)}.
\end{split}
\end{align} 
Respectively for Term \eqref{term a}
\begin{align}
\begin{split}
& X\in \mathcal{B}_{3_{sb},i}^{N,\sigma} \subseteq \mathbb{R}^{6N} \\
\Leftrightarrow & \exists l\in I_{\sigma}: \Big(R_l\neq \infty\ \land\\ &\sum_{j \in \mathcal{M}^N_{sb}(X)}\mathbf{1}_{M^N_{(r_l,R_l),(v_l,V_l)}(X_i)}(X_j) \geq  N^{\sigma(\frac{2}{9})}\big\lceil N^{\frac{2}{9}} R_l^2\min\big(\max(V_l,R_l),1\big)^4\big\rceil\Big) \ \vee \label{def.B_3sb}\\
& \sum_{j \in \mathcal{M}^N_{sb}(X)}1=| \mathcal{M}^N_{sb}(X)|\geq N^2v_{sb}^4r_{sb}^2\geq N^{\frac{2}{9}(1+\sigma)}.
\end{split}
\end{align}
The last line in each case gives an estimate of the absolute number of bad or superbad particles and the line above an estimate of how many bad or superbad particles come close to a good one given a certain inter-particle distance and velocity.
We now derive an upper bound for Term \eqref{term a} and \eqref{term b} under the condition that $X\in \big(\mathcal{B}_{3_{sb},i}^{N,\sigma}\big)^C$ respectively $X\in \big(\mathcal{B}_{3_{b},i}^{N,\sigma}\big)^C$ and prove later that $\mathbb{P}\big( X\in \mathcal{B}_{3_{sb},i}^{N,\sigma}\big)$ and $\mathbb{P}\big( X\in \mathcal{B}_{3_{b},i}^{N,\sigma}\big)$ get small as $N$ increases.\\ 
To this end, we abbreviate for $0\le r\le R $ and $0\le v \le V$
$$ \widetilde{M}^N_{(r,R),(v,V)}(X_i)\coloneqq  G^N(X_i)\cap M^N_{(r,R),(v,V)}(X_i)$$
to distinguish between the collision classes.
As mentioned before, for Term \eqref{term a} we only consider values of $r$ and $R$ that satisfy the constraint
\begin{align}
\big(r=0 \land R=6\delta^N_{sb}=6N^{-\delta_s}\big)\vee \big( r\geq 6\delta^N_{sb}\land R=N^{\sigma}r\big), \label{cond.para.r,R superbad}
\end{align}
respectively for Term \eqref{term b}
\begin{align}
\big(r=0 \land R=6\delta^N_b=6N^{-\delta_b}\big)\vee \big( r\geq 6\delta^N_b\land R=N^{\sigma}r\big). \label{cond.para.r,R bad} 
\end{align}
We will see in Section \ref{2bgoodbad} that those are the worst case options for the estimates.
Recall that $$\sup_{0\le s \le t}|\Psi^N_{s,0}(X)-\Phi^N_{s,0}(X)|_{\infty}\le N^{-s_{\delta}}=\delta^N_{sb}=N^{-\frac{1}{6}-\sigma}$$ and $$\sup_{0\le s \le t}|\Psi^N_{s,0}(X)-\Phi^N_{s,0}(X)|_{\infty}\le N^{-b_{\delta}}=\delta^N_b=N^{-\frac{7}{24}-\sigma}$$ depending on which of the two term we devote ourselves to and for times before the stopping time is `triggered'. 
Thus, we obtain for $0\le t\le \tau^N(X)$ depending on the choice of $r$ that Term \eqref{term a} can be estimated by
{\allowdisplaybreaks \begin{align*}
 & \int_{0}^t\frac{1}{N}\sum_{\substack{j\neq i\\ j\in \mathcal{M}^N_{sb}(X)}}\Big(\big|f^{N}([\Psi^{1,N}_{s,0}(X)]_j-[\Psi^{1,N}_{s,0}(X)]_i)\notag  \\
&- f^{N}(\varphi^{1,N}_{s,0}(X_j)-{\varphi^{1,N}_{s,0}}(X_i))\big|\Big) \mathbf{1}_{\widetilde{M}^N_{(r,R),(v,V)}(X_i)}(X_j) ds \\
\le & \int_{0}^t\frac{1}{N}\sum_{\substack{j\neq i\\ j\in \mathcal{M}^N_{sb}(X)}} \Big(\big|f^{N}([\Psi^{1,N}_{s,0}(X)]_j-[\Psi^{1,N}_{s,0}(X)]_i)\big|\notag  \\
&+ \big|f^{N}(\varphi^{1,N}_{s,0}(X_j)-{\varphi^{1,N}_{s,0}}(X_i))\big|\Big)\mathbf{1}_{\widetilde{M}^N_{(r,R),(v,V)}(X_i)}(X_j) ds \mathbf{1}_{[0,6\delta^N_{sb}]}(r)\notag \\
& +\frac{2}{N}\Delta^N_{sb}(t,X)\sup_{Y\in \widetilde{M}^N_{(r,R),(v,V)}(X_i)}\int_0^t g^{N}(\varphi^{1,N}_{s,0}(Y)-{\varphi^{1,N}_{s,0}}(X_i))ds  \notag \\ 
& \cdot  \sum_{\substack{j\neq i\\ j\in \mathcal{M}^N_b(X)}} \mathbf{1}_{\widetilde{M}^N_{(r,R),(v,V)}(X_i)}(X_j)\mathbf{1}_{[6\delta^N_{sb},\infty)}(r).
\end{align*}
Analogously Term \eqref{term b} can be estimated by
{\allowdisplaybreaks \begin{align*}
 & \int_{0}^t\frac{1}{N}\sum_{\substack{j\neq i\\ j\in \mathcal{M}^N_b(X)}}\Big(\big|f^{N}([\Psi^{1,N}_{s,0}(X)]_j-[\Psi^{1,N}_{s,0}(X)]_i)\notag  \\
&- f^{N}(\varphi^{1,N}_{s,0}(X_j)-{\varphi^{1,N}_{s,0}}(X_i))\big|\Big) \mathbf{1}_{\widetilde{M}^N_{(r,R),(v,V)}(X_i)}(X_j) ds \\
\le & \int_{0}^t\frac{1}{N}\sum_{\substack{j\neq i\\ j\in \mathcal{M}^N_b(X)}} \Big(\big|f^{N}([\Psi^{1,N}_{s,0}(X)]_j-[\Psi^{1,N}_{s,0}(X)]_i)\big|\notag  \\
&+ \big|f^{N}(\varphi^{1,N}_{s,0}(X_j)-{\varphi^{1,N}_{s,0}}(X_i))\big|\Big)\mathbf{1}_{\widetilde{M}^N_{(r,R),(v,V)}(X_i)}(X_j) ds \mathbf{1}_{[0,6\delta^N_b]}(r)\notag \\
& +\frac{2}{N}\Delta^N_b(t,X)\sup_{Y\in \widetilde{M}^N_{(r,R),(v,V)}(X_i)}\int_0^t g^{N}(\varphi^{1,N}_{s,0}(Y)-{\varphi^{1,N}_{s,0}}(X_i))ds  \notag \\ 
& \cdot  \sum_{\substack{j\neq i\\ j\in \mathcal{M}^N_b(X)}} \mathbf{1}_{\widetilde{M}^N_{(r,R),(v,V)}(X_i)}(X_j)\mathbf{1}_{[6\delta^N_b,\infty)}(r) 
\end{align*}
where we utilzied that $|f^N(q+\delta)-f^N(q)|\le g^N(q)|\delta|$ for $q,\delta \in \mathbb{R}^3$ provided that $ \max\big(2 N^{-c},\frac{2}{3}|q|\big)\geq |\delta|$. 

Application of Corollary \ref{corollary phi and  psi} yields that the previous terms are bounded by
\begin{align}
\eqref{term a}\leq & \frac{C}{N}\frac{1}{N^{- \beta}v}\sum_{\substack{j\neq i\\ j\in \mathcal{M}^N_{sb}(X)}} \mathbf{1}_{\widetilde{M}^N_{(r,R),(v,V)}(X_i)}(X_j)\mathbf{1}_{[0,6\delta^N_{sb}]}(r)\notag \\
&+ \frac{C}{N}\frac{\Delta_b^N(t,X)}{r^{2}\max(r,v)}\sum_{\substack{j\neq i\\ j\in \mathcal{M}^N_{sb}(X)}} \mathbf{1}_{\widetilde{M}^N_{(r,R),(v,V)}(X_i)}(X_j)\mathbf{1}_{[6\delta^N_{sb},\infty)}(r) . \label{bound 2a}
\end{align}}
and
\begin{align}
\eqref{term b}\leq & \frac{C}{N}\frac{1}{N^{- \beta}v}\sum_{\substack{j\neq i\\ j\in \mathcal{M}^N_b(X)}} \mathbf{1}_{\widetilde{M}^N_{(r,R),(v,V)}(X_i)}(X_j)\mathbf{1}_{[0,6\delta^N_b]}(r)\notag \\
&+ \frac{C}{N}\frac{\Delta_b^N(t,X)}{r^{2}\max(r,v)}\sum_{\substack{j\neq i\\ j\in \mathcal{M}^N_b(X)}} \mathbf{1}_{\widetilde{M}^N_{(r,R),(v,V)}(X_i)}(X_j)\mathbf{1}_{[6\delta^N_b,\infty)}(r) . \label{bound 2b}
\end{align}}
\subsubsection{Estimate of Term \ref{term b} (i good j bad)}\label{2bgoodbad}
Remark that the assumptions of the Corollary \ref{corollary phi and  psi} are indeed fulfilled in the current situation since according to the constraints on the possible parameters (see \eqref{cond.para.r,R bad}) $r\in [0,6\delta^N_b]$ implies $R=\delta^N_b$ and $r=0$.
Considering the definition of the set of `good' particles $ G^N(X_i)$ it follows that
$$\widetilde{M}^N_{(0,6\delta^N_b),(v,V)}(X_i)=M^N_{(0,6\delta^N_b),(v,V)}(X_i)\cap G^N(X_i)\subseteq \big(M^N_{6\delta^N_b,N^{-\frac{1}{6}}}(X_i)\big)^C$$ which in turn provides
\begin{align}
& X_j\in \widetilde{M}^N_{(0,6\delta^N_b),(v,V)}(X_i) \notag\\
\Rightarrow
& |\varphi^{2,N}_{t_{min},0}(X_j)-{\varphi^{2,N}_{t_{min},0}(X_i)}|\geq N^{-\frac{1}{6}}  \label{rel.values.v_b}
\end{align}  
where $t_{min}$ shall denote a point in time where $|\varphi^{1,N}_{\cdot,0}(X_j)-{\varphi^{1,N}_{\cdot,0}(X_i)}|$ takes its minimum on $[0,T]$.\\ 
Now we want to derive an upper bound for Term \ref{bound 2b} under the condition that
\begin{align*}
& \sum_{j \in \mathcal{M}_b^N(X)}\mathbf{1}_{M^N_{(r,R),(v,V)}(X_i)}(X_j) \le   N^{\frac{3\sigma}{4}}\big\lceil N^{\frac{3}{4}} R^2\min\big(\max(V,R),1\big)^4\big\rceil.
\end{align*}
We will deal with the addends related to $\mathbf{1}_{[0,6\delta^N_b]}(r)$ and $\mathbf{1}_{ [6\delta^N_b, \infty)}(r)$ separately. 
Regarding the first addend, we already discussed that $r=0$ and $R=6\delta^N_b$ due to condition \eqref{cond.para.r,R bad}. 
We obtain
\begin{align}
&\frac{C}{N}\frac{1}{N^{-\beta}\Delta v}\sum_{\substack{j\neq i\\ j\in \mathcal{M}^N_{b}(X)}}\mathbf{1}_{\widetilde{M}^N_{(r,R),(v,V)}(X_i)}(X_j)\mathbf{1}_{[0,6\delta^N_{b}]}(r)\\
&\leq\frac{N^{\frac{3\sigma}{4}}}{N^{\beta+1}\max( v,\delta_b)}+\frac{C|\mathcal{M}_{b}|R^2\min(V,1)^4}{N^{-\beta-1}\max( v,\delta_b)}\\
&\leq\frac{N^{\frac{3\sigma}{4}}}{N^{\beta+1}\max( v,\delta_b)}+\frac{CR^2 \min(V,1)^4N^{\frac{3}{4}(1+\sigma)}}{N^{-\beta+1}\max(N^{-b_\delta},v)}\\
&\leq CN^{-3\sigma-\frac{5}{12}}+ C N^{-\frac{5}{12}-3\sigma}\frac{ \min(V,1)^4}{\max(N^{-b_{\delta}},v)}\label{eq:T2b1.estimated}
\end{align} 
for $R=\delta^N_b=N^{-b_{\delta}}=N^{-\frac{7}{24}-\sigma}$
since we only have to consider values with $v>N^{-b_v}=N^{-\frac{1}{6}}$, see \eqref{rel.values.v_b}. 

For the allowed deviation $\Delta^N_b(t,X)\leq N^{-b_{\delta}}= \delta^N_b=N^{-\frac{7}{24}-\sigma}$ and $R=N^{\sigma}r$ for $r\geq 6\delta^N_b$ (see \eqref{cond.para.r,R bad}) it follows for the second term of \eqref{bound 2b} that
\begin{align}
& \frac{C}{N}\frac{\Delta_{b}^N(t,X)}{r^{2}\max(r,v)}\sum_{\substack{j\neq i\\ j\in \mathcal{M}^N_b(X)}} \mathbf{1}_{\widetilde{M}^N_{(r,R),(v,V)}(X_i)}(X_j) \notag \\
\le & \frac{C}{N}\Big( \frac{N^{-2b_{r}-4b_{v}+2+\sigma}R^2\min\big(\max(V,R),1\big)^4}{r^2 \max(r,v)}  +\frac{N^{\frac{3\sigma}{4}}}{r^2 \max(r,v)}\Big) N^{-b_{\delta}} \notag \\
\le & C\Big( \frac{\min\big(\max(V,R),1\big)^4}{\max(r,v)} N^{-\frac{13}{24}}+N^{-\frac{41}{48}+3\sigma}\Big)\le CN^{-\frac{13}{24}} \leq CN^{-\frac{5}{12}}.\label{eq:T2b2.estimated}
\end{align}
In total we got an upper bound for Term \eqref{term b}. 
All sets belonging to the family $\big(M^N_{(r_i,R_i),(v_i,V_i)}(Y)\big)_{i \in I_\sigma}$ are contained in a `collision class' which takes one of the subsequent forms for suitable parameter $r,v\in [0,1]$
\begin{multicols}{2}
\begin{itemize}
\item[(i)] $ M^N_{(0,6\delta^N_b),(0,6\delta^N_b)}(Y)$
\item[(ii)]$M^N_{(0,6\delta^N_b),(v,N^{\sigma}v)}(Y) $
\item[(iii)]  $M^N_{(0,6\delta^N_b),(1,\infty)}(Y)$
\item[(iv)] $M^N_{(r,N^{\sigma}r),(0,6\delta^N_b)}(Y)$
 \item[(v)] $M^N_{(r,N^\sigma r),(v,N^\sigma v)}(Y)$
\item[(vi)] $  M^N_{(r,N^\sigma r),(1,\infty)}(Y)$,
\end{itemize}
\end{multicols}
\noindent 
except for $M^N_{(N^{-\sigma},\infty),(0,\infty)}(Y)$, which will be considered separately.
Recall that the number of `collision classes' belonging to the cover $|I_{\sigma}|$ is independent of $N$, analogously to Section \ref{Term 2 Sec}.
By comparing the possible values of $r$, $R$, $v$, and $V$ with the estimates \eqref{eq:T2b1.estimated} and \eqref{eq:T2b2.estimated}, it is evident that if $X\in \big(\mathcal{B}^{N,\sigma}_{3,i}\big)^C$ and $\sigma>0$ is chosen sufficiently small for the relevant terms, a set of type (ii), (iv), or (v) with $v=N^{-\sigma}$ or $r=N^{-\sigma}$ results in the 'worst-case scenario.' Consequently, the overall expression for Term \eqref{term b} can be bounded as follows:
\begin{align}
& CN^{-\frac{5}{12}} .\label{bound2gb}
\end{align}
The class where the previous general considerations can not be applied, $M^N_{(N^{-\sigma},\infty),(0,\infty)}(Y)$, the following holds:
\begin{align}
& \int_{0}^t\frac{1}{N}\sum_{\substack{j\neq i\\ j\in \mathcal{M}^N_{b}(X)}}\Big(\big|f^{N}([\Psi^{1,N}_{s,0}(X)]_j-[\Psi^{1,N}_{s,0}(X)]_i)\notag\\
&  - f^{N}(\varphi^{1,N}_{s,0}(X_j)-{\varphi^{1,N}_{s,0}}(X_i))\big|\Big) \mathbf{1}_{M^N_{(N^{-\sigma},\infty),(0,\infty)}(X_i)}(X_j) ds \notag\\
\le & \frac{2}{N}\sup_{Y\in M^N_{(N^{-\sigma},\infty),(0,\infty)}(X_i)}\int_0^t g^{N}(\varphi^{1,N}_{s,0}(Y)-{\varphi^{1,N}_{s,0}}(X_i))ds  \notag   \sum_{\substack{j\neq i\\ j\in \mathcal{M}^N_{b}(X)}} \underbrace{\Delta^N_b(t,X)}_{\leq N^{-b_\delta}}\notag \\
\le & \frac{2}{N} \big(T\frac{C}{(N^{-\sigma})^3}\big)N^{-b_{\delta}}\underbrace{|\mathcal{M}^N_{b}(X)|}_{\le  N^{2-2b_r-4b_v(1+\sigma)}} \notag \\
&\le  CN^{1-2b_r-4b_v-b_{\delta}+c\sigma}\notag \\
&\le CN^{-\frac{13}{24}+\frac{3}{4}\sigma}\notag
\end{align}
for $X\in \big(\mathcal{B}^{N,\sigma}_{3,i}\big)^C$  and $t\le \tau^N(X).$
\subsubsection{Estimate of Term \ref{term a} (i good j superbad)}\label{2bgoodsuperbad}
The estimates on Term \ref{term b} are quite similar to the previous one, except that now $j\in\mathcal{M}^N_{sb}(X).$

We get for times $0\leq t\le \tau^N(X)$ the following $r$-depending estimate
{\allowdisplaybreaks \begin{align}
 & \int_{0}^t\frac{1}{N}\sum_{\substack{j\neq i\\ j\in \mathcal{M}^N_{sb}(X)}}\Big(\big|f^{N}([\Psi^{1,N}_{s,0}(X)]_j-[\Psi^{1,N}_{s,0}(X)]_i)\notag \\
&- f^{N}(\varphi^{1,N}_{s,0}(X_j)-{\varphi^{1,N}_{s,0}}(X_i))\big|\Big) \mathbf{1}_{\widetilde{M}^N_{(r,R),(v,V)}(X_i)}(X_j) ds\\
&\le\frac{C}{N}\frac{1}{N^{-\beta}\Delta v}\sum_{\substack{j\neq i\\ j\in \mathcal{M}^N_{s}(X)}}\mathbf{1}_{\widetilde{M}^N_{(r,R),(v,V)}(X_i)}(X_j)\mathbf{1}_{[0,6\delta^N_{s}]}(r)\\
&+\frac{C \Delta^N_{sb}(t,X)}{N \max(r,v)r^2}\sum_{\substack{j\neq i\\ j\in \mathcal{M}^N_{sb}(X)}}\mathbf{1}_{\widetilde{M}^N_{(r,R),(v,V)}(X_i)}(X_j)\mathbf{1}_{[6\delta^N_{s},\infty]}(r).
\end{align}
For the first summand, $r_s\coloneqq  N^{-\frac{1}{3}-\sigma},v_s\coloneqq  N^{-\frac{5}{18}}$ and $\delta_s=N^{-\frac{1}{6}}$ and in view of the definition of $ G^N(X_i)$ it follows that
$$\widetilde{M}^N_{(0,6\delta^N_s),(v,V)}(X_i)=M^N_{(0,6\delta^N_s),(v,V)}(X_i)\cap G^N(X_i)\subseteq \big(M^N_{6\delta^N_s,N^{-\frac{1}{6}}}(X_i)\big)^C,$$ were $N^{-\frac{1}{6}}$ is the velocity cut off of the bad particles not the superbad ones.
This provides us the necessary implication
\begin{align}
& X_j\in \widetilde{M}^N_{(0,6\delta^N_b),(v,V)}(X_i) \notag\\
\Rightarrow
& |\varphi^{2,N}_{t_{min},0}(X_j)-{\varphi^{2,N}_{t_{min},0}(X_i)}|\geq N^{-\frac{1}{6}}  \label{rel.values.v_s}
\end{align}  
where $t_{min}$ shall denote a point in time where $|\varphi^{1,N}_{\cdot,0}(X_j)-{\varphi^{1,N}_{\cdot,0}(X_i)}|$ takes its minimum on $[0,T]$.\\ 
We derive an upper bound for Term \ref{term b} under the condition that
\begin{align*}
& \sum_{j \in \mathcal{M}s^N(X)}\mathbf{1}_{M^N_{(r,R),(v,V)}(X_i)}(X_j) \le   N^{\frac{2\sigma}{9}}\big\lceil N^{\frac{2}{9}} R^2\min\big(\max(V,R),1\big)^4\big\rceil.
\end{align*}
For the first summand we have for $R=\delta_s$
\begin{align}
&\frac{C}{N}\frac{1}{N^{-\beta}\Delta v}\sum_{\substack{j\neq i\\ j\in \mathcal{M}^N_{sb}(X)}}\mathbf{1}_{\widetilde{M}^N_{(r,R),(v,V)}(X_i)}(X_j)\mathbf{1}_{[0,6\delta^N_{sb}]}(r)\nonumber\\
&\le\frac{C}{N}\frac{R^2\min(V,1)^4|\mathcal{M}_{sb}|}{N^{-\beta+1}\max( v,N^{-\frac{1}{6}})}+\frac{CN^{\frac{2}{9}\sigma}}{N^{-\beta}\max( v,N^{-\frac{1}{6}})}\nonumber\\
&\le\frac{CR^2 \min(V,1)^4|\mathcal{M}_{sb}|}{\max(N^{-\frac{1}{6}},v)N^{-\beta+1}}+\frac{CN^{\frac{2}{9}\sigma}}{N^{-\beta+1}\max(N^{-\frac{1}{6}},v)}\nonumber\\
&\le C N^{1+\beta-2s_r-4s_v-2s_{\delta}}\frac{ \min(V,1)^4}{\max(N^{-\frac{1}{6}},v)}+N^{\frac{2}{9}\sigma+\frac{5}{12}-\sigma+\frac{1}{6}-1}\\
&\le C N^{-\frac{19}{12}}+CN^{-\frac{5}{12}}
\label{bound2a1}
\end{align}
Taking additionally into account that $\Delta^N_{sb}(t,X)\le N^{s_{\delta}}= \delta^N_{sb}$ as well as $R=N^{\sigma}r$ for $r\geq 6\delta^N_{sb}$ (see \eqref{cond.para.r,R superbad}) it follows for the second term of \eqref{term a} that
\begin{align}
& \frac{C}{N}\frac{\Delta_{sb}^N(t,X)}{r^{2}\max(r,v)}\sum_{\substack{j\neq i\\ j\in \mathcal{M}^N_b(X)}} \mathbf{1}_{\widetilde{M}^N_{(r,R),(v,V)}(X_i)}(X_j) \notag \\
\le & \frac{C}{N}\Big( \frac{N^{-2s_{r}-4s_{v}+2+\sigma}R^2\min\big(\max(V,R),1\big)^4}{r^2 \max(r,v)}  +\frac{N^{\frac{2\sigma}{3}}}{r^2 \max(r,v)}\Big) N^{-s_{\delta}} \notag \\
\le & C\Big( \frac{\min\big(\max(V,R),1\big)^4}{\max(r,v)} N^{1-2s_r-4s_v-s_{\delta}}+N^{c\sigma-1-s_{\delta}+b_v}\Big)\notag\\
\le &CN^{-\frac{17}{18}}+CN^{-1}. \label{bound2a2}
\end{align}
The sum of Terms \eqref{bound2a1} and \eqref{bound2a2} forms an upper bound for Term \eqref{term a} under the current assumption. 
All sets which belong to the family $\big(M^N_{(r_i,R_i),(v_i,V_i)}(Y)\big)_{i \in I_\sigma}$ are contained in a `collision class' which takes one of the subsequent forms for suitable parameter $r,v\in [0,1]$
\begin{multicols}{2}
\begin{itemize}
\item[(i)] $ M^N_{(0,6\delta^N_s),(0,6\delta^N_s)}(Y)$
\item[(ii)]$M^N_{(0,6\delta^N_s),(v,N^{\sigma}v)}(Y) $
\item[(iii)]  $M^N_{(0,6\delta^N_s),(1,\infty)}(Y)$
\item[(iv)] $M^N_{(r,N^{\sigma}r),(0,6\delta^N_s)}(Y)$
\item[(v)] $M^N_{(r,N^\sigma r),(v,N^\sigma v)}(Y)$
\item[(vi)] $  M^N_{(r,N^\sigma r),(1,\infty)}(Y)$,
\end{itemize}
\end{multicols}
\noindent 
except for $M^N_{(N^{-\sigma},\infty),(0,\infty)}(Y)$, which will be discussed separately like in the previous section.
A set of kind (ii), (iv) or (v) with $v=N^{-\sigma}$ or $r=N^{-\sigma}$ yields the `worst case option' and thus in total Term \eqref{term a} is bounded by
\begin{align}
 & CN^{-\frac{5}{12}} \hspace{0.4cm} \text{if} \hspace{0.4cm} X\in \big(\mathcal{B}^{N,\sigma}_{3,i}\big)^C. \label{bound2a}
\end{align}
For the last class $M^N_{(N^{-\sigma},\infty),(0,\infty)}(Y)$ the following holds
\begin{align}
& \int_{0}^t\frac{1}{N}\sum_{\mathclap{ \substack{j\neq i\\ j\in \mathcal{M}^N_{sb}(X)}}}\Big(\big|f^{N}([\Psi^{1,N}_{s,0}(X)]_j-[\Psi^{1,N}_{s,0}(X)]_i)\notag\\
&  - f^{N}(\varphi^{1,N}_{s,0}(X_j)-{\varphi^{1,N}_{s,0}}(X_i))\big|\Big) \mathbf{1}_{M^N_{(N^{-\sigma},\infty),(0,\infty)}(X_i)}(X_j) ds \notag\\
\le & \frac{2}{N}\sup_{Y\in M^N_{(N^{-\sigma},\infty),(0,\infty)}(X_i)}\int_0^t g^{N}(\varphi^{1,N}_{s,0}(Y)-{\varphi^{1,N}_{s,0}}(X_i))ds  \notag   \sum_{\substack{j\neq i\\ j\in \mathcal{M}^N_{sb}(X)}} \underbrace{\Delta^N_s(t,X)}_{\le N^{-s\delta}}\notag \\
\le & \frac{2}{N} \big(T\frac{C}{(N^{-\sigma})^3}\big)N^{-s_{\delta}}\underbrace{|\mathcal{M}^N_{sb}(X)|}_{\le  N^{\frac{2}{9}(1+\sigma)}} \notag \\
\le &  CN^{-1+\frac{2}{9}-s_{\delta}+c\sigma}\notag\\
\le & CN^{-\frac{19}{18}+C\sigma},
\end{align}
for $X\in \big(\mathcal{B}^{N,\sigma}_{3sb,i}\big)^C$ and $t\le \tau^N(X).$
This is distinctly smaller than necessary for small enough $\sigma>0$ and concludes the estimates for Term \eqref{term a}.

\subsubsection{Unlikely sets $\mathcal{B}^{N,\sigma}_{3b,i}$ and $\mathcal{B}^{N,\sigma}_{3s,i}$}
It only remains to show that the probability related to the sets $\mathcal{B}_{3b,i}^{N,\sigma}$ and $\mathcal{B}_{3s,i}^{N,\sigma}$ is indeed small enough, i.e. that for any $\gamma>0$ there exists a $C_\gamma$ such that
\begin{align*}
&\mathbb{P}\Big(\sum_{j \in \mathcal{M}_b^N(X)}\mathbf{1}_{M^N_{R,V}(X_i)}(X_j) \geq  N^{\frac{3\sigma}{4}}\big\lceil N^{\frac{3}{
4}} R^2\min\big(\max(R,V),1\big)^4\big\rceil\\
&\qquad \vee |\mathcal{M}_b^N(X)|>N^{\frac{3}{4}(1+\sigma)} \Big)\leq C_\gamma N^{-\gamma}
\end{align*}
and analogously that for any $\eta>0$ there exists a $C_\eta$ such that
\begin{align*}
&\mathbb{P}\Big(\sum_{j \in \mathcal{M}_b^N(X)}\mathbf{1}_{M^N_{R,V}(X_i)}(X_j) \geq  N^{\frac{2\sigma}{9}}\big\lceil N^{\frac{2}{
2}} R^2\min\big(\max(R,V),1\big)^4\big\rceil\\
&\qquad \vee |\mathcal{M}_b^N(X)|>N^{\frac{2}{9}(1+\sigma)} \Big)\leq C_\eta N^{-\eta}
\end{align*}
The proof follows the same pattern as in \cite{grass} and is similar in both cases ('bad' and 'superbad'), so we confine ourselves to the proof in the bad particles case.
For clarity, we define 
\begin{align*}
M\coloneqq  \big\lceil N^{\frac{3}{4}\sigma}\lceil N^{\frac{3}{4}} R^2\min\big(\max(V,R),1\big)^4\rceil \big\rceil.
\end{align*} 
Recall that $j\in \mathcal{M}^N_b(X)$ implies that there is at least on $X_k\in \big(G^N(X_j)\big)^C$ for some $k\in \{1,...,N\}\setminus \{j\}$. We will see that for $R,V>0$ 
\begin{align}
 \sum_{j \in \mathcal{M}_b^N(X)}\mathbf{1}_{M^N_{R,V}(X_i)}(X_j)\geq M \label{unlikely set}
\end{align} 
either implies that there exists a $j \in \{1,...,N\}$ such that
\begin{align} \sum_{k=1}^N \mathbf{1}_{(G^N(X_j))^C}(X_k)\geq\lceil  \frac{N^{\frac{\sigma}{4}}}{2}\rceil  \big)\ \label{unlikely set 1}
\end{align}
or there exists a set $\mathcal{S}\subseteq \{1,...,N\}^2\setminus \bigcup_{n=1}^N\{(n,n)\}$ with the following properties
\begin{align}
& \ \text{(i)}\ \ \ |\mathcal{S}|= \lceil  \frac{N^{-\frac{\sigma}{4}}M}{2}\rceil  \notag \\ & \  \text{(ii)} \ \ \forall (j,k)\in \mathcal{S}:X_j\in(G^N(X_k))^C\cap M^N_{R,V}(X_i)   \notag \\
&\ \text{(iii)} \ (j_1,k_1),(j_2,k_2)\in \mathcal{S}\Rightarrow  \{j_1,k_1\}\cap \{j_2,k_2\}=\emptyset.\label{unlikely set 2}
\end{align}
In the proof of this implication we will name the event $X_m\in  M^N_{R,V}(X_n)$ by the phrase 'collision between particles $m,n$' and the phrase 'hard collision between particles $m,n$' will be applied synonymously to the event $X_m \in (G(X_n))^C$.
Note that if assumption \eqref{unlikely set 1} is not fulfilled, it implies that a given 'bad' particle can have at least $\lceil \frac{N^{\frac{\sigma}{4}}}{2}\rceil$ 'hard collisions' with different particles. Such a 'bad' particle can, 'infect' not more than $\lceil \frac{N^{\frac{\sigma}{4}}}{2}\rceil$ other particles, causing them to be included in the set $\mathcal{M}_b^N(X)$. 

For the following considerations we stick to this case and we will see that under this constraint the relation \eqref{unlikely set}, i.e. $$\sum_{j \in \mathcal{M}_b^N(X)}\mathbf{1}_{M^N_{R,V}(X_i)}(X_j)\geq M $$ implies that the event related to \eqref{unlikely set 2} is fulfilled.

In this case there is a set $\mathcal{C}_0\subseteq\mathcal{M}^N_b(X)$ of `bad' particles which have 'collisions' with the particle $i$. By assumption \eqref{unlikely set} we have $|\mathcal{C}_0|\geq M$ and as the event related to \eqref{unlikely set 1} does not occur, there are at most $\lfloor \frac{N^{\frac{\sigma}{4}}}{2}\rfloor $ particles having a 'hard collision' with particle $i$.
We construct a new set $\mathcal{C}_1\subseteq \mathcal{C}_0$ by `detaching' all of these at most $\lfloor \frac{N^{\frac{\sigma}{4}}}{2}\rfloor $ particles, which are possibly contained in $\mathcal{C}_0$, and it obviously holds that 
\begin{align*}
|\mathcal{C}_1|\geq M -\lfloor \frac{N^{\frac{\sigma}{4}}}{2}\rfloor \geq 1,
\end{align*}
for $N$ large enough.
Similarly we take one of these remaining `bad' particles $j_1$ out of $\mathcal{C}_1$ and since $j_1\in \mathcal{C}_1\subseteq \mathcal{C}_0\subseteq \mathcal{M}^N_b(X)$, there must be at least one further particle having a 'hard collision' with $j_1$. 
By construction of $\mathcal{C}_1$ this can not be $i$, so lets call it $k_1$.
This gets us our first tuple $(j_1,k_1)$ which fulfils condition (ii) of the set $\mathcal{S}$ appearing in \eqref{unlikely set 2}. 
In a next step we `detach' $j_1$ and $k_1$ and all of their at most $2\lfloor \frac{N^{\frac{\sigma}{4}}}{2}\rfloor -2$ remaining 'hard collision partners' from $\mathcal{C}_1$ to obtain a new set $\mathcal{C}_2\subseteq \mathcal{C}_1$. 
This gives us an iteration process (provided that $\mathcal{C}_2\neq \emptyset$) by choosing the next particle $j_2$ out of $\mathcal{C}_2$ and afterwards an arbitrary one of its 'hard collision partners' $k_2$. 
Then the next round can start after having removed $j_2$ and $k_2$ as well as their remaining 'hard collision} partners' from $\mathcal{C}_2$ to obtain $\mathcal{C}_3\subseteq \mathcal{C}_2$. 
By construction after each round of this process at most $2\lfloor \frac{N^{\frac{\sigma}{4}}}{2}\rfloor $ `particle labels' are removed from the set $\mathcal{C}_k$ to obtain $\mathcal{C}_{k+1}$.
Considering that $M\geq N^{\frac{3\sigma}{4}}$, we can reiterate this procedure at least 
\begin{align*}
\lceil \frac{M -\lfloor \frac{N^{\frac{\sigma}{4}}}{2}\rfloor }{N^{\frac{\sigma}{4}}}\rceil\geq \lceil \frac{N^{-\frac{\sigma}{4}}M}{2}\rceil 
\end{align*}
 times. The removal of the 'hard collision partners' of the occurring tuples after each round ensures that condition (iii) is fulfilled and thus this provides us a set $\mathcal{S}$ consisting of tuples $(j_i,k_i)$ like claimed in \eqref{unlikely set 2}. 

Due to this considerations we can determine an upper bound for the probability $\mathbb{P}(X\in \mathcal{B}_{3b,i}^{N,\sigma})$. 
Starting with assumption \eqref{unlikely set 2} we abbreviate 
$$M_1\coloneqq  \lceil \frac{N^{-\frac{\sigma}{4}}M}{2}\rceil \text{ with } M=\big\lceil N^{\frac{3}{4}\sigma}\lceil N^{\frac{3}{4}} R^2\min\big(\max(V,R),1\big)^4\rceil \big\rceil .$$ 
There are less than $\binom{N^2}{K}$ different possibilities to choose $K$ `disjoint'  (condition (iii) of \eqref{unlikely set 2} is fulfilled) pairs $(j,k)$ belonging to $\{1,...,N\}^2\setminus \bigcup_{n=1}^N\{(n,n)\}$. Application of this, Lemma \ref{Prob of group} and $\sup_{Y\in \mathbb{R}^6}\mathbb{P}\big(X_1\in (G^N(Y))^C\big)\le C N^{-\frac{11}{6}-2\sigma}$ yields that the probability of the existence of a set $S$ satisfying the three conditions in \ref{unlikely set 2} is small for large $N$, i.e.
{\allowdisplaybreaks \begin{align}
&\mathbb{P}\Big(\exists \mathcal{S}\subseteq \{1,...,N\}^2\setminus \bigcup_{n=1}^N\{(n,n)\}:|\mathcal{S}|=M_1 \ \land \notag \\ & \hspace{0,6cm} \big(\forall (j,k)\in \mathcal{S}:X_j\in(G^N(X_k))^C\cap M^N_{R,V}(X_i)  \big) \ \land \notag \\
& \hspace{0,6cm} \big((j_1,k_1),(j_2,k_2)\in \mathcal{S}\Rightarrow  \{j_1,k_1\}\cap \{j_2,k_2\}=\emptyset\big)\Big) \notag  \\
\le & \binom{N^2}{M_1}\mathbb{P}\Big( \forall (j,k)\in \{(2,3),(4,5),...,(2M_1,2M_1+1)\}: \notag \\
&\hspace{1,5cm} X_j\in(G^N(X_k))^C\cap M^N_{R,V}(X_1) \Big) \notag  \\
\le &\frac{N^{2M_1}}{M_1!}\Big(\sup_{Y\in \mathbb{R}^6}\mathbb{P}\big(X\in (G^N(Y))^C\big)
\sup_{Z\in \mathbb{R}^6}\mathbb{P}\big(X\in M^N_{R,V}(Z)\big)\Big)^{M_1} \notag \\
\le & C^{M_1}\frac{N^{2M_1}}{M_1^{M_1}} \big(N^{-\frac{11}{6}-2\sigma}\big)^{M_1}\Big(R^2\min\big(\max(V,R),1\big)^4\Big)^{M_1}\notag \\
\le &  (CN^{-\frac{5\sigma}{4}} )^{\frac{N^{\frac{\sigma}{4}} }{2}},
\end{align} 
since $M_1\geq \frac{N^{\frac{\sigma}{4}}}{2}$ for
\begin{align*}
M_1=\lceil \frac{N^{-\frac{\sigma}{4}} }{2}M\rceil \text{ with } M=\big\lceil N^{\frac{3}{4}\sigma}\lceil N^{\frac{3}{4}} R^2\min\big(\max(V,R),1\big)^4\rceil \big\rceil.
\end{align*} 
For any class which appears in $\big(M^N_{(r_i,R_i),(v_i,V_i)}(Y)\big)_{i \in I_{\delta}}$ where $R_l\neq \infty$ this probability decays distinctly faster than necessary. 

To prove that $\sum_{k\in \mathcal{M}^N_b(X)}1\le N^{\frac{3}{4}(1+\sigma)}$ we can also apply the considerations from above by setting the collision class parameters $R,V$ to infinity and thus we obtain the event $\mathbf{1}_{M^N_{\infty,\infty}(X_i)}(X_j)=1$. In the case  $M_1\coloneqq  \lceil \frac{N^{\frac{3}{4}+\frac{\sigma}{4}}}{2}\rceil$ and $\mathbb{P}\big(X_1\in M^N_{R,V}(Y)\big)=1$. 
Applying the above procedure, we get
\begin{align*}
\mathbb{P}\big(\sum_{k\in \mathcal{M}^N_b(X)}1\le N^{\frac{3}{4}(1+\sigma)}\big)\le CN^{-\sigma N^{\frac{3}{4}}}
\end{align*} 
which is small enough.
Now, let's proceed with the considerations regarding assumption \eqref{unlikely set 1}. Therefore we abbreviate $M_2\coloneqq  \lceil  \frac{N^{\frac{\sigma}{4}}}{2}\rceil $ and estimate
{\allowdisplaybreaks
\begin{align}
& \mathbb{P}\Big(X\in \mathbb{R}^{6N}: \big(\exists j \in \{1,...,N\}: \sum_{k\neq j} \mathbf{1}_{(G^N(X_j))^C}(X_k)\geq  M_2\big)\Big)\notag \\
\le & N\mathbb{P}\Big(X\in \mathbb{R}^{6N}: \sum_{k=2}^N \mathbf{1}_{(G^N(X_1))^C}(X_k)\geq M_2\Big) \notag \\
\le & N\binom{N}{M_2}\sup_{Y\in \mathbb{R}^6}\mathbb{P}\big(Z\in \mathbb{R}^{6}: Z\in (G^N(Y))^C  \big)^{M_2} \notag \\
\le & N \frac{N^{M_2}}{M_2!}  \big(CN^{-\frac{11}{6}-2\sigma}\big)^{M_2} \notag \\
\le & CN^{-\frac{1}{4}\lceil  \frac{N^{\frac{\sigma}{4}}}{2}\rceil }, \label{prob.est.} 
\end{align}}
which decreases fast enough as $N$ increases.
In total we obtain as desired
\begin{align}
& \mathbb{P}\big(X\in \mathcal{B}^{N,\sigma}_{3b,i} \big) \notag \\
\le & |I_{\sigma}|\sup_{\substack{R,V>0}}\mathbb{P}\Big(\sum_{j \in \mathcal{M}_b^N(X)}\mathbf{1}_{M^N_{R,V}(X_i)}(X_j) \geq  N^{\frac{3\sigma}{4}}\big\lceil N^{\frac{3}{
4}} R^2\min\big(\max(R,V),1\big)^4\big\rceil \Big) \notag \\
& + \mathbb{P}\big(\sum_{k\in \mathcal{M}^N_b(X)}1\geq N^{\frac{3}{4}(1+\sigma)}\big)  \notag\\
\le &(CN^{-\frac{5\sigma}{4}} )^{\frac{N^{\frac{\sigma}{4}}}{2} } 
 \label{prob.b.3b}
\end{align}
Similarly we can show that the probability related to the set $\mathcal{B}^{N,\sigma}_{3s,i}$ in the superbad particle case is indeed small enough.
A similar estimate holds for the superbad particles
\begin{align}
& \mathbb{P}\big(X\in \mathcal{B}^{N,\sigma}_{3s,i} \big) \notag \\
\le & |I_{\sigma}|\sup_{\substack{R,V>0}}\mathbb{P}\Big(\sum_{j \in \mathcal{M}_b^N(X)}\mathbf{1}_{M^N_{R,V}(X_i)}(X_j) \geq  N^{\frac{2\sigma}{9}}\big\lceil N^{\frac{2}{
9}} R^2\min\big(\max(R,V),1\big)^4\big\rceil \Big) \notag \\
& + \mathbb{P}\big(\sum_{k\in \mathcal{M}^N_b(X)}1\geq N^{\frac{2}{9}(1+\sigma)}\big)  \notag\\
\le &(CN^{-\frac{16\sigma}{9}} )^{\frac{N^{\frac{\sigma}{9}}}{2} } 
 \label{prob.b.3s}
\end{align}
\subsubsection{Estimate of Term \ref{term c} (i good j good) }
Now we are left with the last Term \ref{term c} which measures the fluctuation between two good particles. 
To estimate the term we identify 
\begin{align*}
 v^N_{min}=N^{b_v}=N^{-\frac{1}{6}},
 \end{align*}
since $i,j\in \mathcal{M}^N_g(X).$
To estimate the term we apply Corollary \ref{corollary phi and  psi} and subdivide the term depending on the relative velocity of the particles so that the first term deals with collisions where the relative velocity is below order $N^{-\frac{1}{9}+3\sigma}$ and the second deals with the rest. 
The choice of the value is more or less random as long as the equations stay small.
Corollary \ref{corollary phi and  psi} (ii) is applicable since the relative velocity values for the considered `collision classes' are of distinctly larger order than the deviation between corresponding particle trajectories of the microscopic and the auxiliary system. 
Note that $G^N(X_i)\subseteq M(^N_{ 6\delta^N_b, v^N_{min}}(X_i))^c$ where $\delta_b^N=N^{-\frac{7}{24}-\sigma}$ and
\begin{align*}
&\max_{i\in \mathcal{M}^N_g(X)}\sup_{0\le s \le \tau^N(X)}|[\Psi^N_{s,0}(X)]_i-\varphi^N_{s,0}(X_i)|
\le 
N^{-\frac{5}{12}+\sigma}=N^{-\frac{1}{4}+\sigma} v^N_{min}. 
\end{align*}
Thus Term \ref{term c} is bounded by
{\allowdisplaybreaks
\begin{align}
& \int_{0}^t\Big(\frac{1}{N}\sum_{\substack{j\neq i\\ j\in \mathcal{M}^N_g(X)}}\Big(\big|f^{N}([\Psi^{1,N}_{s,0}(X)]_j-[\Psi^{1,N}_{s,0}(X)]_i)\big| \notag  \\
&+ \big|f^{N}(\varphi^{1,N}_{s,0}(X_j)-{\varphi^{1,N}_{s,0}}(X_i))\big|\Big)\mathbf{1}_{G^N(X_i)\cap M^N_{3N^{-\frac{1}{2}+\sigma},\infty}(X_i)}(X_j)\Big) ds \notag \\
\le & \frac{C}{N}\frac{1}{N^{-\beta} v^N_{min}} \sum_{j\neq i} \mathbf{1}_{G^N(X_i)\cap M^N_{3N^{-\frac{1}{2}+\sigma},N^{-\frac{1}{9}+3\sigma}}(X_i)}(X_j) \notag \\
&+  \frac{C}{N}\frac{1}{N^{-\beta}N^{-\frac{1}{9}+3\sigma}} \sum_{j\neq i} \mathbf{1}_{ G^N(X_i)\cap  M^N_{3N^{-\frac{1}{2}+\sigma},\infty}(X_i)}(X_j). \label{term c estimated}
\end{align}}
This stays sufficiently small since the concerned sets are very unlikely.
To prove this we define
\begin{align}
\begin{split}
& X\in \mathcal{B}_{4,i}^{N,\sigma}\subseteq \mathbb{R}^{6N} \\
\Leftrightarrow &  \sum_{j \neq i}\mathbf{1}_{M^N_{6N^{-\frac{1}{2}+\sigma},N^{-\frac{1}{9}+3\sigma}}(X_i)}(X_j)\geq  N^{\frac{\sigma}{2}}\ \land \\
& \sum_{j \neq i}\mathbf{1}_{M^N_{6N^{-\frac{1}{2}+\sigma},\infty}(X_i)}(X_j)\geq  N^{3\sigma}
\end{split} \label{def.B_4}.
\end{align}
Set $M_1\coloneqq  \lceil N^{\frac{\sigma}{2}} \rceil$ and $M_2\coloneqq  \lceil N^{3\sigma} \rceil$. By the same proof as applied in \eqref{prob.est.} and application of Lemma \ref{Prob of group} we can estimate the probability
\begin{align}
& \mathbb{P}\big(X\in  \mathcal{B}^{N,\sigma}_{4,i}\big) \notag \\
\le & \frac{N^{M_1}}{M_1!}\sup_{Y\in \mathbb{R}^6}\mathbb{P}\big(X_i\in M^N_{6N^{-\frac{1}{2}+\sigma},N^{-\frac{1}{9}+3\sigma}}(Y)\big)^{M_1} \notag \\
& + \frac{N^{M_2}}{M_2!}\sup_{Y\in \mathbb{R}^6}\mathbb{P}\big(X_i\in M^N_{6N^{-\frac{1}{2}+\sigma},\infty}(Y)\big)^{M_2} \notag \\
\le & \frac{(CN)^{M_1}}{M_1!}\big(N^{2(-\frac{1}{2}+\sigma)}\big)^{M_1}\big(N^{4(-\frac{1}{9}+3\sigma)}\big)^{M_1}  + C^{M_2}\frac{N^{M_2}}{(N^{3\sigma})^{M_2}}\big(N^{2(-\frac{1}{2}+\sigma)}\big)^{M_2}  \notag \\
\le &C\big(N^{-\frac{4}{9}+14\sigma} \big)^{N^{\frac{\sigma}{2}}} +\big( CN^{-\sigma}\big)^{N^{3\sigma}}, \label{prob.b.4}
\end{align} 
which for $\sigma>0$ small enough decreases fast enough.

Due to our estimates it holds for $X\in (\mathcal{B}^{N,\sigma}_{4,i})^C$ that Term \eqref{term c estimated}, and thereby Term \eqref{term c}, is bounded by
\begin{align}
&\frac{C}{N}\frac{1}{N^{-\beta}v^N_{b}} N^{\frac{\sigma}{2}} +  \frac{C}{N}\frac{1}{N^{-\beta}N^{-\frac{1}{9}+3\sigma}} N^{3\sigma}\notag \\
\le & 
CN^{-\frac{5}{12}- \frac{\sigma}{2}} +CN^{-\frac{17}{36}} \le CN^{-\frac{5}{12}}\label{bound5}
\end{align}
Due to the previous probability estimates on the unlikely sets it easily follows that for small enough $\sigma>0$ and an arbitrary $\gamma>0$ there is a constant $C>0$ such that $$\mathbb{P}\big(\bigcup_{j\in \{1,2,3,4\}}\bigcup_{i=1}^N \mathcal{B}^{N,\sigma}_{j,i}\big)\le CN^{-\gamma}.$$ 
\subsubsection{Conclusion for case 1 (labelled particle $X_i$ is good)}
For $i\in \mathcal{M}_g^N(X)$ we determined an upper bound for the term 
$$  \big|\int_{t_1}^t\frac{1}{N}\sum_{j\neq i}f^{N}([\Psi^{1,N}_{s,0}(X)]_i-[\Psi^{1,N}_{s,0}(X)]_j)-f^{N}*\widetilde{k}^N_s({\varphi^{1,N}_{s,0}}(X_i))ds\big|$$
which is given by the sum of bounds of the four Terms \eqref{eq:good1}, \eqref{eq:good2}, \eqref{eq:good3} and \eqref{eq:good4}.
We restrict ourselves to the configurations $X\in\Big(\bigcup_{j\in \{1,2,3,4\}}\bigcup_{i=1}^N \mathcal{B}^{N,\sigma}_{j,i}\Big)^C$ and all upper bounds hold for any times $t_1,t\in [0,\tau^N(X)]$. 
For a suitable constant $C>0$, $CN^{-\frac{5}{12}}$ dominates all of these upper bounds except for Term \eqref{term c2}. 
But for our underlying configurations it holds for any $i\in \{1,...,N\}$ and times $t_1,t\in [0,T]$ that 
\begin{align*}
& \big|\int_{t_1}^t\Big(\frac{1}{N}\sum_{j=1}^N g^{N}(\varphi^{1,N}_{s,0}(X_j)-{\varphi^{1,N}_{s,0}}(X_i))\mathbf{1}_{G^N(X_i)}(X_j)\\
&  -\int_{\mathbb{R}^6} g^{N}(\varphi^{1,N}_{s,0}(Y)-{\varphi^{1,N}_{s,0}}(X_i))\mathbf{1}_{G^N(X_i)}(Y)k_0(Y)
d^6Y\Big)ds\Big|\le 1.
\end{align*}
By the definition of $\mathcal{B}^{N,\sigma}_{2,i}$ and thus for $N>1$ and $t_1\le t$ we receive 
\begin{align}
& \int_{t_1}^t\frac{1}{N}\sum_{j=1}^N g^{N}(\varphi^{1,N}_{s,0}(X_j)-{\varphi^{1,N}_{s,0}}(X_i))\mathbf{1}_{G^N(X_i)}(X_j)ds \notag \\
\le & 1+\int_{t_1}^t\int_{\mathbb{R}^6} g^{N}(\varphi^{1,N}_{s,0}(Y)-{\varphi^{1,N}_{s,0}}(X_i))\mathbf{1}_{G^N(X_i)}(Y)k_0(Y)d^6Yds \notag \\
\le &1+ C\ln(N)(t-t_1). \label{est.sum.g}
\end{align}
We used the fact that for $N>1$
\begin{align*}
& \sup_{t_1\leq s\leq t}\int_{\mathbb{R}^6} g^{N}({\varphi^{1,N}_{s,0}}(Y) -{\varphi^{1,N}_{s,0}}(X_i))\mathbf{1}_{G^N(X_i)}(Y)k_0(Y)d^6Y  \\
\leq &  C\sup_{t_1\le s\le t}\int_{\mathbb{R}^3}\min\big(N^{3\beta},\frac{1}{|Y-{\varphi^{1,N}_{s,0}}(X_i)|^3}\big)\widetilde{k}^N_s(Y)d^3Y  \\
\leq & C\ln(N)
\end{align*}
holds. This leads us in particular for times $t_1\le t$  to 
$$\Delta^N_g(t,X)\le \Delta^N_g(t_1,X)+\int_{t_1}^t\delta^N_g(s,X)ds, $$
with the common abbreviations 
\begin{align}
&\delta^N_g(t,X)\coloneqq  \max_{i\in \mathcal{M}^N_g(X)}|[\Psi^{2,N}_{t,0}(X)]_i-{\varphi^{2,N}_{t,0}}(X_i)| \nonumber , \\ &\Delta^N_g(t,X)\coloneqq  \max_{i\in \mathcal{M}^N_g(X)}\sup_{0\le s\le t}|[\Psi^{1,N}_{s,0}(X)]_i-{\varphi^{1,N}_{s,0}}(X_i)|. \label{abbrev.delta}
\end{align}
By choosing the subsequent sequence of time steps $t^*\coloneqq  t_{n+1}-t_n= \frac{C}{\sqrt{\ln(N)}}$ for some constant $C>0$ with
\begin{align*}
& t_n=n\frac{C}{\sqrt{\ln(N)}} \text{ for }n\in \{0,...,\lceil \frac{\sqrt{\ln(N)}}{C}\tau^N(X)\rceil-1\},\\
& t_{\lceil \frac{\sqrt{\ln(N)}}{C}\tau^N(X)\rceil}=\tau^N(X)
\end{align*}
the previous relation implies that for $ t_n\le t\le \tau^N(X)$ 
\begin{align}
\Delta^N_g(t,X)\le \sum_{k=1}^n\sup_{0\le s \le t_k}\delta^N_g(s,X)t^*+\int_{t_n}^t\delta^N_g(s,X)ds.
\label{rel.v.q.bound}
\end{align}
It follows that for any `good' particle $i\in \mathcal{M}^N_g(X)$, the considered configurations and for all times $t\in [t_n,t_{n+1}]$, where $n\in \{0,...,\lceil \frac{\sqrt{\ln(N)}}{C}\tau^N(X)\rceil-1\}$ the following inequality holds
\begin{align}
&\delta^N_g(t,X)\notag \\
\le & \delta^N_g(t_n,X)+\max_{i\in \mathcal{M}^N_g(X)} \big|\int_{t_n}^t\Big(\frac{1}{N}\sum_{j\neq i}f^{N}([\Psi^{1,N}_{s,0}(X)]_i-[\Psi^{1,N}_{s,0}(X)]_j) \notag \\
&   -f^{N}*\widetilde{k}^N_s({\varphi^{1,N}_{s,0}}(X_i))\Big)ds\big| \notag  \\
\le &\max_{i\in \{1,...,N\}}\int_{t_n}^t \frac{2}{N}\sum_{j=1}^N g^{N}(\varphi^{1,N}_{s,0}(X_j)-{\varphi^{1,N}_{s,0}}(X_i))\mathbf{1}_{G^N(X_i)}(X_j)\underbrace{\Delta^N_g(s,X)}_{\le \Delta^N_g(t,X)}ds  \notag \\
&+\delta^N_g(t_n,X) +CN^{-\frac{5}{12}} \notag \\
\le &\big(1+ C\ln(N)\underbrace{(t-t_n)}_{\le t_{n+1}-t_n=t^*}\big)\big( \sum_{k=1}^n\sup_{0\le s \le t_k}\delta^N_g(s,X)t^*+\int_{t_n}^t\delta^N_g(r,X)dr\big)\notag \\
& +\delta^N_g(t_n,X)+CN^{-\frac{5}{12}} \notag \\
\le &\big(1+ C\ln(N)t^*\big)\int_{t_n}^t\delta^N_g(r,X) dr\notag \\
&  + \big(2+ C\ln(N)(t^*)^2\big)\sum_{k=1}^n\sup_{0\le s \le t_k}\delta^N_g(s,X) +CN^{-\frac{5}{12}}. 
\end{align} 
Application of Gronwall`s Lemma implies that for all times $t\in [t_n,t_{n+1}]$ it holds that
\begin{align}
 & \delta^N_g(t,X)\notag \\
\le &  \Big(\big(2+ C\ln(N)(t^*)^2\big)\sum_{k=1}^n\sup_{0\le s \le t_k}\delta^N_g(s,X) +CN^{-\frac{5}{12}}\Big) e^{t^*+ C\ln(N){(t^*)}^2}.
\end{align}
Especially for $t\in [0,t_n]$, we can exchange the left-hand side by its supremum over $[0,t_{n+1}]$. 
For $t^*=\frac{C_1}{\sqrt{\ln(N)}}$ with $C_1\coloneqq  \min\big(\frac{1}{\sqrt{C}},1\big)$ the previous relation implies 
\begin{align}
 & \sup_{0\le s \le t_{n+1}}\delta^N_g(s,X)
\le   3e^2\sum_{k=1}^n\sup_{0\le s \le t_k}\delta^N_g(s,X) +Ce^2N^{-\frac{5}{12}}. \label{Gron.v bound}
\end{align}
Due to this relation it follows for $n\in \{1,...,\lceil \frac{\sqrt{\ln(N)}}{C_1}\tau^N(X)\rceil\}$ that
\begin{align}
& \sup_{0\le s \le t_n}\delta^N_g(s,X)
\le  Ce^2N^{-\frac{5}{12}}(3e^2+1)^{n-1}. \label{upp.bound.v.thm1}
\end{align}
For $n=1$ the relation is obvious due to \eqref{Gron.v bound} and if it holds for $k\in \{1,...,n\}$, $n\in \N$, where we fix the constant $C$ for these estimates, then we obtain that
\begin{align*}
& \sup_{0\le s \le t_{n+1}}\delta^N_g(s,X)\\
\le &  3e^2\sum_{k=1}^n \underbrace{\sup_{0\le s \le t_k}\delta^N_g(s,X) }_{\le   Ce^2N^{-\frac{5}{12}}(3e^2+1)^{k-1}}+Ce^2N^{-\frac{5}{12}} \\
\le & 3e^2\big(  Ce^2N^{-\frac{5}{12}}\frac{(3e^2+1)^n-1}{(3e^2+1)-1}\big)+Ce^2N^{-\frac{5}{12}}\\
=&  Ce^2N^{-\frac{5}{12}}(3e^2+1)^n.
\end{align*}
This confirms the claim and it follows that 
\begin{align}
\sup_{0\le s \le \tau^N(X)}\delta^N_g(s,X)
\le &  Ce^2N^{-\frac{5}{12}}(3e^2+1)^{\lceil \frac{\sqrt{\ln(N)}}{C_1}\tau^N(X)\rceil-1} \notag \\
\le & Ce^2N^{-\frac{5}{12}}N^{ \frac{\ln(3e^2+1)}{\ln(N)}\frac{\sqrt{\ln(N)}}{C_1}T} \notag \\
\le & CN^{-\frac{5}{12}+\frac{\sigma}{2}},
\end{align}
for $N$ large enough.
The received upper bound for the velocity deviation implies that
\begin{align}
& \max_{i\in \mathcal{M}^N_g(X)}\sup_{0\le s \le \tau^N(X)}|[\Psi^N_{s,0}(X)]_i-\varphi^N_{s,0}(X_i)|  
\le  CN^{-\frac{5}{12}+\frac{\sigma}{2}},
\end{align}
which is smaller than necessary since $CN^{-\frac{5}{12}+\frac{\sigma}{2}}< N^{-\frac{5}{12}+\sigma}$ for $\sigma>0$ and $N$ large enough. \\
\subsection{Controlling the deviation of the bad and superbad particles}

Most estimates for the second part can be applied analogously, except that we allow more distance of the observed `bad' or `superbad' particle to its mean-field partner, since $\delta_s=N^{-\frac{1}{6}-\sigma}>N^{-\frac{5}{12}+\sigma}$ and $\delta_b=N^{-\frac{7}{24}-\sigma}>N^{-\frac{5}{12}+\sigma}$.
For the `good' particle this distance is of the same order as the cut-off radius. 
The vast majority of particles is typically `good', so we have control over the `collision partners' in most cases. 
By the definition of the distance, the considered `bad' or `superbad' particle is inside a ball of radius $N^{-\frac{1}{6}-\sigma}$ or respectively $N^{-\frac{7}{24}-\sigma}$ around its related mean-field particle. 
To circumvent this problem, we define a cloud of auxiliary `mean-field particles' around the `bad' or `superbad' particle, like proposed in \cite{grass}. 
`Hard' or `Superhard' collisions might cause that the observed particle departs too far from its initially corresponding mean-field particle, that propagates homogeneously in time. 
Phillip was able to show that for any point in time, we can find an auxiliary particle around the `bad' or `superbad' particle with a disance small than the cut-off.
By exchanging the these particles, we can copy the estimates from Section \ref{case1}. 

To ensure that we can apply Theorem \ref{Prop:LLN}, we have to introduce a `cloud' of auxiliary particles instead of a single one when needed, because the introduced auxiliary particle would depend on the whole configuration and thus be correlated with the remaining particles and we would loose the big advantage of the `mean-field particle'.
If we propagate the whole `cloud' from the beginning at the time of a `hard collision' for a certain particle the initial positions of the related auxiliary particles are chosen independently of the remaining configuration.
We will show that all of the auxiliary particles which belong to the small `cloud' fulfill corresponding demands with high probability like in the previous situation, where we could show for typical initial data that the related mean-field particles fulfill properties which made it possible to prove that the effective and the microscopic dynamics are usually close. 
In the upcoming part we will end up in a very similar situation as in Section \ref{case1} and we will benefit from the proof techniques of the previous chapter.

\subsection{Controlling the deviation of the superbad particles}\label{case2}
To create the particle cloud we first define 
\begin{align}
Q_N \coloneqq   & \{-\lceil N^{\frac{1}{4}} \rceil ,...,-1,0,1,..., \lceil  N^{\frac{1}{4}} \rceil\}^6  \label{Q_N sb}
\end{align}
and for $(k_1,...,k_6)\in Q_N$ the positions or the initial data of the auxiliary particles $X^i_{k_1,...,k_6}\coloneqq  X_i+\sum_{j=1}^6 k_jN^{-\frac{5}{12}+\frac{\sigma}{2}}{e}_j$, where ${e}_j$, $j\in \{1,...,6\}$ is the $j$-th basis vector of $\mathbb{R}^6$. 
According to Lemma \ref{Lemma distance same order}, which ensures that the distance between mean-field particles stays of the same order, and $\delta_s=N^{-\frac{1}{6}-\sigma}$ for $t\leq\tau$, it holds for arbitrary $t_1\in [0,\tau^N(X)]$ and large enough $N$ that 
\begin{align}
|\varphi^N_{0,t_1}([\Psi^N_{t_1,0}(X)]_i)-X_i|\le C|[\Psi^N_{t_1,0}(X)]_i-\varphi^N_{t_1,0}(X_i)|<CN^{-\frac{1}{6}-\sigma}\le N^{-\frac{1}{6}}.\label{dist casd bad}
\end{align}
It is always possible to find a tuple $(k_1,...,k_6)\in Q_N$ for $N$ large enough such that 
\begin{align}
|\varphi^N_{0,t_1}([\Psi^N_{t_1,0}(X)]_i)-X^i_{k_1,...,k_6}|\le \frac{\sqrt{6}}{2} N^{-\frac{5}{12}+\frac{\sigma}{2}} 
\end{align}
since \eqref{dist casd bad} is of smaller order with respect to $N$ than the diameter of the auxiliary `particle cloud' around $X_i$.
Lemma \ref{Lemma distance same order} implies in turn that
\begin{align}
|[\Psi^N_{t_1,0}(X)]_i-\varphi^N_{t_1,0}(X^i_{k_1,...,k_6})|\le C  N^{-\frac{5}{12}+\frac{\sigma}{2}}. 
\end{align} 
If we choose $N\in \mathbb{N}$ large enough such that $C  N^{-\frac{5}{12}+\frac{\sigma}{2}}< \frac{1}{2}  N^{-\frac{5}{12}+\sigma}$ and $\sigma>0$ sufficiently small, then there exists a further point in time $t_2 \in (t_1, T]$ such that not only
$$\sup_{s\in [t_1,t_2]}|[\Psi^{1,N}_{s,0}(X)]_i-{\varphi^{1,N}_{s,0}(X^i_{k_1,...,k_6})}|\le  N^{-\frac{5}{12}+\sigma}$$
holds, but also the following bound for the velocity deviation
$$\sup_{s\in [t_1,t_2]}|[\Psi^{2,N}_{s,0}(X)]_i-{^2\varphi^N_{s,0}(X^i_{k_1,...,k_6})}|\le   N^{-\frac{1}{6}-\sigma}.$$
Now we have a sufficiently good approximation for the trajectory of real particle, given by the trajectory of the auxiliary particle with initial datum $X^i_{k_1,...,k_6}$ for this time span.
We apply this to prove that 
\begin{align}
&\sup_{t_1\le s \le t}|[\Psi^N_{s,0}(X)]_i-{\varphi^N_{s,0}}(X_i)|\nonumber\\
&\leq\sup_{t_1\le s \le t}|[\Psi^N_{s,0}(X)]_i-{\varphi^N_{s,0}}(X^i_{k_1,...,k_6})|+\sup_{t_1\le s \le t}|{\varphi^N_{s,0}}(X^i_{k_1,...,k_6})-{\varphi^N_{s,0}}(X_i)| \label{superbaddouble.upp.bound}
\end{align}
grows slow enough on this interval. 
The considerations for the first term is mostly analogous to the estimates of case 1, see Section \ref{case1}, because the spatial distance between the considered auxiliary particle and the `real' particle is bounded by $N^{-\frac{5}{12}+\sigma}$ like the largest allowed deviation for a `good' particle.

From now on we will assume that for an arbitrary point in time $t_1\in [0,\tau^N(X))$ and $X\in \mathbb{R}^{6N}$ the initial position of the auxiliary particle $X^i_{k_1,...,k_6}$ and $t_2\in (t_1,\tau^N(X)]$ are chosen such that the previously introduced demands are fulfilled on $[t_1,t_2]$. 
Following the notation of \cite{grass} we abbreviate $\widetilde{X}_i\coloneqq  X^i_{k_1,...,k_6}$ but remind that $t_2$ and the choice of $(k_1,...,k_6)\in Q_N$ depends on $i,t_1$ and $X$.\\
Controlling the growth of the second term is a simple application of Lemma \ref{Lemma distance same order}. 
It follows for arbitrary $t\in [t_1,t_2]$ that
\begin{align}
& |{\varphi^N_{t,0}}(\widetilde{X}_i)-{\varphi^N_{t,0}}(X_i)| \notag \\
\le & e^{C(t-t_1)}|{\varphi^N_{t_1,0}}(\widetilde{X}_i)-{\varphi^N_{t_1,0}}(X_i)| \notag \\
\le & e^{C(t-t_1)}\big(\big|\varphi^N_{t_1,0}({X}_i)-[\Psi^N_{t_1,0}(X)]_i\big|+\big|[\Psi^N_{t_1,0}(X)]_i-\varphi^N_{t_1,0}(\widetilde{X}_i)\big|\big)\notag  \\
\le & e^{C(t-t_1)}\big(\big|\varphi^N_{t_1,0}({X}_i)-[\Psi^N_{t_1,0}(X)]_i\big|+N^{-\frac{5}{12}+\sigma}\big), \label{est.eff.superbad.term}
\end{align}
where we applied bound \eqref{dist casd bad} according to the choice of $\widetilde{X}_i$. 
This concludes the estimates for this term and we will return to it at the end of this subsection after estimating Term \eqref{eq:sbad2}, Term \eqref{eq:sbad3} and Term \eqref{eq:sbad1}.\\
For the second term we first remark that
\begin{align}
&  |[^2\Psi^N_{t,0}(X)]_i-{^2\varphi^N_{t,0}}(\widetilde{X}_i)|  \notag \\
\le  & |[^2\Psi^N_{t_1,0}(X)]_i-{^2\varphi^N_{t_1,0}}(\widetilde{X}_i)|  \notag \\
& + \big|\int_{t_1}^{t}\Big(\frac{1}{N}\sum_{j\neq i}f^{N}([\Psi^{1,N}_{s,0}(X)]_i-[\Psi^{1,N}_{s,0}(X)]_j)-f^{N}*\widetilde{k}_s(\varphi^{1,N}_{s,0}(\widetilde{X}_i))\Big)ds\big|. 
\end{align}
To derive an upper bound for the force term, note that the same structure as in the previous case. Thus we can again apply multiple times triangle inequality and obtain essentially the four terms of case 1, see Section \ref{case1},
{\allowdisplaybreaks
\begin{align}
& \big|\int_{t_1}^{t}\frac{1}{N}\sum_{j\neq i}f^{N}([\Psi^{1,N}_{s,0}(X)]_i-[\Psi^{1,N}_{s,0}(X)]_j)-f^{N}*\widetilde{k}_s(\varphi^{1,N}_{s,0}(\widetilde{X}_i))ds\big| \label{force.bad case.sb}\\
\le &  \big|\int_{t_1}^{t}\frac{1}{N}\sum_{j\neq i}f^{N}([\Psi^{1,N}_{s,0}(X)]_i-[\Psi^{1,N}_{s,0}(X)]_j)\mathds 1_{(G^N(\widetilde{X}_i))^C}(X_j)ds\big| \label{eq:sbad1} \\
&+  \big|\int_{t_1}^{t}\frac{1}{N}\sum_{j\neq i}\Big(
f^{N}([\Psi^{1,N}_{s,0}(X)]_i-[\Psi^{1,N}_{s,0}(X)]_j)\mathds 1_{G^N(\widetilde{X}_i)}(X_j) \notag  \\
&-f^{N}(\varphi^{1,N}_{s,0}(\widetilde{X}_i)-{\varphi^{1,N}_{s,0}}(X_j))\mathds 1_{G^N(\widetilde{X}_i)}(X_j)\Big) ds\big| \label{eq:sbad2} \\
& +\big| \int_{t_1}^{t} \frac{1}{N}\sum_{j\neq i}f^{N}(\varphi^{1,N}_{s,0}(\widetilde{X}_i)-{\varphi^{1,N}_{s,0}}(X_j))\mathds 1_{G^N(\widetilde{X}_i)}(X_j)ds  \notag\\
&-\int_{t_1}^{t}\int_{\mathbb{R}^6}f^N(\varphi^{1,N}_{s,0}(\widetilde{X}_i)-{\varphi^{1,N}_{s,0}}(Y))\mathds 1_{G^N(\widetilde{X}_i)}(Y)k_0(Y)d^6Yds\big| \label{eq:sbad3}  \\
& +\big|\int_{t_1}^{t}\Big(\int_{\mathbb{R}^6}f^N(\varphi_{s,0}^{1,N}(\widetilde{X}_i)-{\varphi^{1,N}_{s,0}}(Y))\mathds 1_{G^N(\widetilde{X}_i)}(Y)k_0(Y)d^6Y  \notag\\
&-f^{N}*\widetilde{k}^N_s(\varphi^{1,N}_{s,0}(\widetilde{X}_i))\Big)ds\big|. \label{eq:sbad4}
\end{align} }
\subsubsection{Estimate of Term \ref{eq:sbad4}}
An upper bound for Term \eqref{eq:sbad4} can be derived analogously to the estimates of Term \ref{eq:good4} and thus is also given by $ CN^{-\frac{5}{12}}$ as 
\begin{align*}
& f^{N}*\widetilde{k}^N_s({\varphi^{1,N}_{s,0}}(\widetilde{X}_i))\\
&=\int_{\mathbb{R}^6}f^{N}({\varphi^{1,N}_{s,0}}(\widetilde{X}_i)-{^1Y})k^N_s(Y)d^6Y\\
&= \int_{ \mathbb{R}^6}f^{N}({\varphi^{1,N}_{s,0}}(\widetilde{X}_i)-{\varphi^{1,N}_{s,0}}(Y))k_0(Y)d^6Y,
\end{align*}
which yields
\begin{align*}
& \big|\int_{t_1}^t\int_{ \mathbb{R}^6}f^{N}(\varphi^{1,N}_{s,0}(\widetilde{X}_i)-{\varphi^{1,N}_{s,0}}(Y))k_0(Y)\mathds 1_{G^N(\widetilde{X}_i)}(Y)d^6Yds \notag \\
&-\int_{t_1}^tf^{N}*\widetilde{k}^N_s({\varphi^{1,N}_{s,0}}(\widetilde{X}_i))ds\big| \notag \\
= & \big|\int_{t_1}^t\int_{ \mathbb{R}^6}f^{N}(\varphi^{1,N}_{s,0}(\widetilde{X}_i)-{\varphi^{1,N}_{s,0}}(Y))k_0(Y)(\mathds 1_{G^N(\widetilde{X}_i)}(Y)-1)d^6Yds\big| \notag \\
\le & T\|f^{N}\|_{\infty}\int_{ \mathbb{R}^6}\mathds 1_{(G^N(\widetilde{X}_i))^C}(Y)k_0(Y)d^6Y \notag \\
\le& CTN^{2\beta} (N^{-2b_r-4b_v}+N^{-2s_r-4s_v})\\
\le& CTN^{\frac{15}{18}-2b_r-4b_v}\\
\le&  C N^{-\frac{5}{12}-4\sigma}
\end{align*}

\subsubsection{Estimate of Term \ref{eq:sbad2} and Term \ref{eq:sbad3}}
For the Terms \eqref{eq:sbad2} and \eqref{eq:sbad3} we will utilize Theorem \ref{Prop:LLN}.
Since according to the choice of $t_1, t_2$ and $\widetilde{X}_i$ it holds that $\sup_{t_1\le s \le t_2}|[\Psi^{1,N}_{s,0}(X)]_i-{\varphi^{1,N}_{s,0}}(\widetilde{X}_i)|\le N^{-\frac{5}{12}+\sigma}$, it follows by estimating with the map $g^N$ that
\begin{align}
& \big|\int_{t_1}^{t}\frac{1}{N}\sum_{j\neq i}\Big(f^{N}([\Psi^{1,N}_{s,0}(X)]_i-[\Psi^{1,N}_{s,0}(X)]_j)\mathds 1_{G^N(\widetilde{X}_i)}(X_j) \notag  \\
&-f^{N}(\varphi^{1,N}_{s,0}(\widetilde{X}_i)-{\varphi^{1,N}_{s,0}}(X_j))\mathds 1_{G^N(\widetilde{X}_i)}(X_j)\Big) ds\big| \label{good force,sbad case} \\
\le  &  \big|\int_{t_1}^{t}\frac{1}{N}\sum_{j\in \mathcal{M}^N_b(X)\setminus \{i\}}\Big(f^{N}([\Psi^{1,N}_{s,0}(X)]_i-[\Psi^{1,N}_{s,0}(X)]_j)\mathds 1_{G^N(\widetilde{X}_i)}(X_j) \notag   \\
&-f^{N}(\varphi^{1,N}_{s,0}(\widetilde{X}_i)-{\varphi^{1,N}_{s,0}}(X_j))\mathds 1_{G^N(\widetilde{X}_i)}(X_j)\Big) ds\big| \notag  \\
& +\int_{t_1}^{t}\frac{1}{N}\sum_{j\in \mathcal{M}^N_g(X)\setminus \{i\}}\Big( g^{N}(\varphi^{1,N}_{s,0}(\widetilde{X}_i)-{\varphi^{1,N}_{s,0}}(X_j))\mathds 1_{G^N(\widetilde{X}_i)}(X_j) \notag  \\
& \cdot \big(|[\Psi^{1,N}_{s,0}(X)]_i-{\varphi^{1,N}_{s,0}}(\widetilde{X}_i)|+|[\Psi^{1,N}_{s,0}(X)]_j-{\varphi^{1,N}_{s,0}}({X}_j)| \big)\Big) ds. \label{est.sb2}
\end{align}
All of these terms have basically the same structure as in case 1, see Section \ref{case1}, and the upper bound of the deviation of the true and the auxiliary dynamic is the same as the allowed deviations of `good' particles and so we only have to make minor modifications to the definitions of the unlikely sets $\mathcal{B}^{N,\sigma}_{i,j}$.
We define for $(k_1,...,k_6)\in Q_N$
\begin{align}
& X\in \mathcal{B}_{1,i,(k_1,..,k_6)}^{N,\sigma}\subseteq \mathbb{R}^{6N} \notag\\
\Leftrightarrow & \exists t_1',t_2'\in [0,T]:\notag \\
& \Big|\int_{t_1'}^{t_2'}\Big(\frac{1}{N}\sum_{j\neq i}f^{N}(\varphi^{1,N}_{s,0}(X^i_{k_1,...,k_6})-{\varphi^{1,N}_{s,0}}(X_j))\mathds 1_{G^N(X^i_{k_1,...k_6})}(X_j) \notag \\
&  -\int_{\mathbb{R}^6} f^{N}(\varphi^{1,N}_{s,0}(X^i_{k_1,...,k_6})-{\varphi^{1,N}_{s,0}}(Y)) \notag \\
& \hspace{3,5cm} \cdot \mathds 1_{G^N(X^i_{k_1,...k_6})}(Y)k_0(Y)
d^6Y\Big)ds\Big| > N^{-\frac{5}{12}}\ \vee \label{Def.B-Sbad.2} \\
& \Big|\int_{t_1'}^{t_2'}\Big(\frac{1}{N}\sum_{j\neq i}g^{N}(\varphi^{1,N}_{s,0}(X^i_{k_1,...,k_6})-{\varphi^{1,N}_{s,0}}(X_j))\mathds 1_{G^N(X^i_{k_1,...k_6})}(X_j) \notag \\
&  -\int_{\mathbb{R}^6}  g^N(\varphi^{1,N}_{s,0}(X^i_{k_1,...,k_6})-{\varphi^{1,N}_{s,0}}(Y)) \notag \\
& \hspace{3,5cm}\cdot \mathds 1_{G^N(X^i_{k_1,...k_6})}(Y)k_0(Y)
d^6Y\Big)ds\Big| > 1\ \label{Def.B-Sbad.1} 
\end{align}
Hence, statement \eqref{Def.B-Sbad.2} has the same structure as $\mathcal{B}^{N,\sigma}_{1,i}$ but note that in this case $X_i$ is replaced by the initial data of another auxiliary particle $X^i_{k_1,...,k_6}\coloneqq  X_i+\sum_{j=1}^6k_j N^{-\frac{5}{12}+\frac{\sigma}{2}}e_j$. 
For statement \eqref{Def.B-Sbad.1} a corresponding relationship holds, however with respect to $\mathcal{B}^{N,\sigma}_{2,i}$. 
It follows analogous, to the reasoning applied for the sets $\mathcal{B}^{N,\sigma}_{j,i}, \ j\in \{1,2\}$ that for any $\gamma>0$ there exists a $C_\gamma>0$ such that for all $N\in \N$
$$
\mathbb{P}\big(X\in \mathcal{B}^{N,\sigma}_{1,i,(k_1,...,k_6)}\big)\le C_\gamma N^{-\gamma}
.$$ 
By restricting the initial data to this set we can estimate Term \eqref{eq:sbad3} and the second term of \eqref{est.sb2}.
We are left with the considerations for the first term of \eqref{est.sb2} and  Term \eqref{eq:sbad1}. 
In the proof of case 1, see Section \ref{case1} the set $\mathcal{B}_{3,i}^{N,\sigma}$ was introduced to deal with the corresponding term of \eqref{est.sb2}.
Since the situation is basically the same we just have to modify the definition such that it applies for $X^i_{k_1,...,k_6}$ and for $(k_1,...,k_6)\in Q_N$:
\begin{align}
& X\in \mathcal{B}_{2,i,(k_1,...,k_6)}^{N,\sigma} \subseteq \mathbb{R}^{6N} \notag  \\
\Leftrightarrow & \exists l\in I_{\sigma}: \Big(R_l\neq \infty\ \land \notag  \\ &\sum_{j \in \mathcal{M}^N_s(X)\setminus \{i\}}\mathds 1_{M^N_{(r_l,R_l),(v_l,V_l)}(X^i_{k_1,...,k_6})}(X_j) \geq  N^{\frac{2\sigma}{9}}\big\lceil N^{\frac{2}{9}} R_l^2\min\big(\max(V_l,R_l),1\big)^4\big\rceil\Big) \ \vee \notag  \\
& \sum_{j \in \mathcal{M}^N_s(X)\setminus \{i\}}1\geq N^{\frac{2}{9}(1+\sigma)}  
\end{align}
For $X\in  \big(\mathcal{B}_{2,i,(k_1,\ldots,k_6)}^{N,\sigma} \big)^C$ and $t\in [t_1,t_2]$ the term 
\begin{align}
& \big|\int_{t_1}^{t}\frac{1}{N}\sum_{j\in \mathcal{M}^N_s(X)\setminus \{i\}}\Big(f^{N}([\Psi^{1,N}_{s,0}(X)]_j-[\Psi^{1,N}_{s,0}(X)]_i)\mathds 1_{G^N(\widetilde{X}_i)}(X_j) \notag   \\
&-f^{N}(\varphi^{1,N}_{s,0}(X_j)-{\varphi^{1,N}_{s,0}}(\widetilde{X}_i))\mathds 1_{G^N(\widetilde{X}_i)}(X_j)\Big) ds\big| 
\end{align}
can be estimated similar to case 1, see Section \ref{case1}. 
For this purpose, one has to take into account the choice of the interval $[t_1,t_2]$, because for this time span it holds that
\begin{align*}
& \sup_{t\in [t_1,t_2]}|[\Psi^{1,N}_{s,0}(X)]_j-{\varphi^{1,N}_{s,0}}(\widetilde{X}_i)|\le N^{-\frac{5}{12}+\sigma}\ \land \\
& \sup_{t\in [t_1,t_2]}|[\Psi^{2,N}_{s,0}(X)]_j-{^2\varphi^N_{s,0}}(\widetilde{X}_i)|\le N^{-\frac{1}{6}-\sigma}.
\end{align*}
The estimates from case 1, see Section \ref{case1}, can be copied to the current situation and hence the previously derived upper bound $CN^{-\frac{5}{12}}$ can be applied.\\
This concludes the considerations for Term \eqref{eq:sbad2}.  
Due to the definition of the set \eqref{Def.B-Sbad.1} and the subsequent reasoning it holds for configurations $X\in \big( \mathcal{B}^{N,\sigma}_{1,i,(k_1,\ldots,k_6)}\cup  \mathcal{B}^{N,\sigma}_{2,i,(k_1,\ldots,k_6)}\big)^C$ and $t\in [t_1,t_2]$ that
{\allowdisplaybreaks
\begin{align}
& \big|\int_{t_1}^{t}\frac{1}{N}\sum_{j\in \mathcal{M}^N_b(X)\setminus \{i\}}\Big(f^{N}([\Psi^{1,N}_{s,0}(X)]_i-[\Psi^{1,N}_{s,0}(X)]_j)\mathds 1_{G^N(\widetilde{X}_i)}(X_j) \notag   \\
&-f^{N}(\varphi^{1,N}_{s,0}(\widetilde{X}_i)-{\varphi^{1,N}_{s,0}}(X_j))\mathds 1_{G^N(\widetilde{X}_i)}(X_j)\Big) ds\big| \notag  \\
& +\int_{t_1}^{t}\frac{1}{N}\sum_{j\in \mathcal{M}^N_g(X)\setminus \{i\}}\Big( g^{N}(\varphi^{1,N}_{s,0}(\widetilde{X}_i)-{\varphi^{1,N}_{s,0}}(X_j))\mathds 1_{G^N(\widetilde{X}_i)}(X_j)\notag  \\
& \cdot \big(|[\Psi^{1,N}_{s,0}(X)]_i-{\varphi^{1,N}_{s,0}}(\widetilde{X}_i)|+|[\Psi^{1,N}_{s,0}(X)]_j-{\varphi^{1,N}_{s,0}}({X}_j)| \big)\Big) ds \notag \\
\le &  CN^{-\frac{5}{12}} \notag  \\
& +\Big(1+\int_{t_1}^{t}\int_{\mathbb{R}^6}  g^{N}(\varphi^{1,N}_{s,0}(\widetilde{X}_i)-{\varphi^{1,N}_{s,0}}(Y))k_0(Y)
d^6Yds\Big) \notag\\
& \cdot \sup_{s\in [t_1,t]}\Big(|[\Psi^{1,N}_{s,0}(X)]_i-{\varphi^{1,N}_{s,0}}(\widetilde{X}_i)|+\max_{j\in \mathcal{M}_g^N(X) }|[\Psi^{1,N}_{s,0}(X)]_j-{\varphi^{1,N}_{s,0}}({X}_j)| \Big) \notag \\
\le & CN^{-\frac{5}{12}} +C\big(1+(t-t_1)\ln(N)\big)N^{-\frac{5}{12}} \label{est.sb1}.
\end{align}}
The derivation of the upper bound for the first term was already discussed previously. 
For the upper bound of the second term we remind that $0\le g^N(q)\le C\min(N^{3\beta},\frac{1}{|q|^3})$ which leads to the factor $C\ln(N)$ after the integration. Further for $s\in [t_1,t]$ since  $t\in [t_1,t_2]\subseteq [t_1,\tau^N(X)]$ it holds that
$$|[\Psi^{1,N}_{s,0}(X)]_i-{\varphi^{1,N}_{s,0}}(\widetilde{X}_i)|+\max_{j\in \mathcal{M}_g^N(X) }|[\Psi^{1,N}_{s,0}(X)]_j-{\varphi^{1,N}_{s,0}}({X}_j)|\le 2N^{-\frac{5}{12}},$$ by the constraints on $t_2$ and the definition of the stopping time, see \ref{stoppingtime}).\\
\subsubsection{Estimate of Term \ref{eq:sbad1}}
In contrast to case 1, see Section \ref{case1}, the last remaining Term \eqref{eq:sbad1} has impact on the prove. 
It takes into account the impact of the `superhard' collisions with `superbad' or `hard'  with `bad' collision partners and is given by
\begin{align*}
 \left|\int_{t_1}^{t}\frac{1}{N}\sum_{j\neq i}f^{N}([\Psi^{1,N}_{s,0}(X)]_i-[\Psi^{1,N}_{s,0}(X)]_j)\mathds 1_{(G^N(\widetilde{X}_i))^C}(X_j)ds\right|.
\end{align*} 
The non-negligibility of this term is the first significant modification in contrast to the considerations for the `good' particles in case 1, see Section \ref{case1}. 
For this reason we introduce a set of inappropriate initial data for $(k_1,\ldots,k_6)\in Q_N$ and $i\in \{1,\ldots,N\}$
\begin{align}
\begin{split}
&X\in \mathcal{B}_{3,i,(k_1,\ldots,k_6)}^{N,\sigma}\subseteq \mathbb{R}^{6N} 
\Leftrightarrow   \sum_{j \neq i}\mathds 1_{M^N_{6N^{-\frac{1}{3}-\sigma},N^{-\frac{5}{18}}}(X^i_{k_1,\ldots,k_6})}(X_j)\geq  N^{\frac{\sigma}{2}}
\end{split} \label{def.B_3.superbad}
\end{align} 
It measures the amount of particles coming very close to the auxiliary particle cloud.
For configurations $X \notin  \mathcal{B}_{3,i,(k_1,\ldots,k_6)}^{N,\sigma}$ it holds that this last remaining term is bounded by
\begin{align} 
 CN^{\frac{\sigma}{2}-1}\|f^N\|_{\infty}|t-t_1| & \le  CN^{\frac{\sigma}{2}-1}\big(N^{\frac{5}{12}-\sigma}\big)^2|t-t_1| \notag \\
 & \le CN^{-\frac{1}{6}-\frac{3\sigma}{2}}|t-t_1|. \label{upp.bound.superbad.term}
 \end{align}
As $\mathbb{P}\big(Y\in \mathbb{R}^6:Y\notin G^N(X_i) \big)
\le   CN^{-\frac{5}{4}-2 \sigma}$  it follows that
 \begin{align}
& \mathbb{P}\big( X\in \mathcal{B}^{N,\sigma}_{3,i,(k_1,\ldots,k_6)}\big)\
\le  \binom{N}{\lceil N^{\frac{\sigma}{2}}\rceil}\big(  CN^{-\frac{5}{4}-2 \sigma}\big)^{\lceil N^{\frac{\sigma}{2}}\rceil}\le CN^{-\frac{1}{4}\lceil N^{\frac{\sigma}{2}}\rceil}.
 \end{align}
\subsubsection{Conclusion case 2 (labelled particle $X_i$ is superbad)}
All applied estimates work for arbitrary $t_1,t_2$ fulfilling the initially introduced demands
$$X \in \Big(\bigcup_{j\in \{1,2,3\}}\bigcup_{i=1}^N\bigcup_{(k_1,\ldots,k_6)\in Q_N}\mathcal{B}^{N,\sigma}_{j,i,(k_1,\ldots,k_6)}\Big)^C.$$
From now on we restrict ourselves to these good configurations.
We already discussed that for any $\gamma>0$ there exists a constant $C_\gamma>0$ such that $
\mathbb{P}\big(X\in \mathcal{B}^{N,\sigma}_{1,i,(k_1,\ldots,k_6)}\big)
\le C_{\gamma}N^{-\gamma}$ and according to the proof of the first case it holds that $\mathbb{P}\big(X\in \mathcal{B}^{N,\sigma}_{2,i,(k_1,\ldots,k_6)}\big) \le (CN^{-\frac{16\sigma}{9}} )^{\frac{N^{\frac{\sigma}{9}}}{2}}$. 
Since $|Q_N|\le (3\lceil N^\frac{1}{4}\rceil )^6 \le CN^{\frac{3}{2}} $ (see \eqref{Q_N sb}), it is possible to choose the constant $C_{\gamma}>0$ such that 
\begin{align*}
& \mathbb{P}\Big(\bigcup_{j\in \{1,2,3\}}\bigcup_{i=1}^N\bigcup_{(k_1,\ldots,k_6)\in Q_N}\mathcal{B}^{N,\sigma}_{j,i,(k_1,\ldots,k_6)} \Big) \le C_{\gamma}N^{-\gamma}
\end{align*}
holds for a given $\gamma>0$ and all $N\in \N$ and for all configurations
$$X\in \Big(\bigcup_{j\in \{1,2,3\}}\bigcup_{i=1}^N\bigcup_{(k_1,\ldots,k_6)\in Q_N}\mathcal{B}^{N,\sigma}_{j,i,(k_1,\ldots,k_6)} \Big)^C$$ 
all derived upper bounds are fulfilled for arbitrary `triples' $t_1,t_2$ and $\widetilde{X}_i$ provided they are chosen according to the introduced constraints on them. 
We obtain that Term \eqref{eq:sbad2} is bounded by $C(1+(t-t_1)\ln(N))N^{-\frac{5}{12}}$, see \eqref{est.b1}. The upper bound for Term \eqref{eq:sbad3} and Term \eqref{eq:sbad4} is given by $N^{-\frac{5}{12}}$. The upper bound for Term \eqref{eq:sbad1} is given by $CN^{-\frac{1}{6}-\frac{3\sigma}{2}}(t-t_1)$. 
It follows for $t\in [t_1,t_2]$ and for small enough $\sigma>0$ that the Term \eqref{force.bad case.sb} is bounded by
$$C\big(N^{-\frac{1}{6}-\frac{3\sigma}{2}}(t-t_1)+N^{-\frac{5}{12}} \big).$$ 
With $|[\Psi^{N}_{t_1,0}(X)]_i-{\varphi^{N}_{t_1,0}}(\widetilde{X}_i)|\leq \frac{N^{-\frac{5}{12}+\sigma}}{2}$ we obtain that for any $i\in \{1,\ldots ,N\}$ and for all times $t\in [t_1,t_2]$ the following inequality holds
\begin{align}
&|[\Psi^{2,N}_{t,0}(X)]_i-{\varphi^{2,N}_{t,0}}(\widetilde{X}_i)|\notag \\
\le & |[\Psi^{2,N}_{t_1,0}(X)]_i-{\varphi^{2,N}_{t_1,0}}(\widetilde{X}_i)|\notag \\
&+\big|\int_{t_1}^t\Big(\frac{1}{N}\sum_{j\neq i}f^{N}([\Psi^{1,N}_{s,0}(X)]_i-[\Psi^{1,N}_{s,0}(X)]_j) -f^{N}*\widetilde{k}^N_s({\varphi^{1,N}_{s,0}}(\widetilde{X}_i))\Big)ds\big| \notag  \\
\le & |[\Psi^{2,N}_{t_1,0}(X)]_i-{\varphi^{2,N}_{t_1,0}}(\widetilde{X}_i)| +C\big(N^{-\frac{1}{6}-\frac{3\sigma}{2}}(t-t_1)+N^{-\frac{5}{12}} \big)\\
&\le \frac{N^{-\frac{5}{12}+\sigma}}{2}  +C\big(N^{-\frac{1}{6}-\frac{3\sigma}{2}}(t-t_1)+N^{-\frac{5}{12}} \big)\label{best.t_2,vs}.
 \end{align} 
Now it is straightforward to find an upper bound for the spatial deviation for $t\in [t_1,t_2]$:
\begin{align}
& |[\Psi^{1,N}_{t,0}(X)]_i-{\varphi^{1,N}_{t,0}}(\widetilde{X}_i)|\notag \\
\le & |[\Psi^{1,N}_{t,0}(X)]_i-{\varphi^{1,N}_{t,0}}(\widetilde{X}_i)|+\int_{t_1}^t|[^2\Psi^{N}_{s,0}(X)]_i-{^2\varphi^{N}_{s,0}}(\widetilde{X}_i)|ds \notag \\
\le & \frac{N^{-\frac{5}{12}+\sigma}}{2}+  C\big(N^{-\frac{1}{6}-\frac{3\sigma}{2}}(t-t_1)^2+N^{-\frac{5}{12}}(t-t_1) \big).\label{best.t_2,xs}
\end{align}
The time $t_1$ denotes an arbitrary moment in $[0,\tau^N(X))$ before the stopping time is triggered. At this point in time we argued that it is always possible to find an auxiliary particle of the introduced `auxiliary cloud' which is closer in phase space to the observed `real' particle than $\frac{N^{-\beta}}{2}=\frac{N^{-\frac{5}{12}+\sigma }}{2}$. 
At time $t_2 \in (t_1,\tau^N(X)]$ the distance in (physical) space between this auxiliary particle and the `real' one still fulfils $$\sup_{t_1\le t \le t_2}|[\Psi^{1,N}_{t,0}(X)]_i-{\varphi^{1,N}_{t,0}}(\widetilde{X}_i)|\le N^{-\frac{5}{12}+\sigma},$$ while for the velocity deviation the much larger upper bound $$\sup_{t_1\le t \le t_2}|[\Psi^{2,N}_{t,0}(X)]_i-{\varphi^{2,N}_{t,0}}(\widetilde{X}_i)|\le N^{-\frac{1}{6}-\sigma}$$ was allowed. 
After that point in time maybe a new auxiliary particle of the `auxiliary cloud' which is closer to the observed `real' particle must be chosen for further estimates. 
The possible length of such an interval $[t_1,t_2]$, where the same auxiliary particle can be applied can be derived by \eqref{best.t_2,vs} and \eqref{best.t_2,xs}.

However, for large enough $N\in \N$ and $\sigma>0$ small enough the subsequent implication holds
\begin{align*}
&t-t_1\le N^{-\frac{1}{8}}
\Rightarrow \begin{cases}\frac{N^{-\frac{5}{12}+\sigma}}{2}+C\big(N^{-\frac{1}{6}-\frac{3\sigma}{2}}(t-t_1)^2+N^{-\frac{5}{12}}(t-t_1) \big)\le N^{-\frac{5}{12}+\sigma} \\
\frac{N^{-\frac{5}{12}+\sigma}}{2}+C\big(N^{-\frac{1}{6}-\frac{3\sigma}{2}}(t-t_1)+N^{-\frac{5}{12}} \big)\le CN^{-\frac{7}{24}-\frac{3\sigma}{2}}\le N^{-\frac{1}{6}-\sigma} \end{cases}
\end{align*}
and thus, according to relations \eqref{best.t_2,vs} and \eqref{best.t_2,xs}, the point in time $t_2\coloneqq  t_1+N^{-\frac{1}{8}}$ is a possible option such that the constraints on $t_2$ are fulfilled. 
Hence, bound \eqref{best.t_2,vs} and \eqref{best.t_2,xs} yield for $t_2$ and small enough $\sigma>0$ that
\begin{align*}
\sup_{t_1\le s \le t_2}|[\Psi^{N}_{t,0}(X)]_i-{\varphi^{N}_{t,0}}(\widetilde{X}_i)|\le CN^{-\frac{1}{6}-\frac{3\sigma}{2}}(t_2-t_1)=CN^{-\frac{7}{24}-\frac{3\sigma}{2}}.
\end{align*}
Considering estimate \eqref{est.eff.superbad.term} we obtain for $t\in [t_1,t_1+N^{-\frac{1}{8}}]$, the considered configurations, large enough $N$ and sufficiently small $\sigma>0$ that Term \eqref{superbaddouble.upp.bound} is bounded by
\begin{align}
& \sup_{t_1\le s \le t}|[\Psi^N_{s,0}(X)]_i-{\varphi^N_{s,0}}(X_i)|\notag \\
\le &\sup_{t_1\le s \le t}|[\Psi^N_{s,0}(X)]_i-{\varphi^N_{s,0}}(X^i_{k_1,\ldots,k_6})|+\sup_{t_1\le s \le t}|{\varphi^N_{s,0}}(X^i_{k_1,\ldots,k_6})-{\varphi^N_{s,0}}(X_i)|  \notag \\
 \le & CN^{-\frac{7}{24}-\frac{3\sigma}{2}}+ e^{C(t-t_1)}\big|[\Psi^N_{t_1,0}(X)]_i-\varphi^N_{t_1,0}({X}_i)\big|. 
\end{align}
The first point in time $t_1\in [0,\tau^N(X))$ was chosen arbitrarily and based on that we define a sequence of time steps 
$$t_n\coloneqq  nN^{-\frac{1}{8}}\text{ for }n\in \{0,\ldots,\lceil \tau^N(X)N^{\frac{1}{8}}\rceil-1\} \text{ and }t_{\lceil \tau^N(X)N^{\frac{1}{8}}\rceil}\coloneqq  \tau^N(X)$$
and thereby receive a corresponding sequence of inequalities 
 \begin{align*}
 \sup_{t_n\le s \le t_{n+1}}|[\Psi^N_{s,0}(X)]_i-{\varphi^N_{s,0}}(X_i)|\notag 
\le  CN^{-\frac{7}{24}-\frac{3\sigma}{2}}+ e^{CN^{-\frac{1}{8}}}\big|[\Psi^N_{t_n,0}(X)]_i-\varphi^N_{t_n,0}({X}_i)\big|.
\end{align*}
Inductively we derive that 
\begin{align*}
\sup_{0\le s \le t_n}|[\Psi^N_{s,0}(X)]_i-{\varphi^N_{s,0}}(X_i)|
\le CN^{-\frac{7}{24}-\frac{3\sigma}{2}} \sum_{k=0}^{n-1} e^{2CN^{-\frac{1}{8}}k}.
\end{align*}
An upper bound for the possible values of $n$ is given by $\lceil T N^{\frac{1}{8}}\rceil$ and this yields that
\begin{align*}
\sup_{0\le s \le \tau^N(X)}|[\Psi^N_{s,0}(X)]_i-{\varphi^N_{s,0}}(X_i)|\le CN^{-\frac{1}{6}-\frac{3}{2}\sigma}.
\end{align*}
For sufficiently large $N$ this value stays smaller than the allowed distance between the mean-field and the real trajectory $N^{-\frac{1}{6}-\sigma}$, which shows that also the `superbad' particles do typically not `trigger' the stopping time for the relevant $N$ and $\sigma$.

\subsection{Controlling the deviation of the bad particles}
Now we are left with the last set, the set of bad particles.
This intermediate set was defined as
$$\mathcal{M}_{b}^{N}(X)\coloneqq  \lbrace 1,\hdots,N\rbrace :\exists j\in \lbrace 1,\hdots ,N\rbrace\setminus \lbrace i\rbrace:X_j\in (M_{(r_b,v_b)}^{N}(X_j)\setminus M_{(r_s,v_s)}^{N}(X_j)).$$
The advantage of this set is that it contains less particles than the amount of good particles, but more than amount of superbad ones.
For particles in this set we allow intermediate deviation to their mean-field partners as bad events, i.e. particles coming close to each other, still occur.
We would also like to use the estimates of case 1, see Section \ref{case1}, and therefore we introduce the particle cloud which provides us the auxiliary particles like in case 2. 
This time $Q_N$ is given by
\begin{align}
Q_N \coloneqq   & \{-\lceil N^{\frac{1}{8}} \rceil ,\ldots,-1,0,1,\ldots, \lceil  N^{\frac{1}{8}} \rceil\}^6  \label{Q_N}
\end{align}
for $(k_1,\ldots,k_6)\in Q_N$ the positions $X^i_{k_1,...,k_6}\coloneqq  X_i+\sum_{j=1}^6 k_jN^{-\frac{5}{12}+\frac{\sigma}{2}}{e}_j$.
Let us apply Lemma \ref{Lemma distance same order} and the condition on the distance between the corresponding `real' and mean-field particle before the stopping time is `triggered'.
This gets us for the point in time $t_1\in [0,\tau^N(X)]$ and large enough $N$ that 
\begin{align*}
|\varphi^N_{0,t_1}([\Psi^N_{t_1,0}(X)]_i)-X_i|\le C|[\Psi^N_{t_1,0}(X)]_i-\varphi^N_{t_1,0}(X_i)|<CN^{-\frac{7}{24}-\sigma}.
\end{align*}
By construction, this distance is of smaller order with respect to $N$ than the diameter of the auxiliary `particle cloud' around $X_i$ and if $N$ is sufficiently large it is always possible to find a tuple $(k_1,\ldots,k_6)\in Q_N$ such that 
\begin{align}
|\varphi^N_{0,t_1}([\Psi^N_{t_1,0}(X)]_i)-X^i_{k_1,\ldots,k_6}|\le \frac{\sqrt{6}}{2} N^{-\frac{5}{12}-\frac{\sigma}{2}}. 
\end{align}
Lemma \ref{Lemma distance same order} implies in turn that
\begin{align}
|[\Psi^N_{t_1,0}(X)]_i-\varphi^N_{t_1,0}(X^i_{k_1,\ldots,k_6})|\le C  N^{-\frac{5}{12}-\frac{\sigma}{2}}. 
\end{align} 
For $C  N^{-\frac{5}{12}-\frac{\sigma}{2}}< \frac{1}{2}  N^{-\frac{5}{12}}$  with $N\in \N$  large enough, there exists a further point in time $t_2 \in (t_1, T]$ such that
$$\sup_{s\in [t_1,t_2]}|[\Psi^{1,N}_{s,0}(X)]_i-{\varphi^{1,N}_{s,0}(X^i_{k_1,\ldots,k_6})}|\le  N^{-\frac{5}{12}}$$
and the bound for the velocity deviation
$$\sup_{s\in [t_1,t_2]}|[\Psi^{2,N}_{s,0}(X)]_i-{^2\varphi^N_{s,0}(X^i_{k_1,\ldots,k_6})}|\le   N^{-\frac{1}{6}}$$
holds for $\sigma>0$ sufficiently small.
Like in the previous cases we have to show that $\sup_{t_1\le s \le t}|[\Psi^N_{s,0}(X)]_i-{\varphi^N_{s,0}}(X_i)|$ grows slow enough on this time interval. 
Since this variable is bounded by
\begin{align}
\sup_{t_1\le s \le t}|[\Psi^N_{s,0}(X)]_i-{\varphi^N_{s,0}}(X^i_{k_1,\ldots,k_6})|+\sup_{t_1\le s \le t}|{\varphi^N_{s,0}}(X^i_{k_1,\ldots,k_6})-{\varphi^N_{s,0}}(X_i)| \label{double.upp.bound}
\end{align}
and estimate the growth of these deviations instead.
 
The considerations for the first term is mostly analogous to the estimates of case 1 or case 2, see Section \ref{case1} and \ref{case2}. 
By construction, the spatial distance between the considered auxiliary particle and the `real' particle is bounded from above by $N^{-\frac{5}{12}+\sigma}$.\\
We use the abbreviation $\widetilde{X}_i\coloneqq  X^i_{k_1,\ldots,k_6}$ and assume for the rest of the proof that for an arbitrary point in time $t_1\in [0,\tau^N(X))$ and $X\in \mathbb{R}^{6N}$ the initial position of the auxiliary particle $X^i_{k_1,\ldots,k_6}$ and $t_2\in (t_1,\tau^N(X)]$ are chosen such that the previously introduced demands are fulfilled on $[t_1,t_2]$. \\
The second term has the same structure like \ref{est.eff.superbad.term} in case 2 an can be controlled by application of Lemma \ref{Lemma distance same order}. 
It follows for arbitrary $t\in [t_1,t_2]$ that
\begin{align}
& |{\varphi^N_{t,0}}(\widetilde{X}_i)-{\varphi^N_{t,0}}(X_i)| \notag \\
\le & e^{C(t-t_1)}|{\varphi^N_{t_1,0}}(\widetilde{X}_i)-{\varphi^N_{t_1,0}}(X_i)| \notag \\
\le & e^{C(t-t_1)}\big(\big|\varphi^N_{t_1,0}({X}_i)-[\Psi^N_{t_1,0}(X)]_i\big|+\big|[\Psi^N_{t_1,0}(X)]_i-\varphi^N_{t_1,0}(\widetilde{X}_i)\big|\big)\notag  \\
\le & e^{C(t-t_1)}\big(\big|\varphi^N_{t_1,0}({X}_i)-[\Psi^N_{t_1,0}(X)]_i\big|+N^{-\frac{5}{12}}\big), \label{est.eff.bad.term}
\end{align}
where we regarded the allowed upper bound for $\big|[\Psi^N_{t_1,0}(X)]_i-\varphi^N_{t_1,0}(\widetilde{X}_i)\big|$ according to the choice of $\widetilde{X}_i$. 
We will return to this term later.

For the second term we remark that  
\begin{align}
&  |[\Psi^{2,N}_{t,0}(X)]_i-{^2\varphi^N_{t,0}}(\widetilde{X}_i)|  \notag \\
\le  & |[\Psi^{2,N}_{t_1,0}(X)]_i-{\varphi^{2,N}_{t_1,0}}(\widetilde{X}_i)|  \notag \\
& + \big|\int_{t_1}^{t}\Big(\frac{1}{N}\sum_{j\neq i}f^{N}([\Psi^{1,N}_{s,0}(X)]_i-[\Psi^{1,N}_{s,0}(X)]_j)-f^{N}*\widetilde{k}_s(\varphi^{1,N}_{s,0}(\widetilde{X}_i))\Big)ds\big|. 
\end{align}
The second summand can be estimated by multiple applications of the triangle inequality and we essentially obtain the four terms of case 1 or 2, Section \ref{case1} and \ref{case2}.
{\allowdisplaybreaks
\begin{align}
& \big|\int_{t_1}^{t}\frac{1}{N}\sum_{j\neq i}f^{N}([\Psi^{1,N}_{s,0}(X)]_i-[\Psi^{1,N}_{s,0}(X)]_j)-f^{N}*\widetilde{k}_s(\varphi^{1,N}_{s,0}(\widetilde{X}_i))ds\big| \label{force.bad case.b}\\
\le &  \big|\int_{t_1}^{t}\frac{1}{N}\sum_{j\neq i}f^{N}([\Psi^{1,N}_{s,0}(X)]_i-[\Psi^{1,N}_{s,0}(X)]_j)\mathds 1_{(G^N(\widetilde{X}_i))^C}(X_j)ds\big| \label{eq:bad1} \\
&+  \big|\int_{t_1}^{t}\frac{1}{N}\sum_{j\neq i}\Big(
f^{N}([\Psi^{1,N}_{s,0}(X)]_i-[\Psi^{1,N}_{s,0}(X)]_j)\mathds 1_{G^N(\widetilde{X}_i)}(X_j) \notag  \\
&-f^{N}(\varphi^{1,N}_{s,0}(\widetilde{X}_i)-{\varphi^{1,N}_{s,0}}(X_j))\mathds 1_{G^N(\widetilde{X}_i)}(X_j)\Big) ds\big| \label{eq:bad2} \\
& +\big| \int_{t_1}^{t} \frac{1}{N}\sum_{j\neq i}f^{N}(\varphi^{1,N}_{s,0}(\widetilde{X}_i)-{\varphi^{1,N}_{s,0}}(X_j))\mathds 1_{G^N(\widetilde{X}_i)}(X_j)ds  \notag\\
&-\int_{t_1}^{t}\int_{\mathbb{R}^6}f^N(\varphi^{1,N}_{s,0}(\widetilde{X}_i)-{\varphi^{1,N}_{s,0}}(Y))\mathds 1_{G^N(\widetilde{X}_i)}(Y)k_0(Y)d^6Yds\big| \label{eq:bad3}  \\
& +\big|\int_{t_1}^{t}\Big(\int_{\mathbb{R}^6}f^N(\varphi_{s,0}^{1,N}(\widetilde{X}_i)-{\varphi^{1,N}_{s,0}}(Y))\mathds 1_{G^N(\widetilde{X}_i)}(Y)k_0(Y)d^6Y  \notag\\
&-f^{N}*\widetilde{k}^N_s(\varphi^{1,N}_{s,0}(\widetilde{X}_i))\Big)ds\big| \label{eq:bad4}
\end{align} }
\subsubsection{Estimate of Term \ref{eq:bad4}}
A suitable upper bound for Term \eqref{eq:bad4} can be derived analogously to the previous two cases and thus is given by $ CN^{-\frac{5}{12}}$.
\begin{align}
& \big|\int_{t_1}^t\int_{ \mathbb{R}^6}f^{N}(\varphi^{1,N}_{s,0}(\widetilde{X}_i)-{\varphi^{1,N}_{s,0}}(Y))k_0(Y)\mathds 1_{G^N(\widetilde{X}_i)}(Y)d^6Yds \notag \\
&-\int_{t_1}^t\int_{ \mathbb{R}^6}f^{N}({\varphi^{1,N}_{s,0}}(\widetilde{X}_i)-{\varphi^{1,N}_{s,0}}(Y))k_0(Y)d^6Yds\big| \notag \\
= & \big|\int_{t_1}^t\int_{ \mathbb{R}^6}f^{N}(\varphi^{1,N}_{s,0}(\widetilde{X}_i)-{\varphi^{1,N}_{s,0}}(Y))k_0(Y)(\mathds 1_{G^N(\widetilde{X}_i)}(Y)-1)d^6Yds\big| \notag \\
\le & T\|f^{N}\|_{\infty}\int_{ \mathbb{R}^6}\mathds 1_{(G^N(\widetilde{X}_i))^C}(Y)k_0(Y)d^6Y \notag \\
\le & TN^{2\beta}\mathbb{P}\big(Y \in \mathbb{R}^6:Y\notin G^N(\widetilde{X}_i) \big) \le TN^{2\beta} \mathbb{P}\big(Y\in \mathbb{R}^6:Y\in M^N_{6N^{-b_r },N^{-b_v}}(\widetilde{X}_i) \big) \notag \\
 \le &TN^{2\beta} C(N^{-b_r})^2(N^{-b_v})^4 \notag \\
\le & TN^{\frac{5}{12}-2\sigma}
\end{align}

\subsubsection{Estimate of Term \ref{eq:bad2} and Term \ref{eq:bad3}}

Let us focus on the two Terms \eqref{eq:bad2} and \eqref{eq:bad3}, as both can be estimated by Theorem \ref{Prop:LLN}. 
Since according to the choice of $t_1, t_2$ and $\widetilde{X}_i$ it holds that $\sup_{t_1\le s \le t_2}|[\Psi^{1,N}_{s,0}(X)]_i-{\varphi^{1,N}_{s,0}}(\widetilde{X}_i)|\le N^{-\frac{5}{12}+\sigma}$.
It follows by estimating with the map $g^N$ that
\begin{align}
& \big|\int_{t_1}^{t}\frac{1}{N}\sum_{j\neq i}\Big(f^{N}([\Psi^{1,N}_{s,0}(X)]_i-[\Psi^{1,N}_{s,0}(X)]_j)\mathds 1_{G^N(\widetilde{X}_i)}(X_j) \notag  \\
&-f^{N}(\varphi^{1,N}_{s,0}(\widetilde{X}_i)-{\varphi^{1,N}_{s,0}}(X_j))\mathds 1_{G^N(\widetilde{X}_i)}(X_j)\Big) ds\big| \label{good force,bad case} \\
\le  &  \big|\int_{t_1}^{t}\frac{1}{N}\sum_{j\in \mathcal{M}^N_b(X)\setminus \{i\}}\Big(f^{N}([\Psi^{1,N}_{s,0}(X)]_i-[\Psi^{1,N}_{s,0}(X)]_j)\mathds 1_{G^N(\widetilde{X}_i)}(X_j) \notag   \\
&-f^{N}(\varphi^{1,N}_{s,0}(\widetilde{X}_i)-{\varphi^{1,N}_{s,0}}(X_j))\mathds 1_{G^N(\widetilde{X}_i)}(X_j)\Big) ds\big| \notag  \\
& +\int_{t_1}^{t}\frac{1}{N}\sum_{j\in \mathcal{M}^N_g(X)\setminus \{i\}}\Big( g^{N}(\varphi^{1,N}_{s,0}(\widetilde{X}_i)-{\varphi^{1,N}_{s,0}}(X_j))\mathds 1_{G^N(\widetilde{X}_i)}(X_j) \notag  \\
& \cdot \big(|[\Psi^{1,N}_{s,0}(X)]_i-{\varphi^{1,N}_{s,0}}(\widetilde{X}_i)|+|[\Psi^{1,N}_{s,0}(X)]_j-{\varphi^{1,N}_{s,0}}({X}_j)| \big)\Big) ds. \label{est.b2}
\end{align}
All these terms have basically the same structure as in case 1 or 2, see Section \ref{case1} and \ref{case2}. 
We just have to amend the definitions of the sets $\mathcal{B}^{N,\sigma}_{i,j}$ from the previous to the current situation. 
We define for $(k_1,\ldots,k_6)\in Q_N$
\begin{align}
& X\in \mathcal{B}_{1,i,(k_1,\ldots,k_6)}^{N,\sigma}\subseteq \mathbb{R}^{6N} \notag\\
\Leftrightarrow & \exists t_1',t_2'\in [0,T]:\notag \\
& \Big|\int_{t_1'}^{t_2'}\Big(\frac{1}{N}\sum_{j\neq i}f^{N}(\varphi^{1,N}_{s,0}(X^i_{k_1,\ldots,k_6})-{\varphi^{1,N}_{s,0}}(X_j))\mathds 1_{G^N(X^i_{k_1,\ldots,k_6})}(X_j) \notag \\
&  -\int_{\mathbb{R}^6} f^{N}(\varphi^{1,N}_{s,0}(X^i_{k_1,\ldots,k_6})-{\varphi^{1,N}_{s,0}}(Y)) \notag \\
& \hspace{3,5cm} \cdot \mathds 1_{G^N(X^i_{k_1,\ldots,k_6})}(Y)k_0(Y)
d^6Y\Big)ds\Big| > N^{-\beta+\sigma}\ \vee \label{Def.B-bad.2} \\
& \Big|\int_{t_1'}^{t_2'}\Big(\frac{1}{N}\sum_{j\neq i}g^{N}(\varphi^{1,N}_{s,0}(X^i_{k_1,...,k_6})-{\varphi^{1,N}_{s,0}}(X_j))\mathds 1_{G^N(X^i_{k_1,...k_6})}(X_j) \notag \\
&  -\int_{\mathbb{R}^6}  g^N(\varphi^{1,N}_{s,0}(X^i_{k_1,...,k_6})-{\varphi^{1,N}_{s,0}}(Y)) \notag \\
& \hspace{3,5cm}\cdot \mathds 1_{G^N(X^i_{k_1,\ldots,k_6})}(Y)k_0(Y)
d^6Y\Big)ds\Big| > 1.\ \label{Def.B-bad.1} 
\end{align}
For the second statement \eqref{Def.B-bad.1} we proceed similarly. 
It follows analogous to the reasoning applied for the sets $\mathcal{B}^{N,\sigma}_{j,i}, \ j\in \{1,2\}$ that for any $\gamma>0$ there exists a $C_\gamma>0$ such that for all $N\in \N$
$$
\mathbb{P}\big(X\in \mathcal{B}^{N,\sigma}_{1,i,(k_1,\ldots,k_6)}\big)\le C_\gamma N^{-\gamma}
.$$ 
Like in case 1 or 2, see Section \ref{case1} and \ref{case2}, restricting the initial data to this set is already enough to handle Term \eqref{eq:bad3} and the second term of \eqref{est.b2}.
Thus, we continue with the first term of \eqref{est.b2} and finally deal with Term \eqref{eq:bad4}. 
Therefore we modify the definition of the set $\mathcal{B}_{3,i}^{N,\sigma}$ such that it applies for $X^i_{k_1,\ldots,k_6}$ for $(k_1,\ldots,k_6)\in Q_N$
\begin{align}
& X\in \mathcal{B}_{2,i,(k_1,\ldots,k_6)}^{N,\sigma} \subseteq \mathbb{R}^{6N} \notag  \\
\Leftrightarrow & \exists l\in I_{\sigma}: \Big(R_l\neq \infty\ \land \notag  \\ &\sum_{j \in \mathcal{M}^N_b(X)\setminus \{i\}}\mathds 1_{M^N_{(r_l,R_l),(v_l,V_l)}(X^i_{k_1,\ldots,k_6})}(X_j) \geq  N^{(2-2b_v-4b_r)\sigma}\\&\qquad\qquad\qquad\big\lceil N^{2-2b_v-4b_r} R_l^2\min\big(\max(V_l,R_l),1\big)^4\big\rceil\Big) \ \vee \notag  \\
& \sum_{j \in \mathcal{M}^N_b(X)\setminus \{i\}}1\geq N^{\frac{3}{4}(1+\sigma)}  
\end{align}
For $X\in  \big(\mathcal{B}_{2,i,(k_1,\ldots,k_6)}^{N,\sigma} \big)^C$ and $t\in [t_1,t_2]$ the term 
\begin{align}
& \big|\int_{t_1}^{t}\frac{1}{N}\sum_{j\in \mathcal{M}^N_b(X)\setminus \{i\}}\Big(f^{N}([\Psi^{1,N}_{s,0}(X)]_j-[\Psi^{1,N}_{s,0}(X)]_i)\mathds 1_{G^N(\widetilde{X}_i)}(X_j) \notag   \\
&-f^{N}(\varphi^{1,N}_{s,0}(X_j)-{\varphi^{1,N}_{s,0}}(\widetilde{X}_i))\mathds 1_{G^N(\widetilde{X}_i)}(X_j)\Big) ds\big| 
\end{align}
can be handled by the same estimates as in case 1, see Section \ref{case1}. 
For this purpose, one has to take into account the choice of the interval $[t_1,t_2]$ because for this time span it holds that
\begin{align*}
& \sup_{t\in [t_1,t_2]}|[\Psi^{1,N}_{s,0}(X)]_j-{\varphi^{1,N}_{s,0}}(\widetilde{X}_i)|\le N^{-\frac{5}{12}}\ \land \\
& \sup_{t\in [t_1,t_2]}|[\Psi^{2,N}_{s,0}(X)]_j-{\varphi^{2,N}_{s,0}}(\widetilde{X}_i)|\le N^{-\frac{1}{6}}.
\end{align*}
The estimates can be copied form the previous cases and hence also the previously derived upper bound $CN^{-\frac{5}{12}}$ can be applied.

This concludes the considerations for Term \eqref{eq:bad2} and Term \eqref{good force,bad case}. Due to Definition \eqref{Def.B-bad.1} and the subsequent reasoning it holds for configurations $X\in \big( \mathcal{B}^{N,\sigma}_{1,i,(k_1,\ldots,k_6)}\cup  \mathcal{B}^{N,\sigma}_{2,i,(k_1,\ldots,k_6)}\big)^C$ and $t\in [t_1,t_2]$ that
{\allowdisplaybreaks
\begin{align}
& \big|\int_{t_1}^{t}\frac{1}{N}\sum_{j\in \mathcal{M}^N_b(X)\setminus \{i\}}\Big(f^{N}([\Psi^{1,N}_{s,0}(X)]_i-[\Psi^{1,N}_{s,0}(X)]_j)\mathds 1_{G^N(\widetilde{X}_i)}(X_j) \notag   \\
&-f^{N}(\varphi^{1,N}_{s,0}(\widetilde{X}_i)-{\varphi^{1,N}_{s,0}}(X_j))\mathds 1_{G^N(\widetilde{X}_i)}(X_j)\Big) ds\big| \notag  \\
& +\int_{t_1}^{t}\frac{1}{N}\sum_{j\in \mathcal{M}^N_g(X)\setminus \{i\}}\Big( g^{N}(\varphi^{1,N}_{s,0}(\widetilde{X}_i)-{\varphi^{1,N}_{s,0}}(X_j))\mathds 1_{G^N(\widetilde{X}_i)}(X_j)\notag  \\
& \cdot \big(|[\Psi^{1,N}_{s,0}(X)]_i-{\varphi^{1,N}_{s,0}}(\widetilde{X}_i)|+|[\Psi^{1,N}_{s,0}(X)]_j-{\varphi^{1,N}_{s,0}}({X}_j)| \big)\Big) ds \notag \\
\le &  CN^{-\frac{5}{12}} \notag  \\
& +\Big(1+\int_{t_1}^{t}\int_{\mathbb{R}^6}  g^{N}(\varphi^{1,N}_{s,0}(\widetilde{X}_i)-{\varphi^{1,N}_{s,0}}(Y))k_0(Y)
d^6Yds\Big) \notag\\
& \cdot \sup_{s\in [t_1,t]}\Big(|[\Psi^{1,N}_{s,0}(X)]_i-{\varphi^{1,N}_{s,0}}(\widetilde{X}_i)|+\max_{j\in \mathcal{M}_g^N(X) }|[\Psi^{1,N}_{s,0}(X)]_j-{\varphi^{1,N}_{s,0}}({X}_j)| \Big) \notag \\
\le & CN^{-\frac{5}{12}+\sigma} +C\big(1+(t-t_1)\ln(N)\big)N^{-\frac{5}{12}} \label{est.b1}
\end{align}}
The upper bound for the first summand was already discussed in the previously part. 
For the second summand we regarded that $0\le g^N(q)\le C\min(N^{-\frac{5}{12}},\frac{1}{|q|^3})$. This leads to the factor $C\ln(N)$ after the integration. Further it holds for $s\in [t_1,t]$ due to $t\in [t_1,t_2]\subseteq [t_1,\tau^N(X)]$, that
$$|[\Psi^{1,N}_{s,0}(X)]_i-{\varphi^{1,N}_{s,0}}(\widetilde{X}_i)|+\max_{j\in \mathcal{M}_g^N(X) }|[\Psi^{1,N}_{s,0}(X)]_j-{\varphi^{1,N}_{s,0}}({X}_j)|\le 2N^{-\frac{5}{12}}$$ by the constraints on $t_2$ and the definition of the stopping time.

\subsubsection{Estimate of Term \ref{eq:bad1}}

We finally arrived at the last remaining Term \eqref{eq:bad1}
\begin{align*}
 \big|\int_{t_1}^{t}\frac{1}{N}\sum_{j\neq i}f^{N}([\Psi^{1,N}_{s,0}(X)]_i-[\Psi^{1,N}_{s,0}(X)]_j)\mathds 1_{(G^N(\widetilde{X}_i))^C}(X_j)ds\big|
\end{align*}
Remember that $i \in \mathcal{M}_{b}^{N}(X)$.
This term takes into account the impact of the `hard' collisions which were excluded for the `good' particles. 
But `superhard' collisions are excluded again like in case 1, see Section \ref{case1}, because the considered particle $X_i$ is `bad'. 
That simplifies the situation for us to
\begin{align*}
 &\big|\int_{t_1}^{t}\frac{1}{N}\sum_{j\neq i}f^{N}([\Psi^{1,N}_{s,0}(X)]_i-[\Psi^{1,N}_{s,0}(X)]_j)\mathds 1_{(G^N(\widetilde{X}_i))^C\setminus M_{sb}(\widetilde{X}_i) }(X_j)ds\big|=\\
& \big|\int_{t_1}^{t}\frac{1}{N}\sum_{j\neq i}f^{N}([\Psi^{1,N}_{s,0}(X)]_i-[\Psi^{1,N}_{s,0}(X)]_j)\mathds 1_{ M_{b}(\widetilde{X}_i) }(X_j)ds\big|.
\end{align*}

Fortunately, the estimates for this remaining term are straightforward and a simple application of Corollary \ref{corollary phi and  psi} but first we need to define a set of inappropriate initial data for $(k_1,...,k_6)\in Q_N$ and $i\in \{1,...,N\}$:
\begin{align}
\begin{split}
& X\in \mathcal{B}_{3,i,(k_1,...,k_6)}^{N,\sigma}\subseteq \mathbb{R}^{6N} \\
\Leftrightarrow &  \sum_{j \neq i}\mathds 1_{M^N_{6N^{-\frac{7}{24}-\sigma},N^{-\frac{1}{6}}}(X^i_{k_1,...,k_6})}(X_j)\geq  N^{\frac{3\sigma}{4}}
\end{split} \label{def.B_3.bad}
\end{align} 
It follows for configurations $X \notin  \mathcal{B}_{3,i,(k_1,...,k_6)}^{N,\sigma}$ that
\begin{align*}
& \big|\int_{t_1}^{t}\frac{1}{N}\sum_{j\neq i}f^{N}([\Psi^{1,N}_{s,0}(X)]_i-[\Psi^{1,N}_{s,0}(X)]_j)\mathds 1_{ M_{b}(\widetilde{X}_i) }(X_j)ds\big|\\
&\le\frac{N^{\frac{3\sigma}{4}}}{N} C\min\big(\frac{1}{N^{-\beta}\Delta v}, \frac{1}{\min\limits_{0\le s\le T}|[^1\Psi^{N,\beta}_{s,0}(X)]_i-[^1\Psi^{N,\beta}_{s,0}(X)]_j|\Delta v}\big)\\
&\le\frac{N^{\frac{3\sigma}{4}}}{N} C\min\big(\frac{1}{N^{-\beta}N^{-\frac{1.5}{9}}}, \frac{1}{N^{-\frac{7}{24}-\sigma}N^{-\frac{1.5}{9}}}\big)\le CN^{-\frac{7}{8}+\frac{3\sigma}{4}}.
\end{align*}

This last remaining term is bounded by
\begin{align} 
 CN^{-\frac{7}{8}+\frac{3\sigma}{4}} \label{upp.bound.sbad.term}
 \end{align}
Moreover, by taking into account that $\mathbb{P}\big(Y\in \mathbb{R}^6:Y\notin G^N(X_i) \big)
\le   CN^{-\frac{5}{4}-2 \sigma}$  
it follows that
\begin{align}
& \mathbb{P}\big( X\in \mathcal{B}^{N,\sigma}_{3,i,(k_1,...,k_6)}\big)\
\le  \binom{N}{\lceil N^{\frac{3\sigma}{4}}\rceil}\big(  CN^{-\frac{5}{4}-2 \sigma}\big)^{\lceil N^{\frac{3\sigma}{4}}\rceil}\le CN^{-\frac{1}{4}\lceil N^{\frac{3\sigma}{4}}\rceil}
 \end{align}
which obviously drops sufficiently fast.

\subsubsection{Conclusion case 3 (labelled particle $X_i$ is bad)}
Analogously to case 2, see Section \ref{case2} we have to merge all upper bounds. 
All applied estimates work for arbitrary $t_1,t_2$ fulfilling the initially in 
$$X \in \Big(\bigcup_{j\in \{1,2,3\}}\bigcup_{i=1}^N\bigcup_{(k_1,...,k_6)\in Q_N}\mathcal{B}^{N,\sigma}_{j,i,(k_1,...,k_6)}\Big)^C.$$
We already discussed that for any $\gamma>0$ there exists a constant $C_\gamma>0$ such that $
\mathbb{P}\big(X\in \mathcal{B}^{N,\sigma}_{1,i,(k_1,...,k_6)}\big)
\le C_{\gamma}N^{-\gamma}$ and according to the proof of the first case it holds that $\mathbb{P}\big(X\in \mathcal{B}^{N,\sigma}_{2,i,(k_1,...,k_6)}\big) \le (CN^{-\frac{7\sigma}{3}} )^{\frac{N^{\frac{\sigma}{3}}}{2}}$, see \eqref{prob.b.3b}. 
Since $|Q_N|\le (3\lceil N^\frac{1}{8}\rceil )^6 \le CN $, see \eqref{Q_N}, it is possible to choose the constant $C_{\gamma}>0$ such that 
\begin{align*}
& \mathbb{P}\Big(\bigcup_{j\in \{1,2,3\}}\bigcup_{i=1}^N\bigcup_{(k_1,...,k_6)\in Q_N}\mathcal{B}^{N,\sigma}_{j,i,(k_1,...,k_6)} \Big) \le C_{\gamma}N^{-\gamma}
\end{align*}
holds for a given $\gamma>0$ and all $N\in \N$. 
For arbitrary `triples' $t_1,t_2$ and $\widetilde{X}_i$ all derived upper bounds are fulfilled for configurations
$$X\in \Big(\bigcup_{j\in \{1,2,3\}}\bigcup_{i=1}^N\bigcup_{(k_1,...,k_6)\in Q_N}\mathcal{B}^{N,\sigma}_{j,i,(k_1,...,k_6)} \Big)^C,$$ provided they are chosen according to the introduced constraints. 
We obtain that \eqref{eq:bad2} is bounded by $C(1+(t-t_1)\ln(N))N^{-\frac{5}{12}+\sigma}$, the bound for Term \eqref{eq:bad3} is $N^{-\frac{5}{12}+\sigma}$ by definition, the bound for \eqref{eq:bad4} is $CN^{-\frac{5}{12}}$, as derived in case 1, see Section \ref{case1} and $CN^{-\frac{7}{8}+\frac{3\sigma}{4}}$ constitutes an upper bound for \eqref{eq:bad1}. 
Hence the force term can be estimated by
$$CN^{-\frac{5}{12}+\sigma}.$$
With $|[\Psi^{N}_{t_1,0}(X)]_i-{\varphi^{N}_{t_1,0}}(\widetilde{X}_i)|\leq \frac{N^{-\frac{5}{12}+\sigma}}{2}$
for $t\in [t_1,t_2]$ and $\sigma>0$ for small enough. We obtain that for any $i\in \{1,\ldots ,N\}$ and for all times $t\in [t_1,t_2]$ the following holds
\begin{align}
&|[\Psi^{2,N}_{t,0}(X)]_i-{\varphi^{2,N}_{t,0}}(\widetilde{X}_i)|\notag \\
\le & |[\Psi^{2,N}_{t_1,0}(X)]_i-{\varphi^{2,N}_{t_1,0}}(\widetilde{X}_i)|\notag \\
&+\left|\int_{t_1}^t\Big(\frac{1}{N}\sum_{j\neq i}f^{N}([\Psi^{1,N}_{s,0}(X)]_i-[\Psi^{1,N}_{s,0}(X)]_j) -f^{N}*\widetilde{k}^N_s({\varphi^{1,N}_{s,0}}(\widetilde{X}_i))\Big)ds\right| \notag  \\
\le & \left|[\Psi^{2,N}_{t_1,0}(X)]_i-{\varphi^{2,N}_{t_1,0}}(\widetilde{X}_i)\right| +CN^{-\frac{5}{12}} \\
\leq & \frac{N^{-\frac{5}{12}+\sigma}}{2} +CN^{-\frac{5}{12}} 
\label{best.t_2,v}
 \end{align} 
Now it is straightforward to find an upper bound for the spatial deviation for $t\in [t_1,t_2]$:
\begin{align}
& |[\Psi^{1,N}_{t,0}(X)]_i-{\varphi^{1,N}_{t,0}}(\widetilde{X}_i)|\notag \\
\le & |[\Psi^{1,N}_{t,0}(X)]_i-{\varphi^{1,N}_{t,0}}(\widetilde{X}_i)|+\int_{t_1}^t|[^2\Psi^{N}_{s,0}(X)]_i-{^2\varphi^{N}_{s,0}}(\widetilde{X}_i)|ds \notag \\
\le & \frac{N^{-\frac{5}{12}+\sigma}}{2}+  C\big(N^{-\frac{5}{12}-\frac{\sigma}{2}}(t-t_1)\big)\label{best.t_2,x}
\end{align}
It is always possible to find an auxiliary particle of the introduced `cloud' which is closer in phase space to the observed `real' particle due to previous considerations.
After the time $t_2$ it may be necessary for further estimates to choose a new auxiliary particle of the `cloud' which is closer to the observed `real' particle. 
For large enough $N\in \N$ and small enough $\sigma,\delta>0$ the subsequent implication holds
\begin{align*}
&t-t_1\le N^{-\delta}
\Rightarrow \begin{cases}\frac{N^{-\frac{5}{12}}}{2}+C\big(N^{-\frac{5}{12}-\frac{\sigma}{2}}(t-t_1) \big)\le N^{-\frac{5}{12}} \\
\frac{N^{-\frac{5}{12}}}{2}+C\big(N^{-\frac{5}{12}-\frac{1\sigma}{2}}(t-t_1) \big)\le CN^{-\frac{5}{12}} \end{cases}
\end{align*}
and thus according to relations \eqref{best.t_2,v} and \eqref{best.t_2,x}, $t_2\coloneqq  t_1+N^{-\delta}$ is a possible option such that the constraints on $t_2$ are fulfilled. 
Hence, relation \eqref{best.t_2,v} and \eqref{best.t_2,x} yield for this choice of $t_2$ and small enough $\sigma>0$ that
\begin{align*}
\sup_{t_1\le s \le t_2}|[\Psi^{N}_{t,0}(X)]_i-{\varphi^{N}_{t,0}}(\widetilde{X}_i)|\le CN^{-\frac{5}{12}-\frac{3\sigma}{2}}(t_2-t_1)=CN^{-\frac{5}{12}-\delta-\frac{3\sigma}{2}}.
\end{align*}
For Term \eqref{double.upp.bound} and by additionally considering estimate \eqref{est.eff.bad.term}, we obtain for $t\in [t_1,t_1+N^{-\delta}]$, the considered configurations, large enough $N$ and $\sigma>0$ small enough that
\begin{align}
& \sup_{t_1\le s \le t}|[\Psi^N_{s,0}(X)]_i-{\varphi^N_{s,0}}(X_i)|\notag \\
\le &\sup_{t_1\le s \le t}|[\Psi^N_{s,0}(X)]_i-{\varphi^N_{s,0}}(X^i_{k_1,...,k_6})|+\sup_{t_1\le s \le t}|{\varphi^N_{s,0}}(X^i_{k_1,...,k_6})-{\varphi^N_{s,0}}(X_i)|  \notag \\
 \le & CN^{-\frac{5}{12}-\delta-\frac{3\sigma}{2}}+ e^{C(t-t_1)}\big|[\Psi^N_{t_1,0}(X)]_i-\varphi^N_{t_1,0}({X}_i)\big|. 
\end{align}
Since the point in time $t_1\in [0,\tau^N(X))$ before the stopping time was triggered was chosen arbitrarily, we can define a sequence of time steps 
$$t_n\coloneqq  nN^{-\delta}\text{ for }n\in \{0,...,\lceil \tau^N(X)N^{\delta}\rceil-1\} \text{ and }t_{\lceil \tau^N(X)N^{\delta}\rceil}\coloneqq  \tau^N(X).$$
Thus we receive a corresponding sequence of inequalities 
 \begin{align*}
& \sup_{t_n\le s \le t_{n+1}}|[\Psi^N_{s,0}(X)]_i-{\varphi^N_{s,0}}(X_i)|\notag \\
\le & CN^{-\frac{5}{12}-\delta-\frac{3\sigma}{2}}+ e^{CN^{-\delta}}\big|[\Psi^N_{t_n,0}(X)]_i-\varphi^N_{t_n,0}({X}_i)\big|.
\end{align*}
Inductively we derive that 
\begin{align*}
\sup_{0\le s \le t_n}|[\Psi^N_{s,0}(X)]_i-{\varphi^N_{s,0}}(X_i)|
\le CN^{-\frac{5}{12}-\delta-\frac{3\sigma}{2}} \sum_{k=0}^{n-1} e^{2CN^{-\delta}k} .
\end{align*}
An upper bound for the possible values of $n$ is given by $\lceil T N^{\delta}\rceil$ and this yields that
\begin{align*}
\sup_{0\le s \le \tau^N(X)}|[\Psi^N_{s,0}(X)]_i-{\varphi^N_{s,0}}(X_i)|\le CN^{-\frac{5}{12}-\delta-\frac{3}{2}\sigma}.
\end{align*}
For sufficient large $N$ this value stays smaller than the allowed distance between the mean-field and the real trajectory $N^{-\frac{7}{24}-\sigma}$, which shows that also the `bad' particles do typically not `trigger' the stopping time for the relevant $N$ and $\sigma$. \\
This finally completes the main part of the proof.

We conclude the proof of Theorem \ref{maintheorem} by showing that for $N>1$
\begin{align}
\sup\limits_{x\in \mathbb{R}^6}\sup_{0\le s \le T}|\varphi_{s,0}^{1,N}(x)-{\varphi^{1,\infty}_{s,0}}(x)|\le e^{C\sqrt{\ln(N)}}N^{-2\beta} 
\end{align}
which is smaller than necessary.

\section{Molecular of chaos}
As mentioned in Section \ref{Heuristics} and analogously to Chapter \ref{Chapter: VDB}, we finally prove Theorem \ref{maintheorem} by showing that
\begin{align}
\Delta_N(t)\coloneqq \sup_{x\in \mathbb{R}^6}\sup_{0\le s \le T}|\varphi_{s,0}^{1,N}(x)-{\varphi^{1,\infty}_{s,0}}(x)|\le e^{C\sqrt{\ln(N)}}N^{-2\beta} 
\end{align} 
holds for $N$ large enough.
Note, that this bound is much smaller than necessary.
Therefore let $t\in [0,T]$ be such that $\Delta_N(t)\le N^{-\frac{5}{12}+\sigma}$. It holds for $x\in \mathbb{R}^6$ and $N\in \N\setminus \{1\}$ that
{\allowdisplaybreaks 
\begin{align*}
 & |\varphi_{t,0}^{2,N}(x)-{\varphi^{2,\infty}_{t,0}}(x)|\\
\le &\big|\int_0^t\int_{\mathbb{R}^6}\big(f^N(\varphi^{1,N}_{s,0}(x)-{\varphi^{1,N}_{s,0}}(y))-f^{\infty}({\varphi^{1,\infty}_{s,0}}(x)-{\varphi^{1,\infty}_{s,0}}(y))\big)k_0(y)d^6yds\big|\\
\le &\big|\int_0^t\int_{\mathbb{R}^6}\big(f^N(\varphi^{1,N}_{s,0}(x)-{\varphi^{1,N}_{s,0}}(y))-f^N({\varphi^{1,\infty}_{s,0}}(x)-{\varphi^{1,\infty}_{s,0}}(y))\big)k_0(y)d^6yds\big|\\
 & +\big|\int_0^t\int_{\mathbb{R}^6}\big(f^N(\varphi^{1,\infty}_{s,0}(x)-{\varphi^{1,\infty}_{s,0}}(y))-f^{\infty}({\varphi^{1,\infty}_{s,0}}(x)-{\varphi^{1,\infty}_{s,0}}(y))\big)k_0(y)d^6yds\big|\\
\le & \int_{0}^t2\Delta_N(s)\int_{\mathbb{R}^6} g^N({\varphi^{1,N}_{s,0}}(x)-{\varphi^{1,N}_{s,0}}(y))k_0(y)d^6yds\\
 & +\big|\int_0^t\int_{\mathbb{R}^6}\big(f^N({\varphi^{1,\infty}_{s,0}}(x)-{^1y})-f^{\infty}({\varphi^{1,\infty}_{s,0}}(x)-{^1y})\big)k^{\infty}_s(y)d^6yds\big|\\
\le & C\ln(N)\int_{0}^t\Delta(s) ds+
\big|\int_0^t\int_{\mathbb{R}^6} \frac{^1y}{|^1y|^{3}}\mathds 1_{(0,N^{-\beta}]}(|^1y|) k^{\infty}_s(y+{\varphi}^{\infty}_{s,0}(x))d^6yds\big|\\
&+ \big|\int_0^t\int_{\mathbb{R}^6} {^1y} N^{3\beta}\mathds 1_{(0,N^{-\beta}]}(|^1y|) k^{\infty}_s(y+{\varphi}^{\infty}_{s,0}(x))d^6yds\big|.
\end{align*}}
In the second step we applied the assumption $\Delta_N(t)\le N^{-\beta}$. Remember $g^N(q)$ is bounded by $C\min\big( N^{3\beta},\frac{1}{|q|^{3}}\big)$ for all $q\in \mathbb{R}^3$.
The last two terms are quite similar.
Let us consider the first term and let us use the notation $x=(^1x,{^2x})\in \mathbb{R}^6$. 
Due to the slowly varying mass or charge density, cancellations arise such that this term keeps small enough, i.e.
{\allowdisplaybreaks
\begin{align}
&\big|\int_0^t\int_{\mathbb{R}^6} \frac{^1y}{|^1y|^{3}}\mathds 1_{(0,N^{-\beta}]}(|^1y|) k^{\infty}_s(y+{\varphi}^{\infty}_{s,0}(x))d^6yds\big|\notag \\
= &\big|\int_0^t\int_{\mathbb{R}^6} \frac{^1y}{|^1y|^{3}}\mathds 1_{(0,N^{-\beta}]}(|^1y|) \Big(\big(k^{\infty}_s(y+{\varphi}^{\infty}_{s,0}(x))-k^{\infty}_s((0,{^2y})+{\varphi}^{\infty}_{s,0}(x))\big) \notag \\
&+k^{\infty}_s((0,{^2y})+{\varphi}^{\infty}_{s,0}(x))\Big)d^6yds\big| \notag \\
\le  & \int_0^t\int_{\mathbb{R}^6} \frac{1}{|^1y|^2}\mathds 1_{(0,N^{-\beta}]}(|^1y|)\Big(\big|k^{\infty}_s(y+{\varphi}^{\infty}_{s,0}(x))-k^{\infty}_s((0,{^2y})+{\varphi}^{\infty}_{s,0}(x))\big|\Big)d^6yds. \label{concl.term}
\end{align}}
Note that due to symmetry
\begin{align*}
& \big|\int_0^t\int_{\mathbb{R}^6} \frac{^1y}{|^1y|^{3}}\mathds 1_{(0,N^{-\beta}]}(|^1y|)k^{\infty}_s((0,{^2y})+{\varphi_{s,0}^{\infty}}(x))d^6yds\big|\\
=&\big|\int_0^t\widetilde{k}^{\infty}_s({\varphi_{s,0}^{1,\infty}}(x))\int_{\mathbb{R}^3} \frac{q}{|q|^{3}}\mathds 1_{(0,N^{-\beta}]}(|q|)d^3qds\big|= 0.
\end{align*}
Remember that the initial density fulfills $|\nabla k_0(x)|\le \frac{C}{(1+|x|)^{3+\delta}}$. It follows, that for arbitrary $z\in \mathbb{R}^6$ and $s\in [0,T]$ 
\begin{align}
&\big|k^{\infty}_s(y+z)-k^{\infty}_s((0,{^2y})+z)\big|\mathds 1_{(0,N^{-\beta}]}(|^1y|) \notag \\
=& \big|k_0(\varphi^{\infty}_{0,s}(y+z))-k_0(\varphi^{\infty}_{0,s}((0,{^2y})+z))\big| \mathds 1_{(0,N^{-\beta}]}(|^1y|) \notag  \\
\le & \sup_{z'\in \overline{\varphi^{\infty}_{0,s}(y+z)\varphi^{\infty}_{0,s}((0,{^2y})+z)}}|\nabla k_0(z')| \notag  \\
&\cdot \mathds 1_{(0,N^{-\beta}]}(|^1y|)\Big( \big|\varphi^{\infty}_{0,s}(y+z)-\varphi^{\infty}_{0,s}((0,{^2y})+z)\big|\Big) \notag \\
 \le &   \sup_{z'\in \overline{\varphi^{\infty}_{0,s}(y+z)\varphi^{\infty}_{0,s}((0,{^2y})+z)}}\frac{C}{(1+|z'|)^{3+\delta}} \notag \\
&  \cdot \mathds 1_{(0,N^{-\beta}]}(|^1y|)\Big(C\big|(y+z)-\big((0,{^2y})+z \big) \big|\Big) \notag \\
\le & \sup_{y'\in \mathbb{R}^3:|y'|\le N^{-\beta}}\sup_{z'\in \overline{\varphi^{\infty}_{0,s}((y',{^2y})+z)\varphi^{\infty}_{0,s}((0,{^2y})+z)}}\frac{CN^{-\beta}}{(1+|z'|)^{3+\delta}} \label{t.decay.grad.dens.}
\end{align}
where $\overline{xy}\coloneqq  \{(1-\eta)x+\eta y\in \mathbb{R}^6: \eta\in [0,1]\}$ for $x,y\in \mathbb{R}^6$ and
Lemma \ref{Lemma distance same order} was applied in the second last step.
Note that by choosing a sufficiently large value for $|^2y|$, as it appears in this expression, then all configurations within the set, over which the supremum is taken, exhibit velocities of this magnitude due to the bounded mean-field force. 
Consequently, Term \eqref{t.decay.grad.dens.} diminishes as $|^2y|$ increases, following a decay pattern of $\frac{CN^{-\beta}}{(1+|^2y|)^{3+\delta}}$.
Now we can estimate Term \eqref{concl.term}. 
For arbitrary $z\in \mathbb{R}^6$, in particular $z\coloneqq  {\varphi}^{\infty}_{s,0}(x)$, we get that
 \begin{align*}
&\big|\int_{\mathbb{R}^6} \frac{^1y}{|^1y|^{3}}\mathds 1_{(0,N^{-\beta}]}(|^1y|) k^{\infty}_s(y+z)d^6y\big|\\
\le  &\int_{\mathbb{R}^3} \frac{1}{|^1y|^2}\mathds 1_{(0,N^{-\beta}]}(|^1y|)d^3(^1y) \\
& \cdot \int_{\mathbb{R}^3} \sup_{y'\in \mathbb{R}^3:|y'|\le N^{-\beta}}\sup_{z'\in \overline{\varphi^{\infty}_{0,s}((y',{^2y})+z)\varphi^{\infty}_{0,s}((0,{^2y})+z)}}\frac{CN^{-\beta}}{(1+|z'|)^{3+\delta}}d^3(^2y)\\
\le & CN^{-2\beta}.
\end{align*}
So for any $x\in \mathbb{R}^6$ it follows that
\begin{align}
&\sup_{0\le s \le t}|\varphi_{s,0}^{1,N}(x)-{\varphi^{1,\infty}_{s,0}}(x)|\notag\\
\le & \int_0^t|\varphi_{s,0}^{2,N}(x)-{\varphi^{2,\infty}_{s,0}}(x)| ds\notag \\
\le & C\ln(N)\int_{0}^t\int_0^s\Delta_N(r)drds+CN^{-2\beta}t. \label{last.Gron.}
\end{align}
By means of this inequality, one derives by Gronwall Lemma \ref{Gronwall Lemma} that 
\[
\Delta_N(t) = \sup_{x\in \mathbb{R}^6}\sup_{0\le s \le t}|{\varphi_{s,0}^{1,N}}(x)-{\varphi^{1,\infty}_{s,0}}(x)| \le CN^{-2\beta}te^{\sqrt{C\ln(N)}t}
\]
which shows that the initial assumption $\Delta_N(t)\le N^{-\beta}=N^{-\frac{5}{12}+\sigma}$ stays true for arbitrarily large times \( t \) provided that \( N \in \mathbb{N} \) is large enough. 

Applying the stated bound to the relation
\begin{align*}
    \left|\varphi_{t,0}^{2,N}(x)-{\varphi^{2,\infty}_{t,0}}(x)\right| 
    \le C\ln(N)\int_{0}^t\Delta_N(s)ds+CN^{-2\beta},
\end{align*}
yields the asserted result
\begin{align}
    \sup\limits_{x\in \mathbb{R}^6}\sup_{0\le s \le T}|\varphi_{s,0}^N(x)-{\varphi^{\infty}_{s,0}}(x)|\le e^{C\sqrt{\ln(N)}}N^{-2\beta} \label{dist.mean-field.flow}
\end{align}
for sufficiently large $N$. 
This completes the proof of Theorem \ref{maintheorem}.

\end{document}